\documentclass[11pt]{article}
\usepackage{times}
\usepackage[left=1in, right=1in, top=1in, bottom=1in]{geometry}
\usepackage{amsmath,amsthm,amssymb,amsfonts,paralist}
\usepackage{graphicx}
\usepackage{wrapfig}
\usepackage[small,bf]{caption}
\usepackage{subfig}
\usepackage{paralist}
\usepackage{subfig}
\usepackage[small,compact]{title sec}
\usepackage{listings}
\usepackage{xcolor}
\usepackage{textcomp}
\usepackage{placeins}
\usepackage{setspace}

\newcommand{\algname}{ExaLT}

\newcommand{\R}{\texttt{R}~}

\newtheorem{theorem}{Theorem}

\begin{document}

\title{Accurate Computation of Survival Statistics in Genome-wide Studies}

\author{Fabio Vandin\thanks{Department of Computer Science, Brown University, Providence, RI.}$^{~,\dag}$\\\texttt{\small vandinfa@cs.brown.edu} \and Alexandra Papoutsaki$^{*,\dag}$\\\texttt{\small alexpap@cs.brown.edu} \and Benjamin J. Raphael$^{*\text{,}}$\thanks{Center for Computational Molecular Biology, Brown University, Providence, RI.}$^{~,\ddagger}$\\\texttt{\small braphael@cs.brown.edu} \and Eli Upfal$^{*,}$\thanks{Corresponding authors.}\\\texttt{\small eli@cs.brown.edu}}

\date{}
\maketitle

\begin{abstract}
A key challenge in genomics is to identify genetic variants that distinguish patients with different \emph{survival time} following diagnosis or treatment.
While the log-rank test is widely used 
for this purpose, nearly all implementations of the log-rank test rely on 
an asymptotic approximation that is 
not appropriate in many genomics applications.  This is because: the two populations determined by a genetic variant may have very different sizes; and the evaluation of many possible variants demands highly accurate computation of very small $p$-values.
We demonstrate this problem for 
cancer genomics data where the standard log-rank test leads to many false positive associations between somatic mutations and survival time.
We develop and analyze a novel algorithm, \underline{Exa}ct \underline{L}og-rank \underline{T}est (\algname), that accurately computes the $p$-value of the log-rank statistic under an exact distribution that  is appropriate for any size populations.
We demonstrate the advantages of \algname\ on data from published cancer genomics studies, finding significant differences from the reported $p$-values.  We analyze somatic mutations in six cancer types from The Cancer Genome Atlas (TCGA), finding mutations with known association to survival as well as several novel associations.  In contrast, standard implementations of the log-rank test report dozens-hundreds of likely false positive associations as more significant than these known associations.
\end{abstract}

\onehalfspacing

\section*{Introduction}
\label{sec:intro}
Next-generation DNA sequencing technologies are now enabling the measurement of exomes, genomes, and mRNA expression in many samples.  The next challenge is to interpret these large quantities of DNA and RNA sequence data.  In many human and cancer genomics studies, a major goal is to find associations between an observed phenotype and a particular variable (e.g., a single nucleotide polymorphism (SNP), somatic mutation, or gene expression) from genome-wide measurements of many such variables.  For example, many cancer sequencing studies aim to find somatic mutations that distinguish patients with fast-growing  tumors that require aggressive treatment from patients with better prognosis.  Similarly, many human disease studies aim to find genetic alleles that distinguish patients who respond to particular treatments, i.e. live longer.  In both of these examples one tests the association between a DNA sequence variant and the \emph{survival time}, or length of time that patients live following diagnosis or 
treatment.  

The most widely approach to determine the statistical significance of an observed difference in survival time between two groups is the log-rank test \cite{pmid5910392,citeulike:7445010}.  
An important feature of this test, and related tests in survival analysis \cite{Kalbfleisch:2002fk}, is their handling of \emph{censored} data:  in clinical studies, patients may leave the study prematurely or the study may end before the deaths of all patients.  Thus, a lower bound on the survival time of these patients is known.
Importantly, many studies are designed to test survival differences between two pre-selected populations that differ by one characteristic; e.g. a clinical trial of the effectiveness of a drug.  These populations are selected to be approximately equal in size with a suitable number of patients to achieve appropriate statistical power  (Fig.~\ref{fig:idea}A). In this setting, the null distribution of the (normalized) log-rank statistic is asymptotically (standard) normal; i.e. follows the (standard) normal distribution in the limit of infinite sample size.
 Thus, nearly every available implementation (e.g., the \texttt{LIFETEST} procedure in \texttt{SAS}, and the \texttt{survdiff} package in \R and \texttt{SPlus}) of the log-rank test computes $p$-values from the normal distribution, an approximation that is accurate asymptotically (see Supporting Information).

The design of a genomics study is typically very different from the traditional clinical trials setting.  In a genomics study, high-throughput measurement of many genomics features (e.g. whole-genome sequence or gene expression) in a cohort of patients is performed, and the goal is to \emph{discover} those features that distinguish survival time. Thus, the measured individuals are repeatedly partitioned into two populations determined by a genomic variable (e.g. a SNP) and the log-rank test, or related survival test, is performed (Fig.~\ref{fig:idea}B).  Depending on the variable the sizes of the two populations may be very different: e.g. most somatic mutations identified in cancer sequencing studies, including those in driver genes, are present in $<20$\% of patients \cite{Cancer-Genome-Atlas-Research-Network:2011uq,Cancer-Genome-Atlas-Network:2012uq,Cancer-Genome-Atlas-Network:2012ys,Cancer-Genome-Atlas-Research-Network:2013zr,Garraway:2013uq,Vogelstein:2013fk}.  Unfortunately, in the setting of unbalanced populations, the normal approximation of the log-rank statistic gives poor results.  

While this fact has been noted in the statistics literature~\cite{Latta1981,KellererC1983}, it is not widely known, and indeed the normal approximation to the log-rank test is routinely used to test the association of somatic mutations and survival time (e.g. \cite{Jiao:2011qy,Gaidzik:2011fk} and numerous other publications).
A second issue in genomics setting is that the repeated application of the log-rank test demands the  accurate calculation of very small $p$-values, as the computed $p$-value for a single test must be corrected for the large number of tests; e.g. through a Bonferroni or other multiple-hypothesis correction.
An inaccurate approximation of $p$-values will result in an unacceptable number of false positives/negative associations of genomic features with survival.
These defining characteristics of genomics applications, unbalanced populations and necessity of highly-accurate $p$-values for multiple-hypothesis correction, indicate that standard implementations of the log-rank test are inadequate.

We propose to compute the $p$-value for the log-rank test using an exact distribution determined by the observed number of individuals in each population.  
Perhaps the most famous use of an exact distribution is Fisher's exact test for testing the independence of two categorical variables arranged in a $2 \times 2$ contingency table.  
When the counts in the cells of the table are small, the exact test is preferred to the asymptotic approximation given by the $\chi^2$ test~\cite{1922}.
Exact tests for comparing two survival distributions have received scant attention in the literature.
There are three major difficulties in developing such a test. First, there are multiple observed features that determine the exact distribution including the number of patients in each population and the observed censoring times.
With so many combinations of parameters it is infeasible to pre-compute distribution tables for the test. Thus, we need an efficient algorithm that computes the $p$-value for any given combination of observed parameters.  
Second, we cannot apply a standard Monte-Carlo permutation test to this problem since we are interested in very small $p$-values that are expensive to accurately estimate with such an approach\footnote{The user manual of \texttt{StatXact}~\cite{statxact}, a popular software package for exact inference, reports that, in general, a few hours might be necessary to obtain $p$-values accurate to $10^{-3}$ using  Monte-Carlo permutation tests.}.
Third, there are two possible null distributions for the log-rank test, the conditional and the permutational~\cite{citeulike:7445010,citeulike:3343260,Mantel85,HeimannN1998}.  While both of these distributions are asymptotically normal, the permutational distribution is more appropriate for genomics settings~\cite{citeulike:7445010,Brown84}, as we detail below.  Yet no efficient algorithm is known to compute the $p$-value of the log-rank test under the exact permutational distribution.

We introduce an efficient and mathematically sound algorithm, called \algname~(for \underline{Exa}ct \underline{L}og-rank \underline{T}est), for computing the $p$-value under the exact permutational distribution. 
The run-time of \algname\ is not function of the $p$-value, enabling the accurate calculation of small $p$-values. In contrast to other heuristic approaches, \algname\ returns a conservative estimate of the $p$-value, thus  guaranteeing rigorous control on the number of false discoveries. 
We test \algname\ on 
data from two published cancer studies~\cite{Huang:2009fk,Yan:2009uq}, finding substantial differences between the $p$-values obtained by our exact test and the approximate $p$-values obtained by standard tools in survival analysis. In addition, we run \algname\ on
somatic mutation and survival data from The Cancer Genome Atlas (TCGA) and find a number of mutations with significant association with survival time.  Some of these such as \textit{IDH1} mutations in glioblastoma are widely known; for others such as \textit{BRCA2} and \textit{NCOA3} mutations in ovarian cancer there is some evidence in the literature; while the remaining are genuinely novel.  Most of these are identified only using the exact permutational test of \algname.  In contrast,
the genes reported as highly significant using standard implementations of the log-rank test are not supported by biological evidence; moreover,  these methods report dozens-hundreds of such likely false positive associations as more significant than known genes associated with survival.
These results show that our algorithm is practical, efficient, and avoids a number of false positives, while allowing the identification of genes known to be associated with survival and the discovery of novel, potentially prognostic biomarkers.

\section*{Materials and Methods}

\subsection*{Background: The Log-Rank Test}
We focus here on the two-sample log-rank test of comparing the survival distribution of two groups, $P_0$ and $P_1$.
Let $t_1 < t_2 < \dots <t_k$ be the times of observed, uncensored events. Let $R_j$ be the number of patients \emph{at risk} at time $t_j$, i.e. the number of patients that survived (and were not censored) up to this time, and let $R_{j,1}$ be the number of $P_1$ patients at risk at that time. Let $O_j$ be the number of observed uncensored events in the interval $[t_{j-1},t_j]$, and let $O_{j,1}$ be the number of these events in group $P_1$. If the survival distributions of $P_0$ and $P_1$ are the same, then the expected value  $E[O_{j,1}]=O_j \frac{R_{j,1}}{R_j}$. The log-rank statistic~\cite{pmid5910392,citeulike:7445010} measures the sum of the deviations of $O_{j,1}$  from the expectation,\footnote{In some clinical applications one is more interested in either earlier or later events. In that case the statistic is a weighted sum of the deviations. Our results easily translate to the weighted version of the test.} 
$V= \sum_{j = 1}^{k} \left(O_{j,1} - O_j \frac{R_{j,1}}{R_{j}} \right)$.

Under the null hypothesis of no difference in the survival distributions of the two groups, $E[V]=0$,  and $Pr(|V|\geq |v|)$ is the $p$-value of an observed value $v$.
Two possible null distributions are considered in the literature, the permutational distribution and the conditional distribution (see Figure~S1).

\subsection*{Permutational Log-Rank Test}
In the \emph{permutational} log-rank test~\cite{citeulike:7445010}, the null distribution is obtained by assigning each patient to population $P_0$ or $P_1$ independently of the survival time. Let $n$ be the total number of patients, and $n_1$ the number of patients in group $P_1$. We consider the sample space of all ${n \choose n_1}$ possible selections of survival times and censoring information from the observed data for the $n_1$ patients of group $P_1$.
Each such selection is assigned equal probability ${n \choose n_1}^{-1}$. 

\subsection*{Conditional Log-Rank Test}
In the conditional log-rank test \cite{citeulike:3343260}, the null distribution is defined by conditioning on $O_{j}$ and $R_{j,1}$ for  $j=1,\dots,k$. If at time $t_j$ there are a total of $R_j$ patients at risk, including $R_{j,1}$ patients in $P_1$, then under the assumption of no difference in the survival of $P_0$ and $P_1$ the $O_j$ events at time $t_j$ are split between $P_0$ and $P_1$ according to a hypergeometric distribution with parameters $R_j$, $R_{j,1}$, and $O_{j}$.  

\subsection*{Choice of Null Distribution and Estimation of $p$-values}
We considered the two versions of the log-rank test, the conditional~\cite{citeulike:3343260} and the permutational~\cite{citeulike:7445010}.  The conditional null distribution for the log-rank test is preferred in clinical trials because it does not assume equal distribution of censoring in the two populations. This is important in clinical trials when patients in the two groups are subject to different treatments that may affect their probability of leaving the trial.  However, unequal censoring is not a concern in genome studies, since we do not expect a DNA sequence variant to influence a patient's decision to leave the trial. 
In both  null distributions the prefix sums of the log-rank statistic define a martingale, and by the martingale central limit theorem~\cite{Kalbfleisch:2002fk}, the normalized log-rank statistic 
has an asymptotic $\mathcal{N}(0,1)$ distribution. The normalizing variance is different
in the two null distributions,  but asymptotically the two variances are the same~\cite{Mantel85}, leading to the same $p$-values in the two versions of the test for large balanced populations. Therefore, the distinction between the two versions of the test is largely ignored in practice, where most papers that use the log-rank test or software packages that implement the test do not document the specific test. 
This can be explained, in part, by the widespread use of the log-rank in other scenarios, like clinical trials, where the issues specific to genomic settings do not arise.  The differences between the tests are also rarely discussed in the literature, although there is some discussion \cite{Brown84,Mantel85} on which variance is the appropriate one to use  to compute $p$-values from the asymptotic approximation.

In the case of small and unbalanced populations, the two null distributions yield different $p$-values, and the normal approximation gives poor estimates of both (Fig. S2).  On simulations of cancer data, we found that the $p$-values from the permutational exact test are significantly closer to the empirical $p$-values than the $p$-values obtained from the conditional exact test (Fig.~S3 and Supplemental Text).  Moreover, we prefer the permutational null distribution because it better models the null hypothesis for mutation data and has greater power.

While the exact computation of $p$-values in the conditional null distribution can be computed in polynomial time using a dynamic programming algorithm~\cite{pmid14969496}, no polynomial time algorithm is known for the exact computation of the $p$-value in the permutational null distribution: current implementations are based on a complete enumeration algorithm, making its use impractical for large number of patients\footnote{The \texttt{StatXact} manual recommends using the enumeration algorithm only when the number of samples is at most 20.}.
Several heuristics have been developed for related computations 
including: saddlepoint methods to approximate the mid-$p$-values~\cite{CJS:CJS10002}, methods based on the Fast Fourier Transform (FFT) \cite{Keich:2005fk,Nagarajan:2005uq,Pagano1983}, and branch and bound methods~\cite{Bejerano:2004uq}. 
Such heuristics are shown to be asymptotically correct, converging to the correct $p$-value as the number of samples and computation time grows to infinity.
However, no explicit bounds are known for the accuracy of the computed $p$-value when these heuristics are applied to a fixed sample size and under a bounded computation time.  Given
the systematic error we report below for the
standard implementation of the log-rank test, we argue that such guarantees are essential in this and many similar settings.

\subsection*{\algname}
We developed an algorithm, \underline{Exa}ct \underline{L}og-rank \underline{T}est (\algname), to compute the $p$-value of the log-rank statistic under the exact permutational distribution.
In particular, we designed a fully polynomial time approximation scheme (FPTAS) for computing the $p$-value under the permutational distribution. Our algorithm gives an explicit bound on the error in approximating the true $p$-value, for any given sample size, in polynomial time. Furthermore, the output of our scheme is always a conservative or valid $p$-value estimate.

Since the log-rank statistic depends only on the \emph{order} of the events and not on their actual times, we can w.l.o.g. treat the survival data (including the censored times) as an ordered sequence of events, with no two patients having identical survival times.
Let $n_j=|P_j|$, for $j=0,1$,  be the number of patients in each population and let  $n= n_0 + n_1$ be the total number of patients.
We represent the data by two binary vectors $\mathbf{x} \in \{0,1\}^n$ and $\mathbf{c} \in \{0,1\}^n$, where $x_i=1$ if the $i$th event was in $P_1$ and $x_i=0$ otherwise; $c_i=0$ if the $i$th event was censored and $c_i=1$ otherwise.  Note that $n_1=\sum_{i=1}^n x_i$.
In this notation the log-rank statistic is
\begin{equation}
V=V({\mathbf{x}, \mathbf{c}}) =\sum_{j=1}^n c_j \left(x_j -  \frac{n_1 -\sum_{i=0}^{j-1} x_i}{n-j+1}\right).
\end{equation}

Let $V_t(\mathbf{x}) =\sum_{j=1}^t c_j (x_j -  \frac{n_1 -\sum_{i=0}^{j-1} x_i}{n-j+1})$ be the test statistic $V(\mathbf{x})$ at time $t$.  Note that since $n$, $n_1$, and $\mathbf{c}$ are fixed, the statistic depends only on the value of $\mathbf{x}$.
Assume the observed log-rank statistic has value $v$. The $p$-value of the observation $v$ is the probability $Pr(|V(\mathbf{x})|\geq |v|)$ computed in the probability space in which the $n_1$
events of $P_1$ are uniformly distributed among the $n$ events.

For any $0\leq t \leq n$ and $0\leq r \leq n_1$, let 
$P(t,r,v)=Pr (V_t(\mathbf{x}) \le v \text{ and } \sum_{i=1}^t x_i =r )$
denote the joint probability of $V_t(\mathbf{x}) \le v$ and exactly $r$ events from $P_1$ occur in the first $t$ events.
Let $Q(t,r,v)=Pr(V_t(\mathbf{x}) \geq v \text{ and } \sum_{i=1}^t x_i =r).$
denote the joint probability of $V_t(\mathbf{x}) \ge v$ and exactly $r$ events from $P_1$ occur in the first $t$ events.

At time $0$: $P(0,r,v) =1$ if $r=0$ and $v\geq 0$, otherwise $P(0,r,v) =0$. Similarly  $Q(0,r,v) =1$ if $r=0$ and $v\leq 0$, otherwise $Q(0,r,v) =0$.

 Given the values of $P(t,r,v)$ 
 for all $v$ and $r$, we can compute the values of $P(t+1,r,v)$ 
 using the following relations:
If  $c_{t+1}=1$ then
\[\begin{array}{rcl}
P(t+1,r,v) &=&(1-\frac{n_1-r}{n-t})P(t,r,v+\frac{n_1-r}{n-t})\\
&+&\frac{n_1-(r-1)}{n-t}P(t,r-1,v-(1-\frac{n_1-(r-1)}{n-t})).
\end{array}\]
If $c_{t+1}=0$ then
 $$P(t+1,r,v)=(1-\frac{n_1-r}{n-t})P(t,r,v)+\frac{n_1-(r-1)}{n-t}P(t,r-1,v).$$
Analogous equations hold for $Q(t,r,v)$.

The process defined by these equations guarantees that the $n$ events always include $n_1$ events of $P_1$. 
Thus, the $p$-value of the observation $v$ is given by
$ Pr(|V(\mathbf{x})|\geq |v|)=P(n,n_1,-|v|)+Q(n,n_1,|v|).$
 
For fixed $t$ and $r$, $P(t+1,r,v)$ and $Q(t+1,r,v)$ are step functions. 
For example,  if $c_{t+1}=1$, then as we vary $v$, $P(t+1,r,v)$ changes only at the points in which $P\left(t,r,v+\frac{n_1-r}{n-t}\right)$ or $ P\left(t,r-1,v-\left(1-\frac{n_1-(r-1)}{n-t}\right)\right)$ change values. Thus, we only need to compute the function $P(t+1,r,v)$ at these points. At $t=0$ the function $P(0,r,v)$ assumes up to 2 values.
 If $P(t,r,v)$ assumes $m(t,r)$ values and $P(t,r-1,v)$ assumes $m(t,r-1)$ values, then $P(t+1,r,v)$ assumes up to
$m(t,r)+m(t,r-1)$ values. 

Similar relations hold for $P(t+1,r,v)$ when $c_{t+1}=0$, and for computing $Q(t,r,v)$ in the two cases.
Thus, in $n$ iterations the process computes the exact probabilities $P(n,r,v)$ and $Q(n,r,v)$, but it may have to compute probabilities for an exponential number of different values of $v$ in some iterations.

We construct a polynomial time algorithm by modifying the above procedure to compute the probabilities of only a polynomial number of values in each iteration.
We first observe that since the probability space consists of $n\choose {n_1}$ equal probability events, all non-zero probabilities in our analysis are $\geq n^{-n_1}$.
For $\varepsilon>0$,  fix $\varepsilon_1$ such that $(1-\varepsilon_1)^{-n} = 1 + \varepsilon$. Note that $\epsilon_1=O(\epsilon/n)$.  We discretize  the interval of possible non-zero probabilities $[n^{-n_1}, 1]$, using the values
$(1-\varepsilon_1)^k$, for $k=0,\dots,\ell =\frac{-n_1 \log n}{\log (1-\varepsilon_1)}=
O( \varepsilon^{-1} n n_1 \log n )$.
The approximation algorithm estimates $P(t,r,v)$ with a step function
$\tilde{P}(t,r,v)$ defined by a sequence of $\ell$ points $v^{t}_{k,r}$, $k=0,\dots,\ell$, such that $v^{t}_{k,r}$  is an estimate for the largest $v$ such that
$P(t,r,v)\leq (1-\varepsilon_1)^k$. We prove that if iteration $t$
computes a function $\tilde{P}(t,r,v)(1-\epsilon_1)^t \leq P(t,r,v)\leq \tilde{P}(t,r,v)$, then starting from $\tilde{P}(t,r,v)$ the $t+1$ iteration computes
an estimate $\tilde{P}(t+1,r,v)(1-\epsilon_1)^{t+1} \leq P(t+1,r,v)\leq \tilde{P}(t+1,r,v)$. (Fig.~S4 provides the intuition for how the approximation at time $t+1$ is computed from the approximation at time $t$.) Thus, after $n$ iterations we have an $\epsilon$-approximations for $P(n,n_1,v)$. Similar computations obtain an $\epsilon$-approximation for $Q(n,n_1,v)$.
 The details of the algorithm and analysis are given in the Supporting Information.

We implemented the FPTAS in our software \algname~and evaluated its performance as $n$ and $\varepsilon$ varies, and by comparing its running time with the running time of the exhaustive enumeration algorithm for the permutational test (Fig.~S5). Our implementation of the FPTAS is very efficient, with significant speed-up compared to the exhaustive algorithm. The \texttt{C++} implementation of \algname~ and a \R package to run \algname~are available at \texttt{http://compbio.cs.brown.edu/projects/exactlogrank/}.

\subsection*{Synthetic Data}
We used synthetic data to assess the accuracy of the asymptotic approximations. We generated data as follow: when no censoring was included, we generated the survival times for the patients from an exponential distribution, and the group labeling (mutated or not) were assigned to patients independently of their survival time; when censoring in $f\%$ of the patients was included, we selected $f\%$ of the patients to have censored survival time  independently of their survival time and group. The censoring time was assumed to happen just before the observed survival time.

We used synthetic data to compare  the empirical p-value and the p-values from the exact tests as well. In this case we generated
synthetic data using two related but different procedures. In the first procedure, we mutate a gene $g$ in exactly
a fraction $f$ of all patients. In the second procedure, we mutated a gene $g$ in each patient independently
with probability $f$. The second procedure models the fact that mutations in a gene $g$ are found in each
patient independently with a certain probability. In both cases the survival information is generated from the same distribution for all patients. The survival time comes
from the exponential distribution with expectation equal to $30$, and censoring variable from an exponential
distribution resulting in $30\%$ of censoring.

\subsection*{Mutation and Clinical TCGA Data}
We analyzed somatic mutation  and clinical data, including survival information, from the public TCGA data portal (\texttt{https://tcga-data.nci.nih.gov/tcga/}). In particular we considered single nucleotide variants and small indels for  
colorectal carcinoma (COADREAD), glioblastoma multiforme (GBM), kidney renal clear cell carcinoma (KIRC), lung squamous cell carcinoma (LUSC), ovarian serous adenocarcinoma  (OV), and uterine corpus endometrial carcinoma (UCEC). 
We restricted our analysis to patients for which somatic mutation and survival data were both available.  We only considered genes mutated in $> 1\%$ of patients.
Since genes mutated in the same set of patients would have the same association to survival, they are all equivalent for an automated analysis of association between mutations and survival; we then collapsed them into \emph{metagenes}, recording the genes that appear in a metagene. In our experiments we used \algname~to compute the exact permutational $p$-value whenever the mutation frequency of a gene was $\le 10\%$, and we used the asymptotic approximation for genes with mutations frequency higher than $10\%$, since our simulations shows that the approximation is accurate for the range of parameters considered (see Supporting Information).

\section*{Results}
\subsection*{Accuracy of Asymptotic Approximations} 
We first assessed the accuracy of the asymptotic approximation for the log-rank test on simulated data from a cohort of 500 patients with a gene $g$ mutated in $5\%$ of these patients, a frequency that is not unusual for cancer genes in large-scale sequencing studies \cite{Cancer-Genome-Atlas-Research-Network:2011uq,Cancer-Genome-Atlas-Network:2012uq,Cancer-Genome-Atlas-Network:2012ys,Cancer-Genome-Atlas-Research-Network:2013zr}.  We compared the survival times of the population $\mathcal{P}(g)$ of patients with a mutation in $g$ to the survival of the population $\bar{\mathcal{P}}(g)$ of patients with no mutation in $g$.
We computed $p$-values using \R \texttt{survdiff} on multiple random instances (in order to obtain a distribution for the $p$-value of $g$) in which $\mathcal{P}(g)$ and $\bar{\mathcal{P}}(g)$ have the same survival distribution. Fig.~S2 in Supporting Information shows the difference between the $p$-values computed by the asymptotic approximation are much smaller than expected under the null hypothesis, with the smallest $p$-values showing the largest deviation from the expected uniform distribution.

The inaccuracy of the asymptotic log-rank test results in a large number of false discoveries: for example, considering a randomized version of a cancer mutation dataset 
(Table~S1) in which no mutation is associated with survival (i.e. no true positives), 
the asymptotic log-rank test reports 110 false discoveries (Bonferroni correction) or 291 false discoveries (False Discovery Rate (FDR) correction), with significance level $\alpha =0.05$ (Fig.~1).
We found that the inaccuracy of the asymptotic log-rank test results mostly from the imbalance in the sizes of the two populations, rather than the total number of patients or the number of patients in the smaller population (see Fig.S2a-d and Supporting Text).

\subsection*{Published Cancer Studies}
To demonstrate the applicability of \algname\, we compared the $p$-values from the exact distribution to $p$-values from the asymptotic approximation reported in two recently published cancer genomics studies \cite{Huang:2009fk,Yan:2009uq}.  
Huang et al.~\cite{Huang:2009fk} divides patients into groups defined by the number of risk alleles of five single nucleotide polymorphisms (SNPs), and compares the survival distribution of the resulting populations. 
In one comparison, the survival distribution of 2 patients ($13\%$ of total) with at most 2 risk alleles was compared with the survival distribution of 14 patients with more than 2 risk alleles, and a $p$-value of 0.012 is reported.  Thus, this association is significant at the traditional significance level of $\alpha=0.05$.  However, \algname\ computes an exact $p$-value of $0.17$,  raising doubts about this association.
In another comparison patients at a different disease stage were considered, and the division of the patients into groups as above resulted in comparing the survival distribution of 8 patients ($17\%$ of total) with the survival distribution of 40 patients, and a $p$-value of $6\times 10^{-6}$ is reported.  In contrast, \algname\ computes an exact $p$-value of $2\times 10^{-3}$, a reduction of three orders of magnitude in the significance level.  Additional comparisons are shown in the Supporting Information.

In~\cite{Yan:2009uq}, the survival distribution of 14 glioblastoma patients ($11\%$ of total) with somatic \textit{IDH1} or \textit{IDH2} mutations was compared to the survival distribution of 115 patients with wild-type \textit{IDH1} and \textit{IDH2}. The reported $p$-value from the 
asymptotic approximation is $2\times 10^{-3}$, while the exact permutational $p$-value is $5\times 10^{-4}$, indicating a \emph{stronger} association between somatic mutations in \textit{IDH1} or \textit{IDH2} and (longer) survival than reported. Notably, this same association 
has been reported in three other glioblastoma studies~\cite{Nobusawa:2009fk,Myung:2012zr,Houillier:2010fk}.

\subsection*{TCGA Cancer data}
\label{sec:res_cancer}
We analyzed somatic mutation and survival data from studies of six different cancer types (Table~S1) from The Cancer Genome Atlas (TCGA). For each mutated gene, we compared the $p$-value obtained using the asymptotic approximation as computed by the \R package \texttt{survdiff} 
to the exact $p$-value as computed by \algname (Tables~S2-S4).
Fig.~2 shows the exact  $p$-values and the \R \texttt{survdiff} $p$-values for the glioblastoma multiforme (GBM) dataset and ovarian serous adenocarcinoma (OV) dataset. The $p$-values for the other datasets are shown in Fig.~S6.

For most datasets the asymptotic $p$-values obtained from \R \texttt{survdiff} are very different from the ones obtained with the exact $p$-values obtained by \algname, and the ranking of the genes by $p$-value is very different as well (see Supporting Information). For example, in GBM \emph{none} of the top 25 genes reported by \R are in the list of the top 25 genes reported by the exact permutational test.
Since genomics studies are typically focused on the \emph{discovery} of novel hypotheses that will be further validated, this striking difference in the ranking of genes by the two algorithms is important: a  poor ranking of genes by their association with survival will lead to many false discoveries undergoing additional experimental validation.
While several of genes ranked in the top 10 by \algname\ are known to have mutations associated with survival (Table~1), 
\emph{none} of the top 10 genes reported by \R \texttt{survdiff} (Table~S4) have mutations known to be associated with survival.
\R \texttt{survdiff} ranks dozens-hundreds of presumably false positives associations as more significant than these known genes.  Moreover, \R \texttt{survdiff} reports extremely strong association with survival for many of these higher ranked, but likely false positive, genes; e.g., in uterine corpus endometrial carcinoma (UCEC), 13 genes have $p<10^{-8}$ and an additional 19 genes have $p < 10^{-5}$, but none of these have a known association with survival.

The top 10 genes reported by \algname\ contain several novel associations that are supported by the literature and are not reported using \R \texttt{survdiff}.
In GBM, \algname\ identifies 
\textit{IDH1} ($p \le 7\times 10^{-5}$), \emph{VARS2} ($p \le 8\times 10^{-3}$) and \emph{GALR1} ($p \le 9\times 10^{-3}$), among others. 
As noted above, the association between mutations in \textit{IDH1} and survival has been previously reported in GBM \cite{Yan:2009uq,Nobusawa:2009fk,Myung:2012zr,Houillier:2010fk}.
A germline variant in \emph{VARS2} has been reported to be a prognostic marker, associated with survival, in early breast cancer patients~\cite{Chae:2011ly}. The expression of \emph{GALR1} has been reported to be associated with survival in colorectal cancer~\cite{Stevenson:2012fk}, and its inactivation by methylation has been associated with survival in head and neck cancer~\cite{Misawa:2008uq,Misawa:2013kx}. In OV, \algname\ identifies
\textit{BRCA2} ($p \le 7\times 10^{-3}$) and \textit{NCOA3} ($p \le 3\times 10^{-3}$), and others. 
Germline and somatic mutations in \emph{BRCA2} (and \emph{BRCA1}) have been associated with survival in two ovarian cancer studies~\cite{Bolton:2012bh,Cancer-Genome-Atlas-Research-Network:2011uq}.
A polymorphism in \emph{NCOA3} has been associated with breast cancer~\cite{Burwinkel:2005vn}, and its amplification has been associated with survival in ER-positive tumors~\cite{Burandt:2013ys}. 

Thus, the exact test implemented by \algname\ appears to have higher sensitivity and specificity in detecting mutations associated with survival on the sizes of cohorts analyzed in TCGA.
Finally, we note that the exact conditional test obtains results similar to \R \texttt{survdiff}, confirming that the 
the exact permutational test implemented by \algname\ is a more appropriate exact test for genomics studies. (See Supporting Information.)

\section*{Discussion}
In this work we focus on the problem of performing survival analysis in a genomics setting, where the populations being compared are not defined in advance, but rather are determined by a genomic measurement.  The two distinguishing features of such studies are that the populations are typically unbalanced and  that many survival tests are performed for different measurements, requiring highly accurate $p$-values for multiple hypothesis testing corrections.
We show empirically that the asymptotic approximations used in available implementations of the log-rank test produce anti-conservative estimates of the true $p$-values when applied to unbalanced populations, resulting in a large number of false discoveries. This is not purely a phenomenon of small population size:  the approximation remains inaccurate even for a large number of samples (e.g., 100) in the small population.
This inaccuracy makes asymptotic approximations unsuitable for cancer genomic studies, where the vast majority of the genes are mutated in a small proportion of all samples \cite{Cancer-Genome-Atlas-Network:2012ys,Cancer-Genome-Atlas-Research-Network:2011uq,Cancer-Genome-Atlas-Network:2012uq} and also for genome-wide association studies (GWAS) where rare variants may be responsible for a difference in drug response or other phenotype.

The problem with the log-rank test for unbalanced populations has previously been reported~\cite{Latta1981,KellererC1983}, but the implications for genomics studies have not received attention.
Note that the issue of unbalanced populations is further exacerbated by any further subdivision of the data: e.g. by considering mutations in specific locations or protein domains; by considering the impact of mutations on a specific therapeutic regimen;  by testing the association of mutations with survival in a particular subtype of cancer; by grouping into more than two populations; or by correcting for additional covariates such as age, stage, grade, etc. All of these situations occur in genomics studies.

We considered the two versions of the log-rank test, the conditional~\cite{citeulike:3343260} and the permutational~\cite{citeulike:7445010}, and we found that 
the exact permutational distribution is more accurate in genomics settings.
We introduce \algname, the first efficient algorithm to compute highly accurate $p$-values for the exact permutational distribution.
We implemented and tested our algorithm on data from two published cancer studies, showing that the exact permutational $p$-values are significantly different from the $p$-values obtained using the asymptotic approximations. We also ran \algname\ on somatic mutation and survival data from six cancer types from The Cancer Genome Atlas (TCGA), showing that our algorithm is practical, efficient, and  allows the identification of genes known to be associated with survival in these cancer types as well as novel associations.

While our focus here was the log-rank test,  our results are relevant to more general survival statistics.
First, in some survival analysis applications, samples are given different weights; our algorithm can be easily adapted to a number of these different weighting schemes. Second, an alternative approach in survival analysis is to perform linear regression using the Cox Proportional-Hazards model~\cite{Kalbfleisch:2002fk}. Testing the significance of the resulting coefficients in the regression is typically done using a test that is equivalent to the log-rank test, and thus our results are relevant for this approach as well.  See Supporting Information.

The challenges of extending multivariate regression models to the multiple-hypothesis setting of genome-wide measurements is not straightforward.
Direct application of such a multivariate Cox regression will often not give reasonable results as: there are a limited number of samples and a large number of genomic variants; and many variants are rare and not associated with survival.  
Witten and Tibshirani (2010) \cite{pmid19654171} recently noted these difficulties for gene expression data stating that: ``\textit{While there are a great number of methods in the literature for identification of significant genes in a microarray experiment with a two-class outcome \dots the topic of identification of significant genes with a survival outcome is still relatively unexplored.}''  We propose that exact tests such as the one provided here will be useful building blocks for more advanced models of survival analysis in the genomics setting.

\section*{Acknowledgments}
This work is supported by NSF grants IIS-1016648 and IIS-1247581. BJR is supported by a Career Award at the Scientific Interface from the Burroughs Wellcome Fund, an Alfred P. Sloan Research Fellowship, and an NSF CAREER Award (CCF-1053753). EU is a member of the scientific advisory board and a consultant at Nabsys.
The results published here are in whole or part based upon data generated by The Cancer Genome Atlas pilot project established by the NCI and NHGRI. Information about TCGA and the investigators and institutions who constitute the TCGA research network can be found at  \texttt{http://cancergenome.nih.gov/}.

\newpage
\singlespacing

\begin{figure}
\centering
\includegraphics[width=\textwidth]{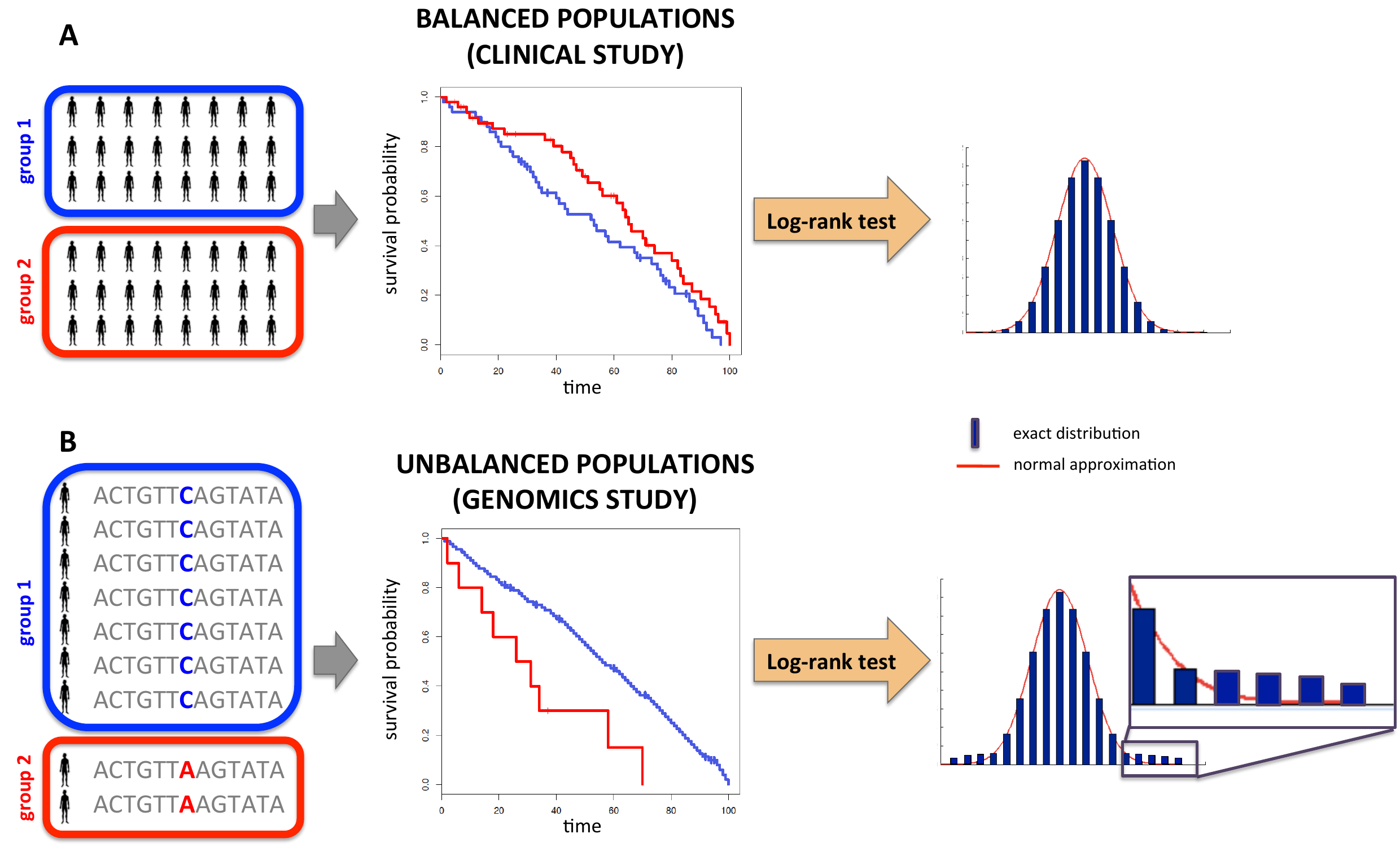}
\caption{Difference between survival analysis in a clinical setting with balanced populations and genomics setting, with unbalanced populations. (A) In a typical clinical study, two pre-selected groups of similar size are compared. 
Because the groups are balanced and each has a suitable number of patients, the asymptotic approximation (normal distribution) used in common implementations of the log-rank test gives an accurate approximation of the exact distribution, resulting in accurate $p$-values. (B) In a genomics study, the two groups are defined by a genetic variant.  In many cases, the sizes of the groups are unbalanced, with one group being much larger than the other. 
In this situation, the asymptotic distribution does not accurately approximate the exact distribution of the log-rank statistic, and the resulting $p$-values computed from the tail of the distribution (see inset) are inaccurate.}
\label{fig:idea}
\end{figure}

\begin{figure}
\centering
\includegraphics[width=0.5\textwidth]{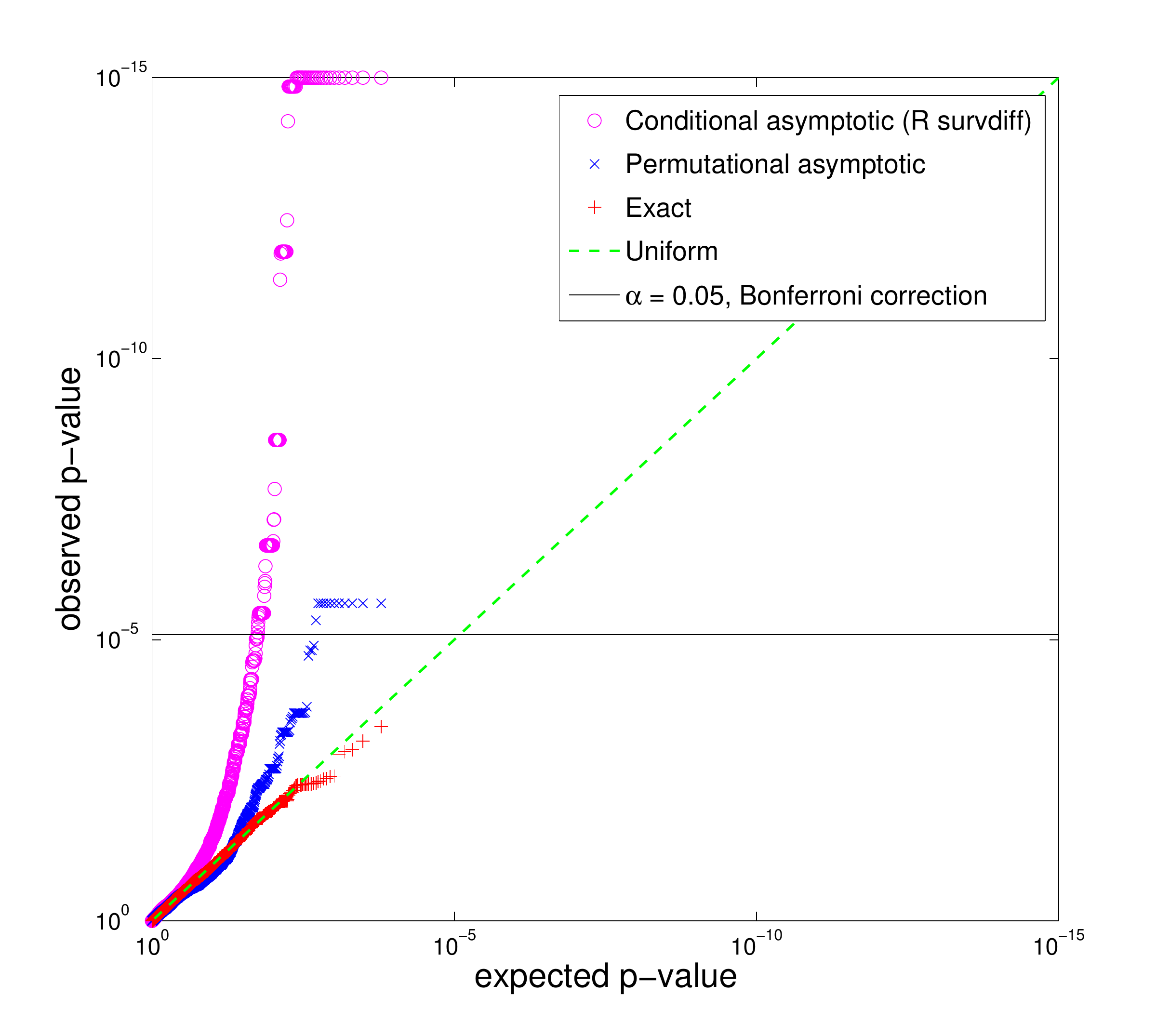}
\caption{Differences between observed and expected $p$-values from different forms of the log-rank test on a randomized cancer dataset consisting of somatic mutations in 6184 genes. The $p$-values for the genes should be distributed uniformly (green line), since there is no association between mutations and survival in this random data.  Asymptotic approximations of the log-rank statistic (purple and blue) yield $p$-values that deviate significantly from the uniform distribution, incorrectly reporting many genes whose mutations are significantly associated with survival.  
In particular, the asymptotic log-rank test in \R reports 110 genes with significant association, using a Bonferroni corrected $p$-value $<0.05$ (black line), or 291 genes with significant association using a less conservative FDR $=0.05$.  In contrast, the exact test makes no false discoveries.}
\label{fig:randomGBM}
\end{figure}

\begin{figure}
                \includegraphics[width=0.4\textwidth]{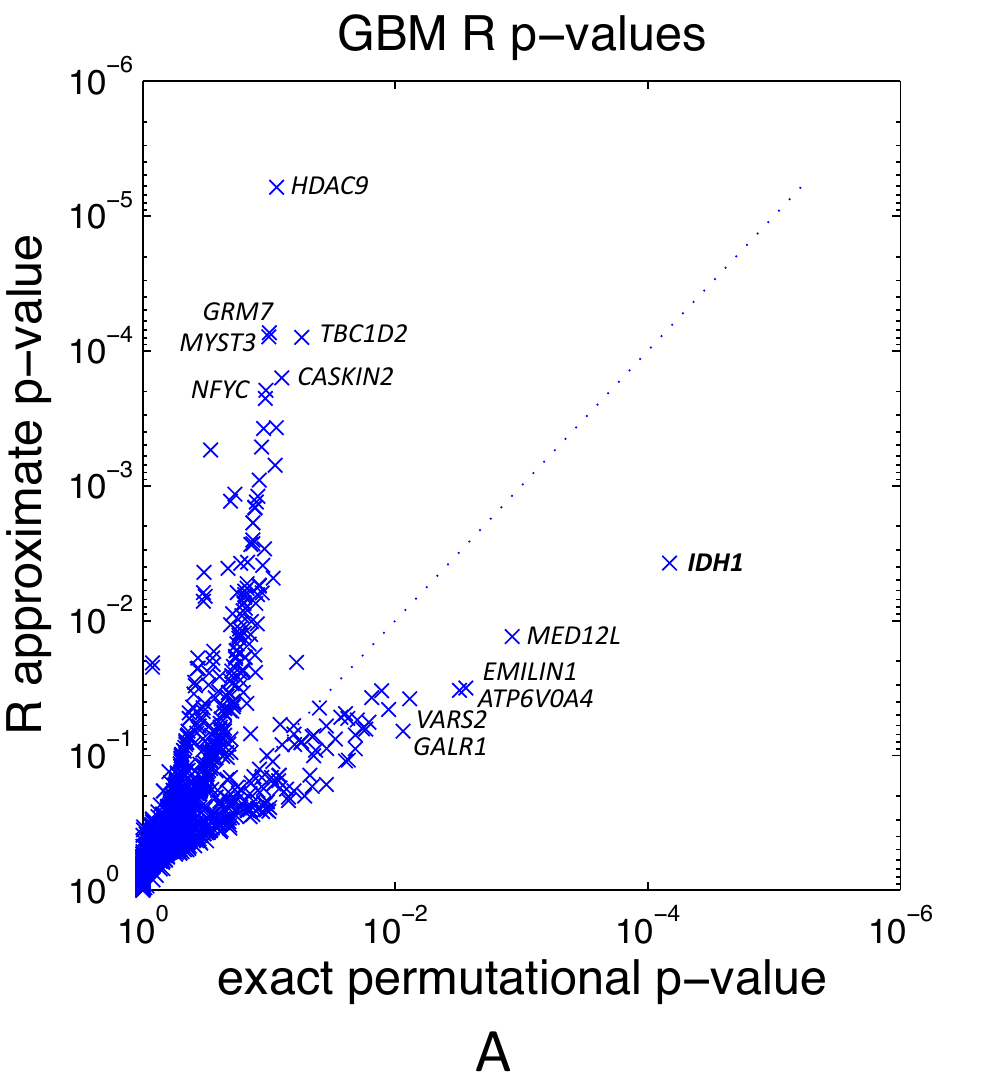}
                \label{fig:R_GBM}
                \quad
                \includegraphics[width=0.4\textwidth]{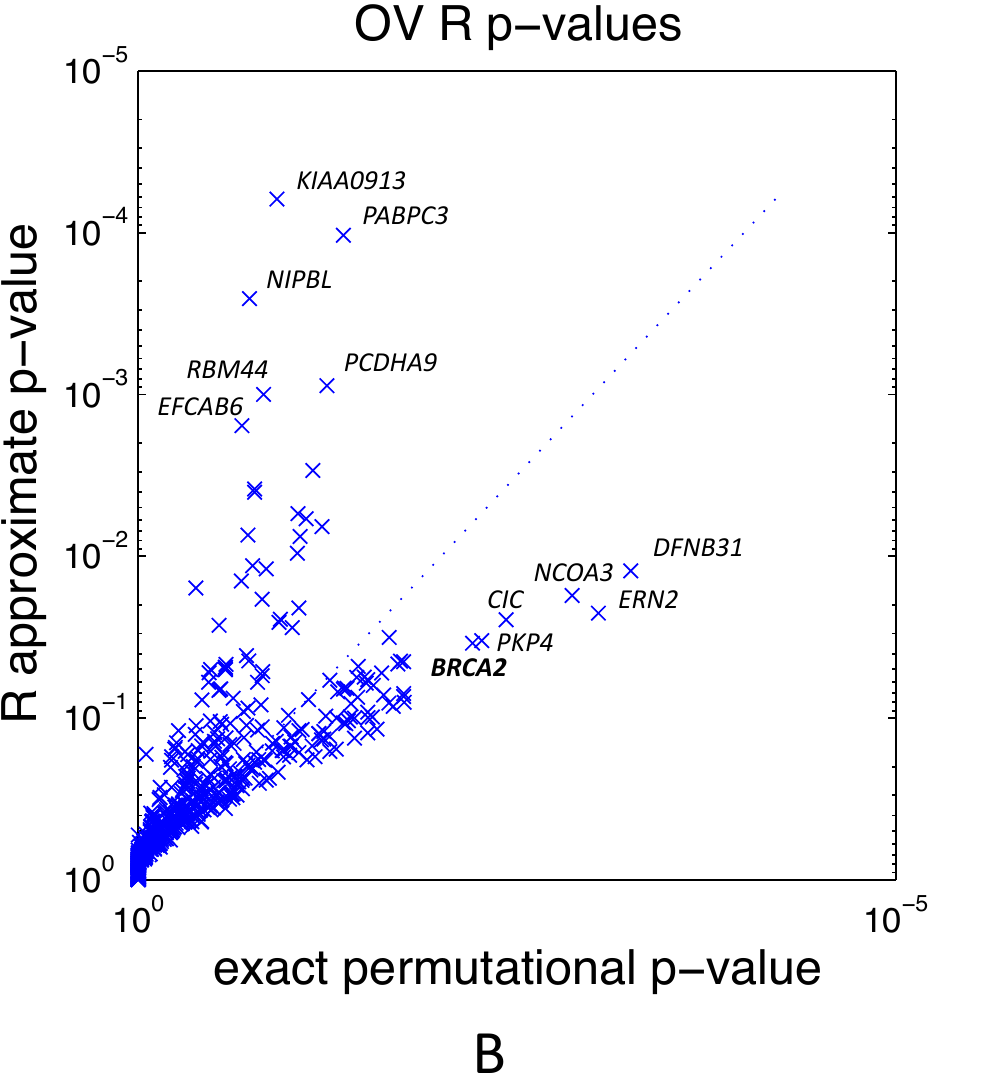}
                \label{fig:R_OV}
        \caption{Comparison of the $p$-values of association between somatic mutations and survival time for TCGA glioblastoma (GBM) and ovarian (OV) datasets.  Each data point represents a gene. (A) Comparison of the \R \texttt{survdiff} $p$-values and the exact permutational $p$-values for the GBM dataset. (B) Comparison of the \R \texttt{survdiff} $p$-values and the exact permutational $p$-values for the OV dataset.}
        \label{fig:cancer_data}
\end{figure}

\FloatBarrier

\begin{table*}
\label{tab:results}
\caption{Genes with mutations previously reported to associated with survival that were identified by \algname. For each gene we show: the cancer type in which the gene is identified; the rank and $p$-value using the exact permutational test computed by \algname\ and the approximate test computed by the\texttt{survdiff} package in \R; the number of samples with a mutation in the gene; and references supporting the association of mutations with survival.}
\begin{footnotesize}
\begin{tabular}{cccccccccc}
& & & \multicolumn{2}{c}{ \algname\ (exact permutational) } & &  \multicolumn{2}{c}{\R \texttt{survdiff} (asymptotic)} & & \\
\cline{4-5} \cline{7-8}
gene & cancer type & & rank & $p$ & & rank & $p$ & num. mut. & refs\\
\hline
\emph{KRAS} & COADREAD & & 4 & 1.3e-02&  & 82 & 1.4e-02& 85 & \cite{Heinemann:2009kx}\\
\emph{FBXW7} & COADREAD & & 10 & 5.6e-02&   & 142 & 5.8e-02& 21 &\cite{Cancer-Genome-Atlas-Network:2012ys}\\
\emph{IDH1} & GBM & & 1 & 6.7e-05&   & 27 & 3.8e-03& 15 & \cite{Nobusawa:2009fk,Myung:2012zr}\\
\emph{BAP1} & KIRC & & 9 & 9.7e-03&  & 98 & 1.2e-02& 27 & \cite{Hakimi:2012ve}\\
\emph{BRCA2} & OV & & 6 & 6.2e-03&   & 32 & 3.5e-02& 9 & \cite{Cancer-Genome-Atlas-Research-Network:2011uq,Bolton:2012bh}\\
\emph{ARID1A} & UCEC & & 5 & 9.7e-03&  & 219 & 1.2e-02& 73 & \cite{Wang:2011dq}\\
\hline
\end{tabular}
\end{footnotesize}
\end{table*}

\FloatBarrier

\newpage

\newcommand{\problem}[2]{\vspace{5pt}\noindent{\bf #1:}~\emph{#2}\vspace{5pt}}
\renewcommand{\thefigure}{S\arabic{figure}}
 \renewcommand{\thetable}{S\arabic{table}}

{\Large Supplementary Information: Accurate Computation of Survival Statistics in Genome-wide Studies}

\vspace{1cm}

\section*{Implementations of the log-rank test}

Common statistical packages provide implementations\footnote{This information was derived directly from the software manual and/or the publication cited in the manual.} of the log-rank test for the following distributions:
\begin{itemize}
\item asymptotic conditional: SAS (\texttt{LIFETEST}), R and S-Plus (\texttt{survdiff}), SPSS, GraphPad Prism.
\item asymptotic permutational and exact permutational: StatXact.
\end{itemize}

\section*{Background} 
\label{model}
\subsection*{Model}
Suppose a set $\mathcal{G}$ of genes was sequenced in a collection $\mathcal{P}$ of patients, all of whom have the same disease.
Each sequenced gene\footnote{One may also consider mutations at different levels of resolution; e.g. partitioning patients according to mutations in individual nucleotides or protein domains.} $g \in \mathcal G$ partitions the set of patients into two subsets: the $\mathcal{P}(g)$, with patients with a mutation in $g$, and the $\bar{\mathcal{P}}(g)$, with patients with no mutation in $g$.  The goal is to identify genes whose mutational status is highly correlated with the survival time, in the sense that the survival distribution of patients in $\mathcal{P}(g)$ is different from the survival distribution of patients in $\bar{\mathcal{P}}(g)$.
A key challenge in survival analysis is dealing with \emph{censored} patients whose exact survival time is unknown.  Censoring occurs for a variety of reasons, but the most common is that the study only lasts for a finite amount of time, and some fraction of patients remain alive at the conclusion of the study.  In addition, 
during the course of the study patients may leave the study for a variety of reasons, that are unrelated to their treatment or disease state. 
The censored survival time is the last time the patient was observed in the study, which is a lower bound 
for the patient's survival time\footnote{In some references patients who survived the study are called \emph{right censored} and patients who withdrew from the study are called \emph{randomly} censored.}.  Survival analysis assumes that censoring is \emph{non informative},  i.e. the event that a patient is censored is independent of the patient's survival beyond the censoring time. 
The log-rank test~\cite{pmid5910392} (or family of tests) is the most commonly used non-parametric test for comparing the survival distribution of two or more populations with data subject to censoring. The advantage of this test is that it includes the censored data in its statistic, rather then removing it from the data.  Since a large fraction of patients may be censored (e.g., up to $94\%$ in the data below), it is not desirable to remove this ``missing data" from consideration. In the section below,  we describe two different versions of the log-rank test, the conditional log-rank and the permutational log-rank test.  

\newpage
\subsection*{Basic survival analysis: the log-rank test}
\label{sec:estimate}
\begin{wrapfigure}{l}{0.32\textwidth}
\centering
  \includegraphics[width=0.34\textwidth]{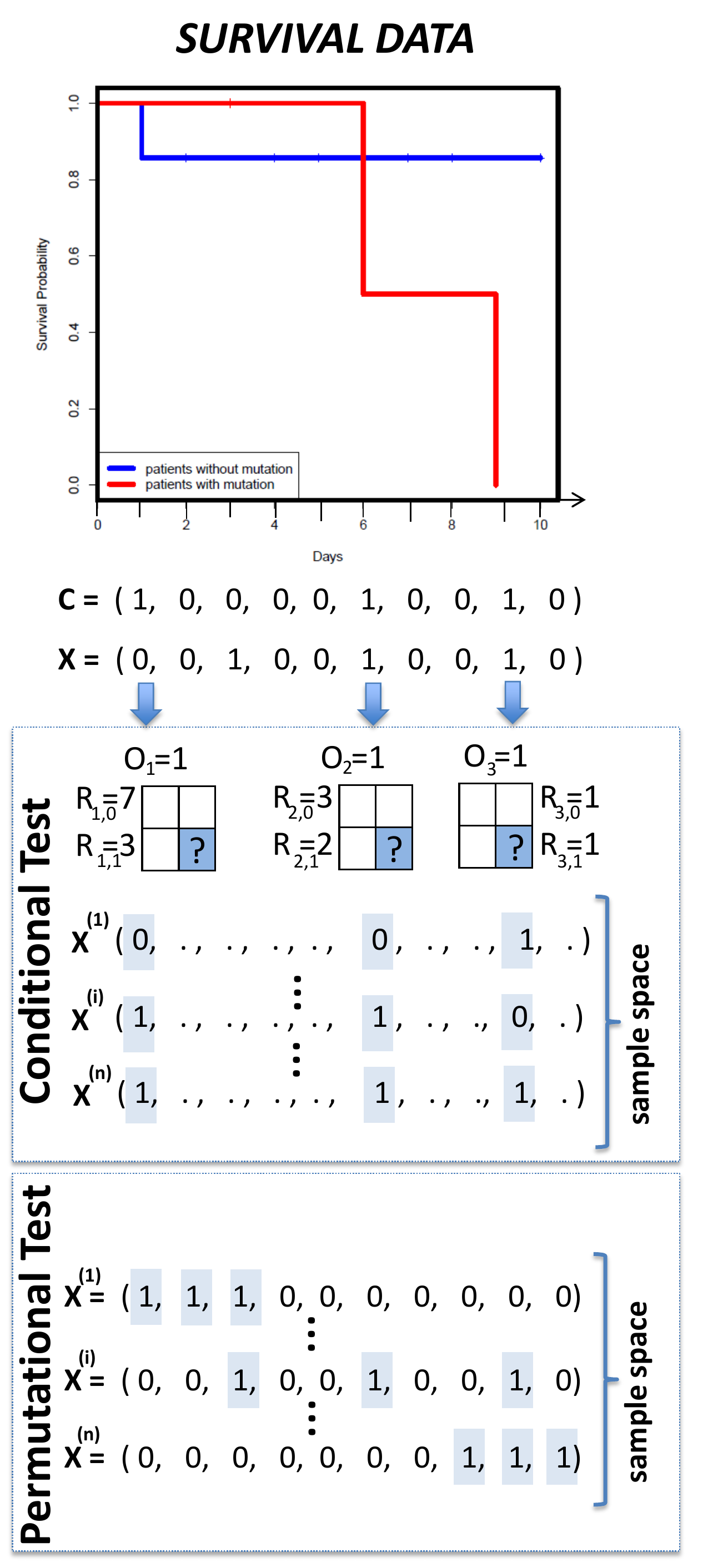}
\caption{(Top) The log-rank test compares the Kaplan-Meier curves of the two groups. 
(Middle) Survival data is represented by sorting patients by increasing survival. $\mathbf{x}$ represents group labels for patients, and $\mathbf{c}$ represents censoring information ($c_i=0$ if event at time $t_i$ is censored, $c_i=1$ otherwise). 
(Bottom) The conditional test is defined by a series of independent contingency tables, conditioning on the patients at risk at each non-censored time.
The permutational test considers all ${n \choose n_1}$ possible locations of the $n_1$ patients with label 1 in the vector $\mathbf{x}$.}
\label{fig:summary}
\end{wrapfigure}
We focus on the two-samples log-rank test of comparing the survival distribution of two groups, $P_0$ and $P_1$.
Let $t_1 < t_2 < \dots <t_k$ be the times of observed not censored events. Let $R_j$ be the number of patients \emph{at risk} at time $t_j$, i.e. the number of patients that survived (and were not censored) up to this time, and let $R_{j,1}$ be the number of $P_1$ patients at risk at that time. Let $O_j$ be the number of observed, not censored events in the interval $[t_{j-1},t_j]$, and let $O_{j,1}$ be the number of these events in group $P_1$. If the survival distributions of $P_0$ and $P_1$ are the same then in expectation $E[O_{j,1}]=O_j \frac{R_{j,1}}{R_j}$.  The log-rank statistic~\cite{pmid5910392,citeulike:7445010} measures the sum of the deviations of $O_{j,1}$ from this equal distribution\footnote{In some clinical applications one is more interested in either earlier or later events. In that case, the statistic is a weighted sum of the deviations. Our results easily translate to the weighted version of the test.} expectation,
\begin{equation}
V= \sum_{j = 1}^{k} \left(O_{j,1} - O_j \frac{R_{j,1}}{R_{j}} \right).
\end{equation}

Since the log-rank statistic depends only on the \emph{order} of the events and not on their actual times, we can w.l.o.g. treat the survival data (including censored times) as an ordered sequence of events, with no two patients having identical survival times.
Let $n_i=|P_i|$ be the number of patients in each set and let  $n= n_0 + n_1$ be the total number of patients.
We represent the data with two binary vectors $\mathbf{x} \in \{0,1\}^n$ and $\mathbf{c} \in \{0,1\}^n$, where $x_i=1$ if the $i$th event was in $P_1$ and $x_i=0$ otherwise; $c_i=0$ if the $i$th event was censored and $c_i=1$ otherwise.  Note that $n_1=\sum_{i=1}^n x_i$.
In this notation the log-rank statistic is
\begin{equation}
V=V({\mathbf{x}, \mathbf{c}}) =\sum_{j=1}^n c_j \left(x_j -  \frac{n_1 -\sum_{i=0}^{j-1} x_i}{n-j+1}\right).
\end{equation}

Clearly, the further $V$ is from zero, the more likely it is the case that the two survival distributions are different. To quantify this intuition, we define the null hypothesis of no difference in the survival distributions of the two groups, and then compute the distribution of the test statistic $V$ under the null hypothesis. Two possible null distributions are considered in the literature, defining two versions of the log-rank test (see Figure~\ref{fig:summary}).

\paragraph{Conditional log-rank test~\cite{citeulike:3343260}.}
In this version, the null distribution is defined by conditioning on $O_{j}$ and $R_{j,1}$ for  $j=1,\dots,k$. If at time $t_j$ there are a total of $R_j$ patients at risk, including $R_{j,1}$ patients in $P_1$, then under the assumption of no difference in the distributions of $P_0$ and $P_1$, we expect the $O_j$ events at that time to be split between $P_0$ and $P_1$ according to a hypergeometric distribution with parameters $R_j$, $R_{j,1}$, and $O_{j}$.  

Under this null distribution the expectation of the log-rank statistic is $0$, and its variance is $\sigma_h^2 = \sum_{j=1}^{k} {O_j\frac{R_{1,j}}{R_j}  \left(1-\frac{R_{1,j}}{R_j}\right)\frac{R_j - O_j}{R_j - 1}  }$ (Mantel-Haenszel variance~\cite{citeulike:3343260}.
Note that this test does not assume equal distribution of censoring in the two groups.  This property is important in clinical trials when patients in the two groups are subject to different treatments that may affect their probability of leaving the trial. In the case of cancer mutation data, under the null hypothesis it is unlikely that the presence of a mutation changes the probability that a patient leaves the trial, and thus we do not face this difficulty.  However, a major disadvantage of this test is that the number of events in $P_1$, generated by this distribution, is a random variable that is equal to $n_1$ only on average. 
 In the case of an unbalanced population, where $n_1$ is small, it can significantly affect the computed $p$-value.

\paragraph{Permutational log-rank test~\cite{citeulike:7445010}.}

In this version, we observe that under the null hypothesis the distribution of the group labels, $x_i$'s, is independent of the survival information. Therefore, we consider the sample space of all ${n \choose n_1}$ possible locations of the $n_1$ patients of group $P_1$ in the vector $\mathbf{x}$,
 and each possibility is assigned equal probability ${n \choose n_1}^{-1}$. For this reason, the resulting distribution is usually called \emph{permutational} distribution of the log-rank statistic~\cite{citeulike:7445010}.
Under this null hypothesis the expectation of the log-rank statistic is $0$, and the variance~\cite{Brown84} is $\sigma_{p}^2 = \frac{n_1 n_2}{n(n-1)} \left(k - \sum_{i=1}^{k} \frac{1}{R_i} \right)$.
Note that in this distribution the number of patients in $P_1$ is exactly $n_1$. The validity of this log-rank test depends on the probability of censoring being equal in the two groups. 
As discussed above this assumption holds in our application. 

\subsection*{Estimating the $p$-value}
Under both null distributions above, the expectation $E[V]=0$.  Given an observed value $v$, its $p$-value is $Pr(|V|\geq |v|)$. In the two null distributions the prefix sums of the log-rank statistic define a martingale, and therefore, by the martingale central limit theorem~\cite{Kalbfleisch:2002fk}, the normalized log-rank statistic $V/\sigma$, where $\sigma$ is either $\sigma_h$ or $\sigma_p$, has an asymptotic $\mathcal{N}(0,1)$ distribution, which gives an easy method for computing the $p$-value. Furthermore, asymptotically the two variances $\sigma^2_h$ and $\sigma^2_p$ are the same~\cite{Mantel85}, thus for large balanced populations the two versions of the test give the same results.  Therefore, the distinction between the two versions of the test is mostly ignored in the literature, although there is some discussion of which variance is the appropriate to use \cite{Brown84,Mantel85}.  

The situation is drastically different in the setting of genome-wide cancer survival analysis.  As was reported in~\cite{Latta1981,KellererC1983} and we show in the next section, the Normal approximation gives a poor estimate for the $p$-value in the range of population sizes inherent in the genome-wide association studies. Thus, we need an efficient algorithm for computing a correct estimate of the $p$-values that does not depend on the Normal approximation. Furthermore, we also report that in this range of parameters, the $p$-value of the log-rank statistic
is very sensitive to the choice of null distribution: since the conditional distribution matches the problem parameters only in expectation, we prefer the permutational null distribution that matches exactly the problem's parameters.

\section*{Accuracy of Asymptotic Approximations}

We applied the log-rank test based on asymptotic approximations to randomly generated survival and mutation data. We focused on the case of unbalanced populations.
We compared the $p$-values obtained from the asymptotic approximations with the uniform distribution that is expected under the null hypothesis. We use $n$ to denote the total number of samples, and $n_1$ the number of samples in  the small population. Fig.~\ref{fig:100_not_enough} shows that even when the number of patients in the small population is large ($n_1=100$), when the imbalance between populations increases, the accuracy of the asymptotic approximation decreases. Fig.~\ref{fig:100_200_500_1000_5_0_permutational} shows that for a fixed ratio $n_1/n$, the asymptotic approximation improves when the total populations size increases. Fig.~\ref{fig:100_nocensoring_permutational} shows that for a fixed $n$, the asymptotic approximation improves when the imbalance decreases. In addition to the normal approximation, Fig.~\ref{fig:random_500} includes the $\chi^2$ approximation, and shows the results considering $10^5$ data points with $n=500$ total samples, $n_1=5\% n$ samples with a mutation in the gene, and same 
survival distribution for all patients. In particular, the survival time comes from an exponential distribution with the same expectation (equal to 30), and censoring variable from an exponential distribution resulting in $30\%$ of censoring. These results show that with $n=100$, $n_1$ must be $> 20\% n$ for the asymptotic permutational approximation to be accurate, while with $n=500$, $n_1$ must be $\ge 5\%$ for the asymptotic permutational approximation to be accurate.

\begin{figure}[htbp]
\centering
		\subfloat[][]{
		 \includegraphics[width=0.4\textwidth]{{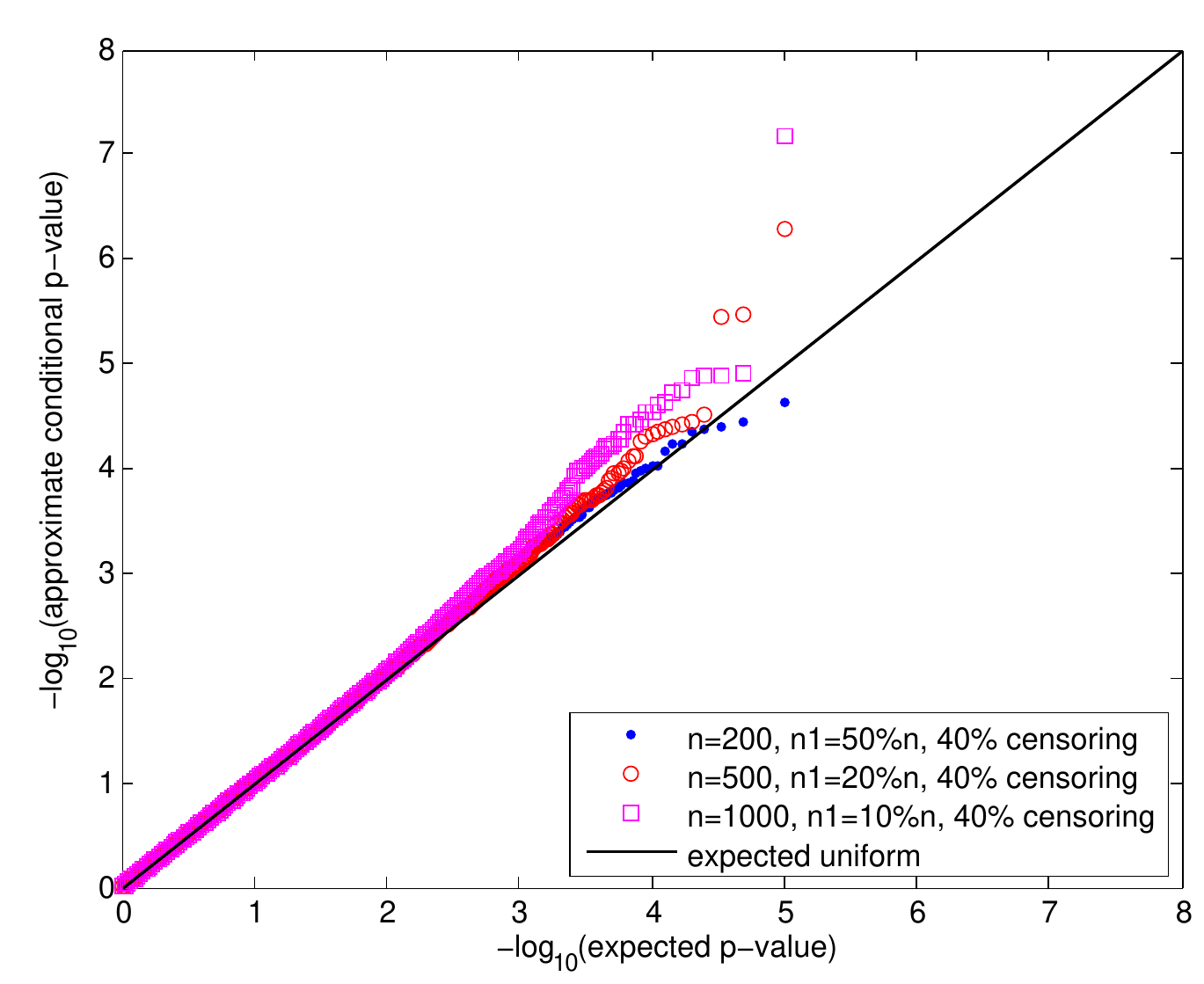}}
		\label{fig:100_not_enough} 
                }
                \quad 
                 \subfloat[][]{
             \includegraphics[width=0.4\textwidth]{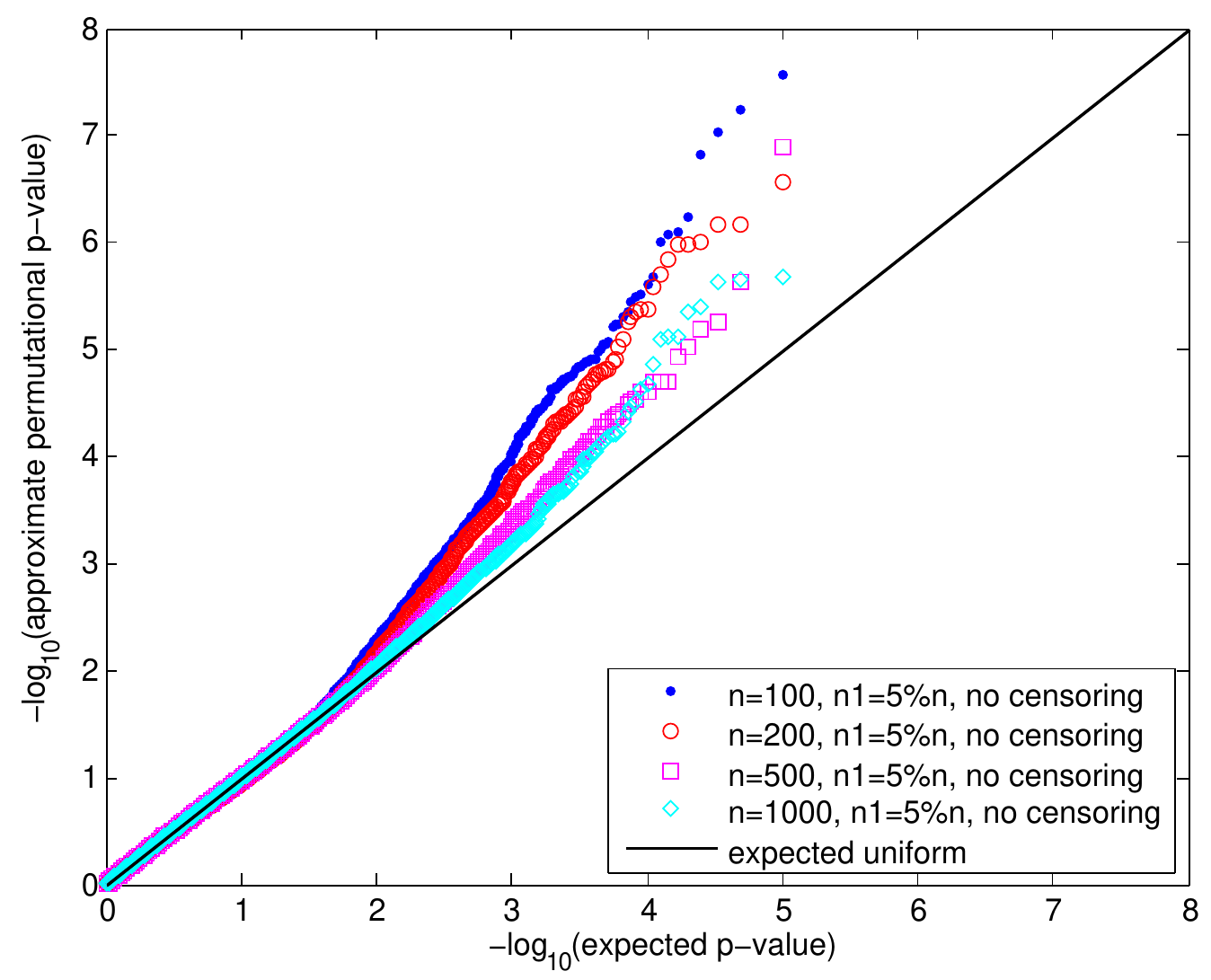}
              \label{fig:100_200_500_1000_5_0_permutational}
                }

		 \subfloat[][]{
		 \includegraphics[width=0.4\textwidth]{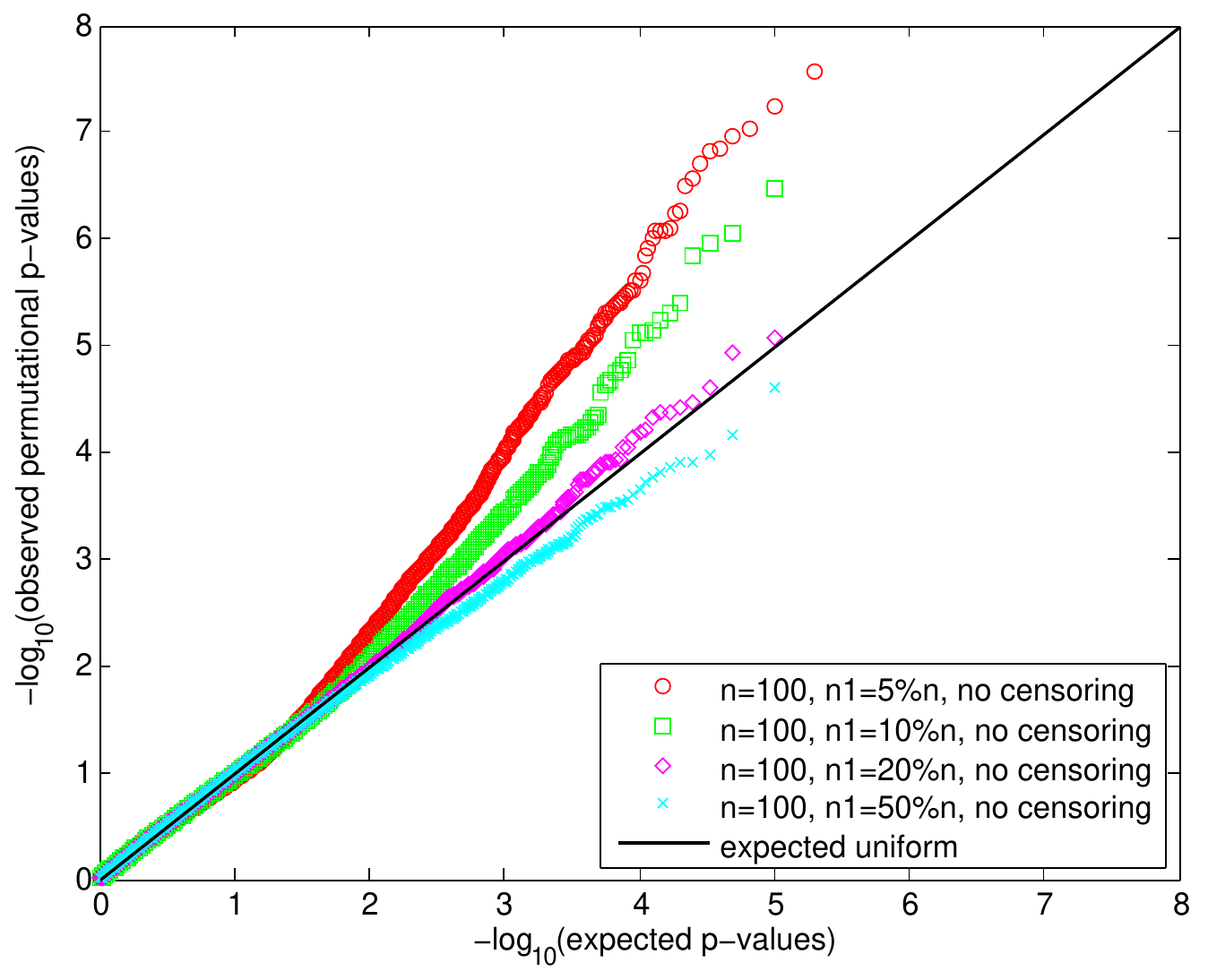}
		\label{fig:100_nocensoring_permutational} 
                }
                \quad
		\subfloat[][]{
             	\includegraphics[width=0.4\textwidth]{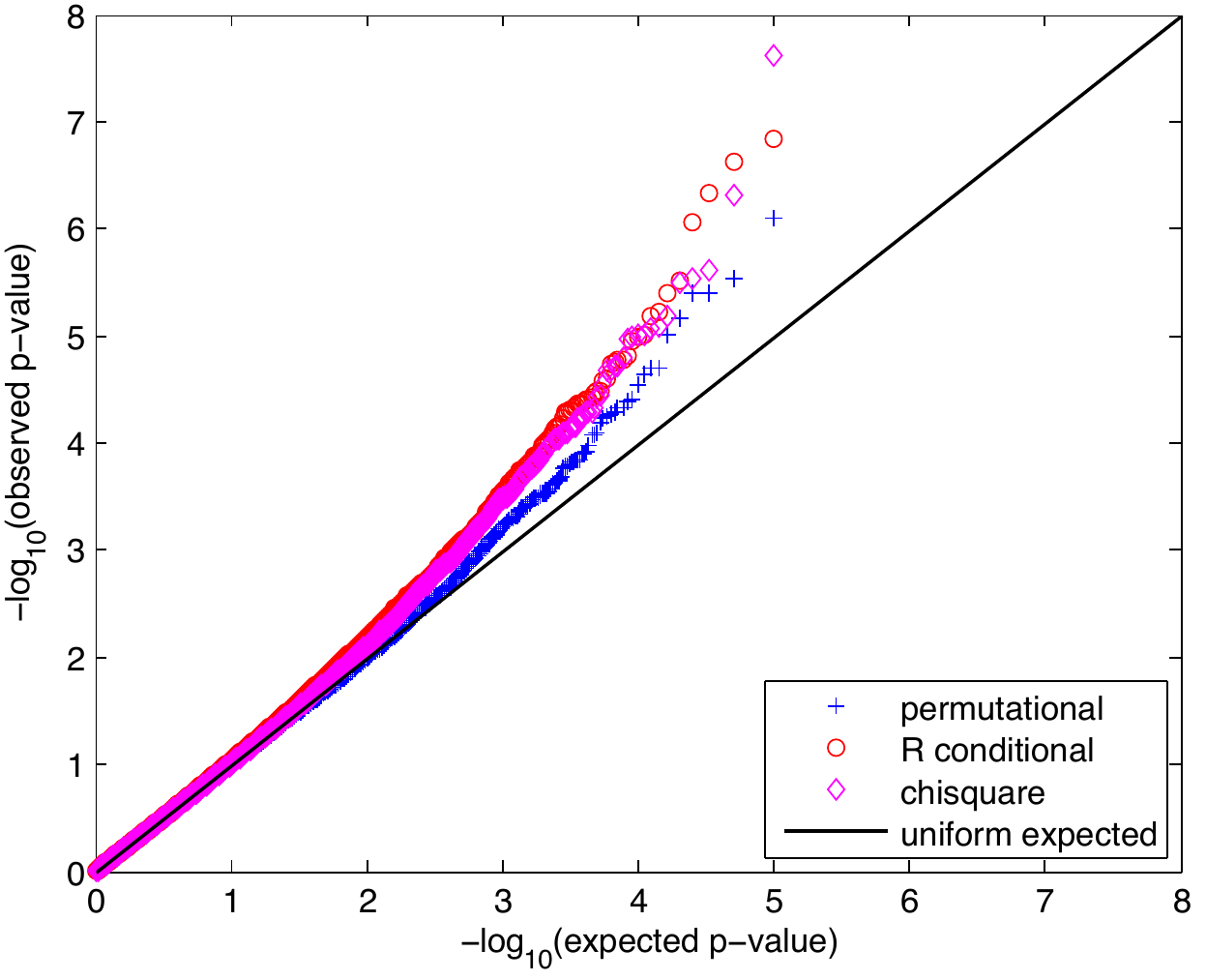}
              \label{fig:random_500}
                }

\caption{Comparison of the $p$-values from asymptotic approximations and the uniform distribution. (a) Distribution of $p$-values obtained using the conditional approximation, and distribution of $p$-values for the uniform distribution. Generated considering $ 10^5$ instances with $n_1=100$ samples in the small population, different number $n$ of samples in total, and same survival distribution for all patients ($\approx 40\%$ censoring). (b) Distribution of $p$-values obtained using the permutational approximation, and distribution of $p$-values for the uniform distribution. Generated considering $10^5$ data points with $n_1=5\% n$ samples in the small population, $n$ total samples, and no censoring. ({c}) Distribution of $p$-values obtained using the permutational approximation, and distribution of $p$-values for the uniform distribution. Generated considering $10^5$ data points with $n=100$ total samples, different values of $n_1$, and no censoring. (d) Distribution of $p$-values obtained using 
different approximations, and distribution of $p$-values for the uniform distribution. Generated considering $ 10^5$ instances with $n=500$ total samples, $n_1=5\%n$ samples with a mutations in the gene, and same survival distribution for all patients ($\approx 30\%$ censoring).}\label{fig:approx_limits}
\end{figure}

\section*{Comparison of Exact Tests on Synthetic Data}

\subsection*{Comparison of Exact Distributions}

We find that the $p$-values from the permutational exact test are significantly closer ($p< 10^{-3}$) to the empirical $p$-values than the $p$-values obtained from the conditional exact test (Fig.~S3).

We compared the accuracy of exact $p$-values for the permutational and conditional distributions in our setting of unbalanced small populations using synthetic data. We generate synthetic data using two related but different procedures.  In the first procedure, we mutate a gene $g$ in exactly a fraction $f$ of all patients.  In the second procedure, we mutated a gene $g$ in each patient independently with probability $f$. The second procedure models the fact that mutations in a gene $g$ are found in each patient independently with a certain probability (that depends on the background mutation rate, the length of the gene, etc.).  Thus, when repeating a study on a cohort of patients of the same size only the expected number of patients in which $g$ is mutated is the same, and the observed number may vary.  In both cases the survival information is generated from the same distribution for all patients.
 The survival time comes from the exponential distribution with expectation equal to 30, and censoring variable from an exponential distribution resulting in $30\%$ of censoring.  In Fig.~\ref{fig:cmp_empirical_exact} we compare the $p$-values computed from the exact permutational test and the exact conditional test with the empirical $p$-values (obtained repeating the experiment 10000 times) for the first distribution, while in Fig.~\ref{fig:cmp_empirical_expect} we compare the $p$-values computed from the exact permutational test and the exact conditional test with the empirical $p$-values (obtained repeating the experiment 10000 times) for the second distribution. In both cases the $p$-values (restricted to $p$-values $\le 0.01$) from the exact permutational distribution have a significantly ($p<10^{-3}$) higher R coefficient than the $p$-values from the exact conditional distribution when compared to the empirical $p$-values (considering the $-log_{10}$ $p$-values in order to compute the R coefficient). Therefore the $p$-values from the permutational exact test are significantly closer to the empirical $p$-values than the $p$-values obtained from the conditional exact test. 

\begin{figure}[htbp]
\centering
                \subfloat[][]{
             \includegraphics[width=0.4\textwidth]{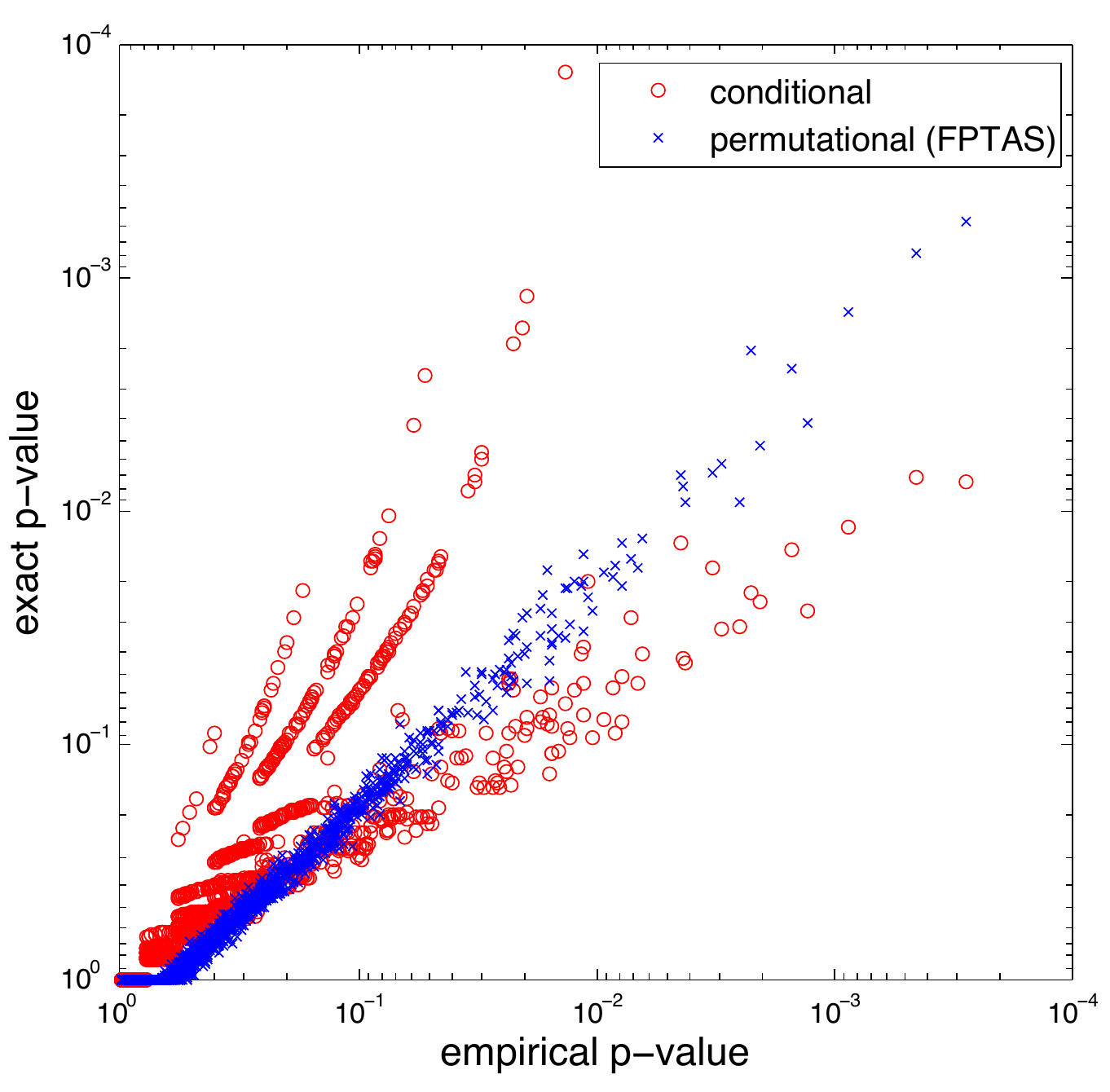}
              \label{fig:cmp_empirical_exact}
                }
                \quad
		 \subfloat[][]{
		 \includegraphics[width=0.4375\textwidth]{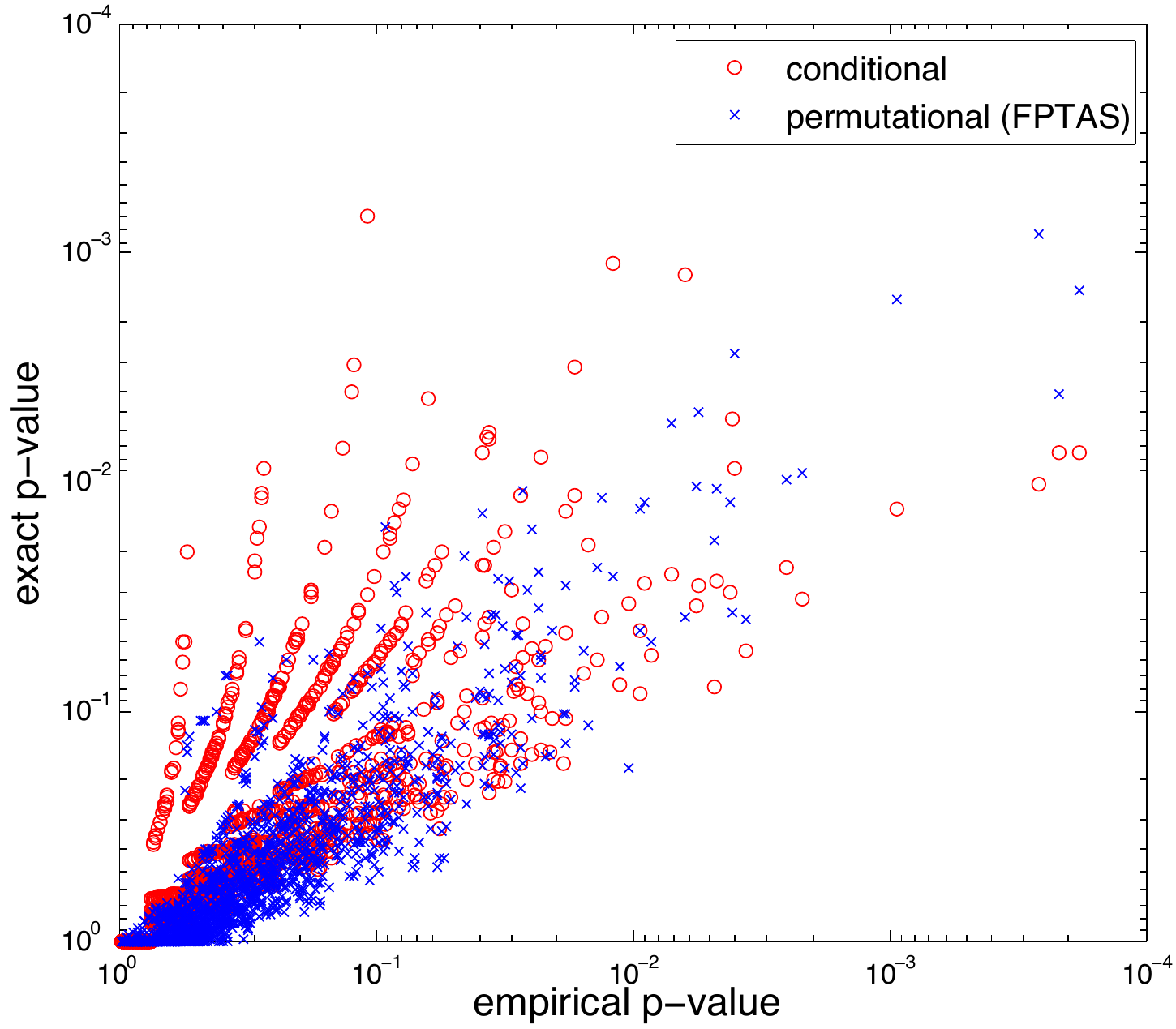}
		\label{fig:cmp_empirical_expect} 
                }
\caption{Comparison of the $p$-values from the exact tests and the empirical $p$-values for two different null distributions. The R coefficients comparing the $-log_{10}$ exact $p$-values to the $-log_{10}$ empirical $p$-values are the following: in Fig.~(a), permutational $= 0.96$, conditional $= 0.88$; in Fig~(b), permutational $= 0.72$, conditional $= 0.43$. For both distributions the difference between R coefficients is significant $(p<10^{-3})$. (a) Comparison of exact conditional $p$-values, exact permutational $p$-values, and empirical $p$-values for $n=100, n_1=5\%n$, and $30\%$ censoring. Each point represents an instance of survival data. (b) Comparison of exact conditional $p$-values, exact permutational $p$-values, and empirical $p$-values for $n=100$, expectation $n_1$=$5\%n$, and $30\%$ censoring. Each point represents an instance of survival data.}
\label{fig:cmp_empirical}
\end{figure}

\section*{Algorithms}

As shown by the results in previous sections,  to carry out an effective genome-wide survival analysis for cancer somatic mutation we need an accurate estimate of the log-rank statistic $p$-values in the permutational null distribution. 

While the exact $p$-value in the conditional test can be computed in quadratic time~\cite{MethaPG85,pmid14969496} no polynomial time algorithm is known for the problem of computing the exact $p$-value for the permutational test. Abd-Elfattah and Butler~\cite{CJS:CJS10002} use saddlepoint methods to determine the mid-$p$-values for the permutational distribution.
Heuristic methods may be derived from solutions to related problems.
In particular the method of Pagano and Tritchler \cite{Pagano1983}, based on the Fast Fourier Transform (FFT), may be adapted to compute some approximation of the exact $p$-value in polynomial time, but no guarantee on the accuracy of the approximation is provided by their method. Branch and bound methods (in the spirit of the method proposed by Bejerano et al.~\cite{Bejerano:2004uq}) may be used to compute the exact $p$-value, but may require exponential time in the worst case.
Note that since $p$-values can be really small, we do not want to use an MCMC approach, that requires to sample a number of random permutations at least proportional to $c^{-1}$ in order to obtain an estimate for a $p$-value equal to $c$.

In the permutational distribution $n$ and $n_1$ are fixed, and thus computing the $p$-value is equivalent to solving the following counting problem.

\problem{More Extreme Assignments Counting Problem}{
Given $n, n_1 \in \mathbb{N},$ with $ n_1\le n$, $v \in \mathbb{R}$ and $\mathbf{c} \in \{0,1\}^n$ determine the number of vectors $ \mathbf{x} \in \{0,1\}^n$ that satisfy:  $\sum_{i=1}^{n}x_i = n_1$ and $|V(\mathbf{x},\mathbf{c})| \ge v$.} 

Dividing the number of vectors $\mathbf{x}$ by ${n\choose n_1}$, which defines the sample space size, gives the $p$-value of $v$.
Based on the similarity between this problem and Knapsack Counting Problem~\cite{CambridgeJournals:1772044}, we conjecture that the problem may also be $\#P$-complete. 

\subsection*{FPTAS for the permutational distribution}
We provide a Fully Polynomial Time Approximation Scheme (FPTAS) to estimate the $p$-value from the permutational distribution,
i.e. an algorithm
that given $n,n_1,\mathbf{c}$, and $v$, for any $\varepsilon>0$ computes an $\varepsilon$-approximation of $Pr(|V(\mathbf{x})|\geq v)$ in time that is polynomial in $n$ and $\varepsilon^{-1}$.
The FPTAS is derived from a pair of recurrence relations that compute the exact probability, but may not terminate in polynomial time.  We then modify the process to obtain a  fully polynomial time approximation scheme.

\paragraph{Exact computation.}
Let $V_t(\mathbf{x}) =\sum_{j=1}^t c_j \left(x_j -  \frac{n_1 -\sum_{i=0}^{j-1} x_i}{n-j+1}\right)$ be the test statistic $V(\mathbf{x})$ at time $t$.  Note that since $n,n_1$, and $\mathbf{c}$ are fixed, the statistic depends only on the values of $\mathbf{x}$.
Assume the observed log-rank statistic has value $v$. The $p$-value of the observation $v$ is the probability $Pr(|V(\mathbf{x})|\geq |v|)$ computed in the probability space in which the $n_1$
events of $P_1$ are uniformly distributed among the $n$ events.
For any $0\leq t \leq n$ and $0\leq r \leq n_1$, let $P(t,r,v)$ denote the joint probability $V_t(\mathbf{x}) \le v$ and exactly $r$ events of $P_1$ in the first $t$ events,
 \begin{equation*}
 P(t,r,v)=\Pr\left(V_t(\mathbf{x}) \le v \text{ AND } \sum_{i=1}^t x_i =r\right).
 \end{equation*}
 Similarly, let $Q(t,r,v)$ denote the joint probability of $V_t(\mathbf{x}) \ge v$ and exactly $r$ events of $P_1$ in the first $t$ steps,
 \begin{equation*}
 Q(t,r,v)=\Pr\left(V_t(\mathbf{x}) \geq v \text{ AND } \sum_{i=1}^t x_i =r\right).
 \end{equation*}

 At time $0$,
 \[ P(0,r,v) = \left \{\begin{array}{ll}
 1 & \mbox{if}~r=0~\mbox{and}~v\geq 0 \\
 0 & \mbox{otherwise,} \end{array} \right.
 ~\mbox{and}~~~~
 Q(0,r,v) = \left \{\begin{array}{ll}
 1 & \mbox{if}~r=0~\mbox{and}~v\leq 0 \\
 0 & \mbox{otherwise.} \end{array} \right.  \]

 Given the values of $P(t,r,v)$ and $Q(t,r,v)$ for all $v$ and $r$, we can compute the values of $P(t+1,r,v)$ and $Q(t,r,v)$ using the following relations:

If  $c_{t+1}=1$ then
 $$P(t+1,r,v) = (1-\frac{n_1-r}{n-t})P(t,r,v+\frac{n_1-r}{n-t})+\frac{n_1-(r-1)}{n-t}P(t,r-1,v-(1-\frac{n_1-(r-1)}{n-t})),~\mbox{and}$$
 $$Q(t+1,r,v) = (1-\frac{n_1-r}{n-t})Q(t,r,v+\frac{n_1-r}{n-t})+\frac{n_1-(r-1)}{n-t}Q(t,r-1,v-(1-\frac{n_1-(r-1)}{n-t})).$$

If $c_{t+1}=0$ then
 $$P(t+1,r,v)=(1-\frac{n_1-r}{n-t})P(t,r,v)+\frac{n_1-(r-1)}{n-t}P(t,r-1,v), ~\mbox{and}$$
 $$Q(t+1,r,v)=(1-\frac{n_1-r}{n-t})Q(t,r,v)+\frac{n_1-(r-1)}{n-t}Q(t,r-1,v).$$

The process defined by these equation guarantees that the $n$ events include $n_1$ events of $P_1$. Thus, for $r\not= n_1$, $P(n,r,v)=0$ and $Q(n,r,v)=0$, and the $p$-value is given by\footnote{In the exact computation $Q(n,n_1,V)=1-P(n,n_1,V)$, but in the approximate algorithm below we need to compute each of the probability functions separately.}
\begin{equation*}
 Pr(|V(\mathbf{x})|\geq |v|)=P(n,n_1,-|v|)+Q(n,n_1,|v|).
\end{equation*}
 The functions $P(t+1,r,v)$ and $Q(t+1,r,v)$ are step functions. To compute the function $P(t+1,r,v)$ when $c_{t+1}=1$, we note that as we vary $v$ the value of $P(t+1,r,v)$ changes only at the points in which $P(t,r,v+\frac{n_1-r}{n-t})$ or $ P(t,r-1,v-(1-\frac{n_1-(r}{n-t}))$ change values. Thus, the function needs to be computed only at these points.

 At $t=0$ the function $P(0,r,v)$ assumes up to 2 values.
 If $P(t,r,v)$ assumes $m(t,r)$ values and $P(t,r-1,v)$ assumes $m(t,r-1)$ values, then $P(t+1,r,v)$ assumes up to
$m(t,r)+m(t,r-1)$ values. Similar relation hold for $P(t+1,r,v)$ when $c_{t+1}=0$, and for computing $Q(t,r,v)$ in the two cases.
Thus, in $n$ iterations the process computes the exact probabilities $P(n,r,v)$ and $Q(n,r,v)$, but it may have to compute probabilities for an exponential number of different values of $v$ in some iterations.

\paragraph{Approximation Algorithm.}
\begin{wrapfigure}{l}{0.6\textwidth}
\centering
  \includegraphics[width=0.6\textwidth]{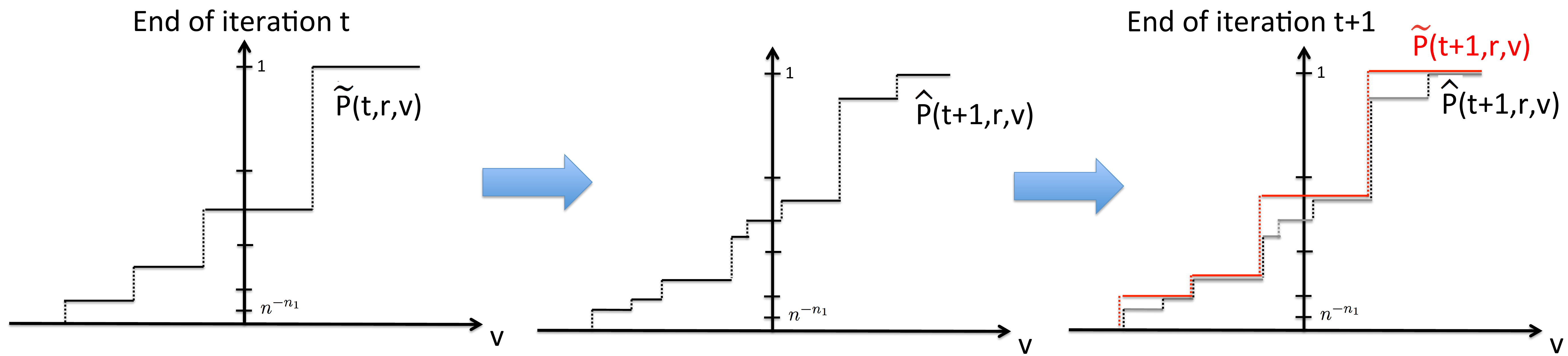}
\caption{Starting from the approximation $\tilde{P}(t,r,v)$ at time $t$ that uses $\ell$ values of $v$ to approximate $P(t,r,v)$, in the $t+1$ iteration the FPTAS computes an approximation $\hat{P}(t+1,r,v)$ for $P(t+1,r,v)$ that uses up to $2\ell$ values of $v$; then the approximation $\tilde{P}(t+1,r,v)$ is built starting from $\hat{P}(t+1,r,v)$ by appropriately reducing the number of values of $v$ considered, while maintaining guarantees on the approximation.}
\label{fig:FPTAS_appendix}
\end{wrapfigure}

We first note that since the probability space consists of $n\choose {n_1}$ equal probability events, all non-zero probabilities in our analysis are $\geq n^{-n_1}$.
For $0<\varepsilon<1$,  fix $\varepsilon_1$ such that $(1-\varepsilon_1)^{-n} = 1 + \varepsilon$. Note that $\epsilon_1=O(\epsilon/n)$.  We discretize  the interval of possible non-zero probabilities $[n^{-n_1}, 1]$, using the values
$(1-\varepsilon_1)^k$, for $k=0,\dots,\ell =\frac{-n_1 \log n}{\log (1-\varepsilon_1)}=
O( \varepsilon^{-1} n n_1 \log n )$.
The approximation algorithm computes estimates for $P(t,r,v)$ and $Q(t,r,v)$ in two separate processes.

\paragraph{Estimating $P(t,r,v)$.}

Let $\tilde{P}(t,r,v)$ be a step function defined by a sequence of points $v^{t}_{k,r}$, $k=0,\dots,\ell$. The value of the function in the interval $[v^{t}_{k+1,r},v^{t}_{k,r}]$ is $(1-\varepsilon_1)^k$,
and for $v>v^t_{0,r}$ the value of the function is 1. Consecutive points in the sequence may be the same ($v^{t}_{k+1,r}=v^{t}_{k,r}$), in that case the value of $\tilde{P}(t,r.v)$ is $(1-\varepsilon_1)^{k_v}$, where $k_v=\arg\max_k [v\leq v^{t}_{k,r}]$. (Not that the sequence $v^{t}_{k,r}$, $k=0,\dots,\ell$ is non-increasing in $k$, since larger $k$ corresponds to smaller probability.)

For $t=0$, we define $\tilde{P}(0,r,v)$  by the set of points $v^0_{k,0}= 0$ and $v^0_{k,r}=\infty$ for $r>0$, $k=0,\dots,\ell$.
These functions satisfy $\tilde{P}(0,r,v)=P(0,r,v)$ for all $r$ and $v$.

Assume that iteration $t+1$ starts with a set of functions $\tilde{P}(t,r,v)$, for $r=0,\dots, n_1$ such that for all $r$ and $v$
$$(1-\varepsilon_1)^t \tilde{P}(t,r,v)\leq P(t,r,v)\leq \tilde{P}(t,r,v).$$
We show that iteration $t+1$ computes functions $\tilde{P}(t+1,r,v)$ with the same approximation properties. (Figure~\ref{fig:FPTAS_appendix} shows how the approximation at time t + 1 is computed from the approximation at time t.) 

To compute an estimate for the functions $P(t+1,r,v)$, $r=0,\dots,n_1$,
we use the relations given in the exact computation, estimating $P(t,r,v)$ by $\tilde{P}(t,r,v)$.
In the case $c_{t+1}=1$ we use
$$\hat{P}(t+1,r,v) = (1-\frac{n_1-r}{n-t})\tilde{P}(t,r,v+\frac{n_1-r}{n-t})+\frac{n_1-r}{n-t}\tilde{P}(t,r-1,v-(1-\frac{n_1-r}{n-t})),$$
and compute (for each $r$) the function at the $2\ell$ points corresponding to change in values in the functions $\tilde{P}(t,r,v)$ and $\tilde{P}(t,r-1,v)$:
$$v^t_{k,r}=v+\frac{n_1-r}{n-t}~\mbox{ and}~ v^t_{k,r-1}=v-(1-\frac{n_1-r}{n-t}),~\mbox{ for}~ k=1,\dots,\ell.$$
In the case $c_{t+1}=0$ we use
 $$\hat{P}(t+1,r,v)=(1-\frac{n_1-r}{n-t})\tilde{P}(t,r,v)+\frac{n_1-r}{n-t}\tilde{P}(t,r-1,v),$$
 and compute the function (for each $r$) in the $2\ell$ points $v^t_{k,r}$ and $v^t_{k,r-1}$.

Let $v_1\leq v_2\leq \dots \leq v_{2\ell}$ be the $2\ell$ points for which the value of $\hat{P}(t+1,r,v)$ was computed.
We extend the function $\hat{P}(t+1,r,v)$ to a step function over all values of $v$, such that
$\hat{P}(t+1,r,v)=\hat{P}(t+1,r,v_j)$, where $j$ is the largest index such that $v\geq v_j$.

Since we computed the function $\hat{P}(t+1,r,v)$ in all points in which $\tilde{P}(t,r,v)$ and $\tilde{P}(t,r-1,v)$ change values, and by the assumptions on the values of $\tilde{P}(t,r,v)$,
for all $r$ and $v$ we have
$$\hat{P}(t+1,r,v)(1-\varepsilon_1 )^{t+1} \leq P(t+1,r,v) \leq \hat{P}(t+1,r,v).$$

We now approximate the function $\hat{P}(t+1,r,v)$ by a function $\tilde{P}(t+1,r,v)$ that is defined by the sequence of only $\ell$ values:
$$v^{j}_{k,r}=\arg \max_v[\tilde{P}(j,r,v)\leq (1-\varepsilon_1)^k ], ~k=0,\dots,\ell.$$
Consider a value $v$ such that $v^{t+1}_{k+1,r} \leq v <v^{t+1}_{k,r}$. We have:
\begin{equation}
P(t+1,r,v) \leq \hat{P}(t+1,r,v)\leq \hat{P}(t+1,r,v^{t+1}_{k,r})\leq \tilde{P}(t+1,r,v^{t+1}_{k,r}),
\end{equation}
and
\begin{equation}
P(t+1,r,v)\geq P(t+1,r,v_{k+1,r}^{t+1}) \geq \hat{P}(t+1,r,v^{t+1}_{k,r})(1-\varepsilon_1)^{t+1} \ge \tilde{P}(t+1,r,v^{t+1}_{k,r})(1-\varepsilon_1)^{t+1}.
\end{equation}
Thus, our estimate $\tilde{P}(n,n_1,-v)$ for $P(V(x)\leq -v)={P}(n,n_1,-v)$ satisfies
\begin{equation*}
P(n,n_1,-v)\leq \tilde{P}(n,n_1,-v)\leq P(n,n_1,-v)\frac{1}{(1-\varepsilon_1)^t} \le P(n,n_1,-v) (1+\varepsilon).
\end{equation*}

\paragraph{Estimating $Q(t,r,v)$.}
\label{es-Q}
Recall that 
 $Q(t,r,v)$ is the probability of exactly $r$ events of $P_1$ in the first $t$ steps and the statistic at time $t$ is $\geq v$,
 \begin{equation*}
 Q(t,r,v)=\Pr\left(\sum_{j=1}^t c_j \left(x_j -  \frac{n_1 -\sum_{i=1}^j x_i}{n-j}\right) \geq v~\mbox{ AND}~\sum_{i=1}^t x_i =r\right).
 \end{equation*}
Let $\tilde{Q}(t,r,v)$ be a step function defined by a sequence of points $v^{t}_{k,r}$, $k=0,\dots,\ell$. The value of the function in the interval $[v^{t}_{k,r},v^{t}_{k+1,r}]$ is $(1-\varepsilon_1)^k$ and for $v<v^t_{0,r}$ the value of the function is 1.
Consecutive points in the sequence may be the same ($v^{t}_{k+1,r}=v^{t}_{k,r}$), in that case the value of $\tilde{Q}(t,r,v)$ is $(1-\varepsilon_1)^{k_v}$, where $k_v=\arg\max_k [v\geq v^{t}_{k,r}]$. (Note that the sequence $v^{t}_{k,r}$, $k=0,\dots,\ell$ is monotone non-decreasing in $k$, since larger $k$ corresponds to smaller probability.)

For $t=0$, we define $\tilde{Q}(0,r,v)$  by the set of points $v^0_{k,0}= 0$ and $v^0_{k,r}=\infty$ for $r>0$, $k=0,\dots,\ell$.
These functions satisfy $\tilde{Q}(0,r,v)=Q(0,r,v)$ for all $r$ and $v$.

Assume that iteration $t+1$ starts with a set of functions $\tilde{Q}(t,r,v)$, for $r=0,\dots, n_1$ such that for all $r$ and $v$
$$(1-\varepsilon_1)^t \tilde{Q}(t,r,v)\leq Q(t,r,v)\leq \tilde{Q}(t,r,v).$$
We show then iteration $t+1$ computes functions $\tilde{Q}(t+1,r,v)$ with the same approximation properties.

To compute an estimates for the functions $Q(t+1,r,v)$, $r=0,\dots,n_1$, 
we use the relations given in the exact computation, estimating $Q(t,r,v)$ by $\tilde{Q}(t,r,v)$. 
In the case $c_{t+1}=1$ we use
$$\hat{Q}(t+1,r,v) = (1-\frac{n_1-r}{n-t})\tilde{Q}(t,r,v+\frac{n_1-r}{n-t})+\frac{n_1-r}{n-t}\tilde{Q}(t,r-1,v-(1-\frac{n_1-r}{n-t})),$$
and compute (for each $r$) the function at the $2\ell$ points corresponding to change in values in the functions $\tilde{Q}(t,r,v)$ and $\tilde{Q}(t,r-1,v)$:
$$v^t_{k,r}=v+\frac{n_1-r}{n-t}~\mbox{ and}~ v^t_{k,r-1}=v-(1-\frac{n_1-r}{n-t}),~\mbox{ for}~ k=1,\dots,\ell.$$
In the case $c_{t+1}=0$ we use 
 $$\hat{Q}(t+1,r,v)=(1-\frac{n_1-r}{n-t})\tilde{Q}(t,r,v)+\frac{n_1-r}{n-t}\tilde{Q}(t,r-1,v),$$
 and compute the function (for each $r$) in the $2\ell$ points $v^t_{k,r}$ and $v^t_{k,r-1}$.
 
Let $v_1\leq v_2\leq \dots \leq v_{2\ell}$ be the $2\ell$ points for which the value of $\hat{Q}(t+1,r,v)$ was computed.
We extend the function $\hat{Q}(t+1,r,v)$ to a step function over all values of $v$, such that 
$\hat{Q}(t+1,r,v)=\hat{Q}(t+1,r,v_j)$, where $j$ is the largest index such that $v\geq v_j$.
 
Since we computed the function $\hat{Q}(t+1,r,v)$ in all points in which $\tilde{Q}(t,r,v)$ and $\tilde{Q}(t,r-1,v)$ change values, and by the assumptions on the values of $\tilde{Q}(t,r,v)$,
for all $r$ and $v$ we have
$$\hat{Q}(t+1,r,v)(1-\varepsilon_1 )^{t+1} \leq Q(t+1,r,v) \leq \hat{Q}(t+1,r,v).$$

We now approximate the function $\hat{Q}(t+1,r,v)$ by a function $\tilde{Q}(t+1,r,v)$ that is defined by the sequence of only $\ell$ values:
$$v^{j}_{k,r}=\arg \min_v[\tilde{Q}(j,r,v)\leq (1-\varepsilon_1)^k ], ~k=0,\dots,\ell.$$ 
Consider a value $v$ such that $v^{t+1}_{k,r} \leq v <v^{t+1}_{k+1,r}$. We have:
\begin{equation}
Q(t+1,r,v) \leq \hat{Q}(t+1,r,v)\leq \hat{Q}(t+1,r,v^{t+1}_{k,r})\leq \tilde{Q}(t+1,r,v^{t+1}_{k,r}),
\end{equation}
and 
\begin{equation}
Q(t+1,r,v)\geq Q(t+1,r,v_{k+1,r}^{t+1}) \geq \hat{Q}(t+1,r,v^{t+1}_{k,r})(1-\varepsilon_1)^{t+1} \ge \tilde{Q}(t+1,r,v^{t+1}_{k,r})(1-\varepsilon_1)^{t+1}.
\end{equation}
Thus, our estimate $\tilde{Q}(n,n_1,v)$ for $Q(V(x)\geq  v)={Q}(n,n_1,v)$ satisfies
\begin{equation*}
Q(n,n_1,v)\leq \tilde{Q}(n,n_1,v)\leq Q(n,n_1,v)\frac{1}{(1-\varepsilon_1)^t} \le Q(n,n_1,v) (1+\varepsilon).
\end{equation*}

From the discussion above the following theorem is readily derived.

\begin{theorem}
The algorithm above is a FPTAS for computing $Pr(|V(\mathbf{x})|\geq |v|)$.
\end{theorem}
\begin{proof}
We first consider the approximation ratio. Note that $\Pr(|V(\mathbf{x})|\geq |v|) = \Pr(V(\mathbf{x}) \ge |v|) + \Pr(V(\mathbf{x})\le -|v|)$; from the discussion above we have that $\Pr(V(\mathbf{x})  \le P(n,n_1,-|v|) \leq \Pr(V(\mathbf{x}) (1+\varepsilon)$ and $\Pr(V(\mathbf{x}) \ge |v|)\le Q(n,n_1,|v|) \le Pr(V(\mathbf{x}) (1+\varepsilon)$, therefore if we define $\tilde{p}= P(n,n_1,-|v|) + Q(n,n_1,|v|)$ we have that $\Pr(|V(\mathbf{x})|\geq |v|) \le \tilde{p} \le (1+\varepsilon) \Pr(|V(\mathbf{x})|\geq |v|)$.

The run-time of each iteration is $O(\ell n_1)=O(\varepsilon^{-1} n n_1^2 \log n)$ and there are $n$ iterations, thus for any $\varepsilon>0$ the algorithm computes an $\varepsilon$-approximation in  $O(\varepsilon^{-1} n^2 n_1^2 \log n)$ time.
\end{proof}

\subsection*{FPTAS running time}

We ran experiments to study how the running time of the FTPAS varies for different values of $n,n_1$, and $\varepsilon$. We also compared the running time of the FPTAS with the running time of the exhaustive algorithm for permutational distribution. For simplicity, in our tests we assumed no censoring (i.e., $\mathbf{c}={1}^n$). Results are shown in Figure~\ref{fig:runtime_FPTAS}. Figure~\ref{fig:runtime_FPTAS_exhaustive_small} shows the average runtime of the FPTAS with $n_1=10,\varepsilon=5$, for different values of $n$.  (Standard deviations are not shown since they are very small compared to the runtime.) For the same instances we also ran the exhaustive algorithm 10 times, stopping it after 5 hours (i.e., 18000 seconds) if it did not terminate. Results for the exhaustive algorithm are shown in Figure~\ref{fig:runtime_FPTAS_exhaustive_small} as well. The starred ($*$) values of $n$ are values for which the exhaustive algorithm was stopped after 5 hours in each of the 10 runs. The exhaustive algorithm is 
practical only for very small values of $n$, while the FPTAS can be used for much larger values of $n$. Figure~\ref{fig:runtime_FPTAS_exhaustive_n1} shows how the runtime of the FPTAS varies for different values of $n_1$, with $n=100, \varepsilon=5$. As expected the runtime increases with $n_1$, but it is still practical for values of $n_1$ up to $0.2n$. We report the runtime of the exhaustive algorithm for comparison. Note that for $n=100, n_1=20$ the exhaustive algorithm would take more than $160$ years even running on a $100$ Ghz machine under the unrealistic assumption that it could compute the log-rank statistic of a vector $\mathbf{x}$ every clock cycle. Figure~\ref{fig:runtime_FPTAS_eps} shows how the runtime of the FPTAS varies for different values of the approximation parameter $\varepsilon$. We measured the runtime over 10 runs with $n=100, n_1=10$, and no censoring. As expected, the runtime decreases by increasing $\varepsilon$.

\begin{figure}
\centering
                \subfloat[][]{
                 \includegraphics[width=0.3\textwidth]{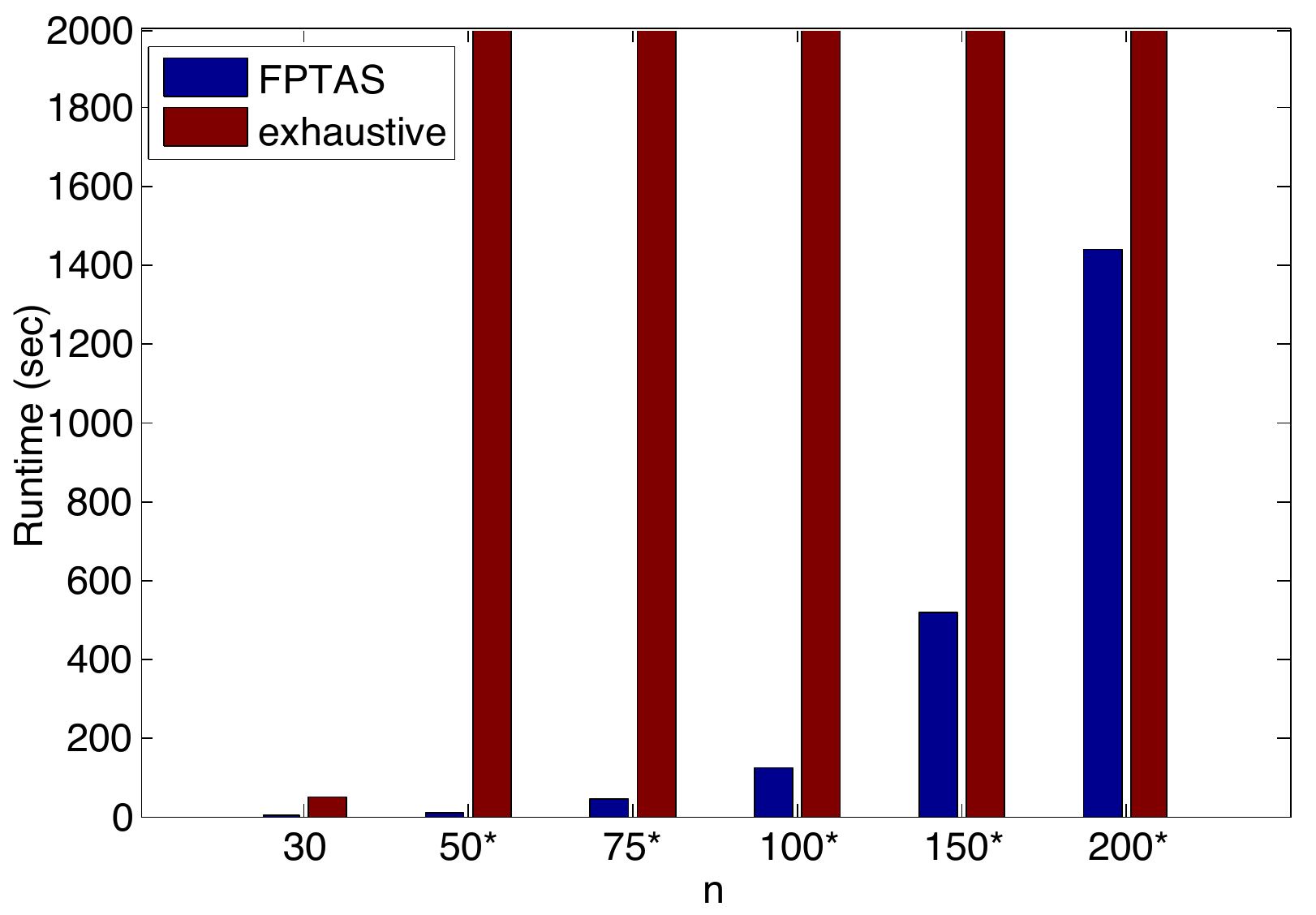}
                 \label{fig:runtime_FPTAS_exhaustive_small}
                 }
                \quad
		 \subfloat[][]{
                \includegraphics[width=0.25\textwidth]{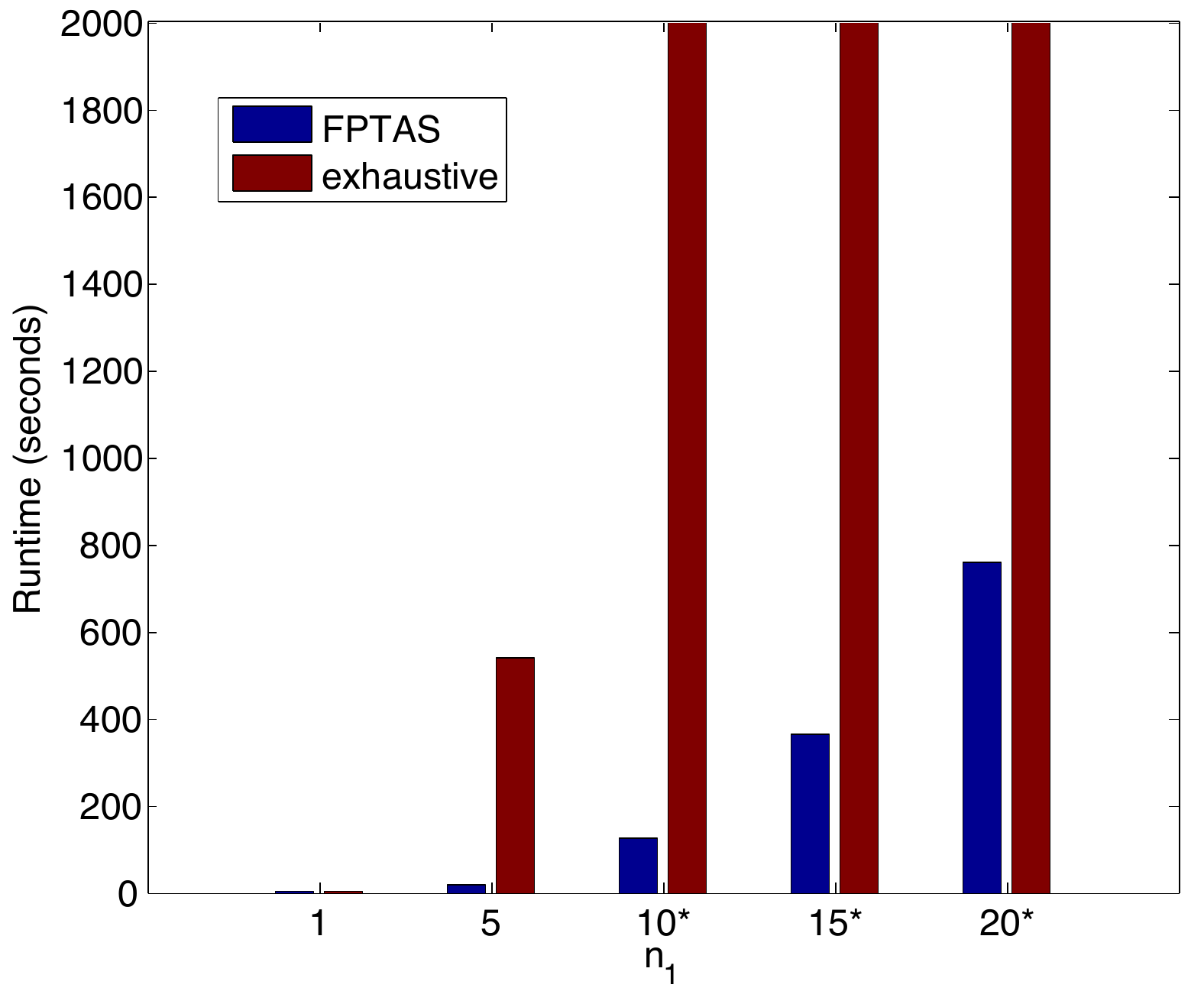}
                \label{fig:runtime_FPTAS_exhaustive_n1}
                }        
                \quad
                \subfloat[][]{
                \includegraphics[width=0.3\textwidth]{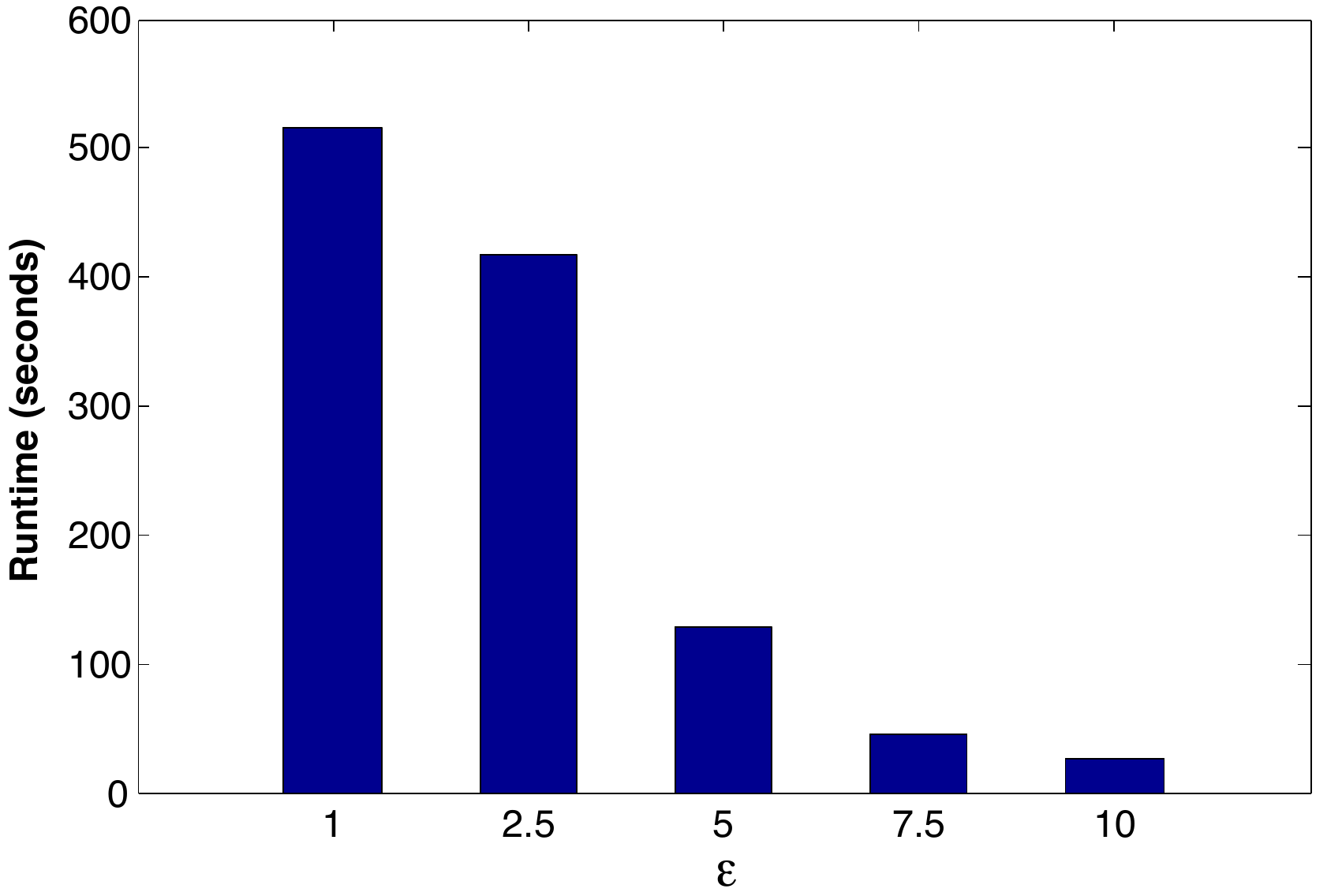}
                \label{fig:runtime_FPTAS_eps} 
               }
        \caption{Running time of the FPTAS, and its comparison with the running time of the exhaustive enumeration algorithm, for different values of the parameters. (a) Runtime of FPTAS and  of the exhaustive enumeration for different values of  $n$, and for $n_1=10, \varepsilon=5$, no censoring. (b) Runtime of FPTAS and of the exhaustive enumeration for $n=10, \varepsilon=5$, no censoring, and different  values of $n_1$. ({c}) Runtime of the FPTAS for different values of $\varepsilon$, and for $n=100, n_1=10$, no censoring.}\label{fig:runtime_FPTAS}
\end{figure}

\section*{Cancer data}

\subsection*{TCGA data}

We analyzed somatic mutation  and clinical data, including survival information, from the TCGA data portal (\texttt{https://tcga-data.nci.nih.gov/tcga/}). In particular we considered single nucleotide variants (SNPs) and small indels for 
colorectal carcinoma (COADREAD), glioblastoma multiforme (GBM), kidney renal clear cell carcinoma (KIRC), lung squamous cell carcinoma (LUSC), ovarian serous adenocarcinoma  (OV), and uterine corpus endometrial carcinoma (UCEC). 
We restricted our analysis to patients for which somatic mutation and survival data were both available. We only considered genes mutated in $> 1\%$ of patients. Since genes mutated in the same patients have the same association to survival, we collapsed them into \emph{metagenes}, recording the genes that appear in a metagene. Table~\ref{table:cancer_data} shows a number of statistics for each dataset.
Fig.~\ref{fig:cancer_data1} shows the comparison between the exact permutational $p$-value, computed using the FPTAS when the gene was mutated in $<10\%$ of the samples and the \R \texttt{survdiff} $p$-value or the exact conditional $p$-value for each considered gene.

We compared the genes reported in the top positions by the different tests. The number of genes shared in the top positions of the lists obtained by the exact permutational test and the exact conditional test, and by the exact permutational test and \R~ \texttt{survdiff} for the different cancer types are: in COADREAD, 6 genes are in the top $10$ positions by the exact permutational test and by \R~ \texttt{survdiff}, and 7 in  the top $10$ positions by the exact permutational test and  by the exact conditional test, while 17 genes are in  the top $25$ positions by the exact permutational test and by \R~ \texttt{survdiff}, and 21 in the top $25$ positions by the exact permutational test and by the exact conditional test; in GBM, no gene is in  the top $10$ positions by the exact permutational test and by \R~ \texttt{survdiff}, and  2 in the top $10$ positions by the exact permutational test and  by the exact conditional test, while no gene is in  the top $25$ positions by the exact permutational test and by \R \
texttt{survdiff}, and 2 in the top $25$ positions by the exact permutational test and by the exact conditional test; in KIRC, 2 genes are in  the top $10$ positions by the exact permutational test and by \R~ \texttt{survdiff}, and 4 in the top $10$ positions by the exact permutational test and  by the exact conditional test, while 9 genes are in  the top $25$ positions by the exact permutational test and by \R~ \texttt{survdiff}, and 15 in the top $25$ positions by the exact permutational test and by the exact conditional test; in LUSC, no gene is in  the top $10$ positions by the exact permutational test and by \R~ \texttt{survdiff}, and no gene is in the top $10$ positions by the exact permutational test and  by the exact conditional test, while 1 gene is in  the top $25$ positions by the exact permutational test and by \R~ \texttt{survdiff}, and 4 in the top $25$ positions by the exact permutational test and by the exact conditional test; in OV, no gene is in  the top $10$ positions by the exact permutational test and 
by \R~ \texttt{survdiff}, and 1 gene is in the top $10$ positions by the exact permutational test and  by the exact conditional test, while 3 genes are in  the top $25$ positions by the exact permutational test and by \R~ \texttt{survdiff}, and 7 in the top $25$ positions by the exact permutational test and by the exact conditional test; in UCEC, 5 genes are in  the top $10$ positions by the exact permutational test and by \R~ \texttt{survdiff}, and 7 in the top $10$ positions by the exact permutational test and  by the exact conditional test, while 15 genes are in  the top $25$ positions by the exact permutational test and by \R~ \texttt{survdiff}, and 21 in the top $25$ positions by the exact permutational test and by the exact conditional test.

Table~\ref{table:permutational}, Table~\ref{table:conditional}, and Table~\ref{table:Rsurvdiff} report the 10 genes with smallest $p$-values identified using the exact permutational test, the exact conditional test, or \R~ \texttt{survdiff} for the six cancer datasets. For each cancer type, the top 10 genes, their rank, and $p$-value using the exact permutational test, the exact conditional test, and \R~ \texttt{survdiff} are reported. The number of samples with a mutation in the gene and refences supporting the association of mutations in the gene with survival are also reported.

\begin{table}[h]
\begin{center}
\caption{Parameters of the cancer datasets analyzed.}
\label{table:cancer_data}
\begin{tabular}{c | c c c c}
dataset & num. patients & $\%$ censoring & num. mutated genes & num. genes mutated $> 1\%$.\\
\hline
COADREAD & 188 & 92$\%$ & 2716 & 1099\\
GBM & 268 & 30$\%$& 3702& 1689\\
KIRC &292 & 74$\%$ & 4448& 1879\\
LUSC & 175 & 60$\%$& 8465&5776 \\
OV &315 & 43$\%$& 3740&665 \\
UCEC &235 & 93$\%$& 8304& 5792\\
\end{tabular}
\begin{center}
{\footnotesize
For each dataset we show: the number of patients with both mutation and survival data; the percentage of patients with censored survival data; the number of (meta)genes mutated in at least one sample; the number of (meta)genes mutated in more than $1\%$ of all samples.
}
\end{center}
\end{center}
\end{table}

\begin{figure}[h]
\centering

                \subfloat[][]{
                \includegraphics[width=0.16\textwidth]{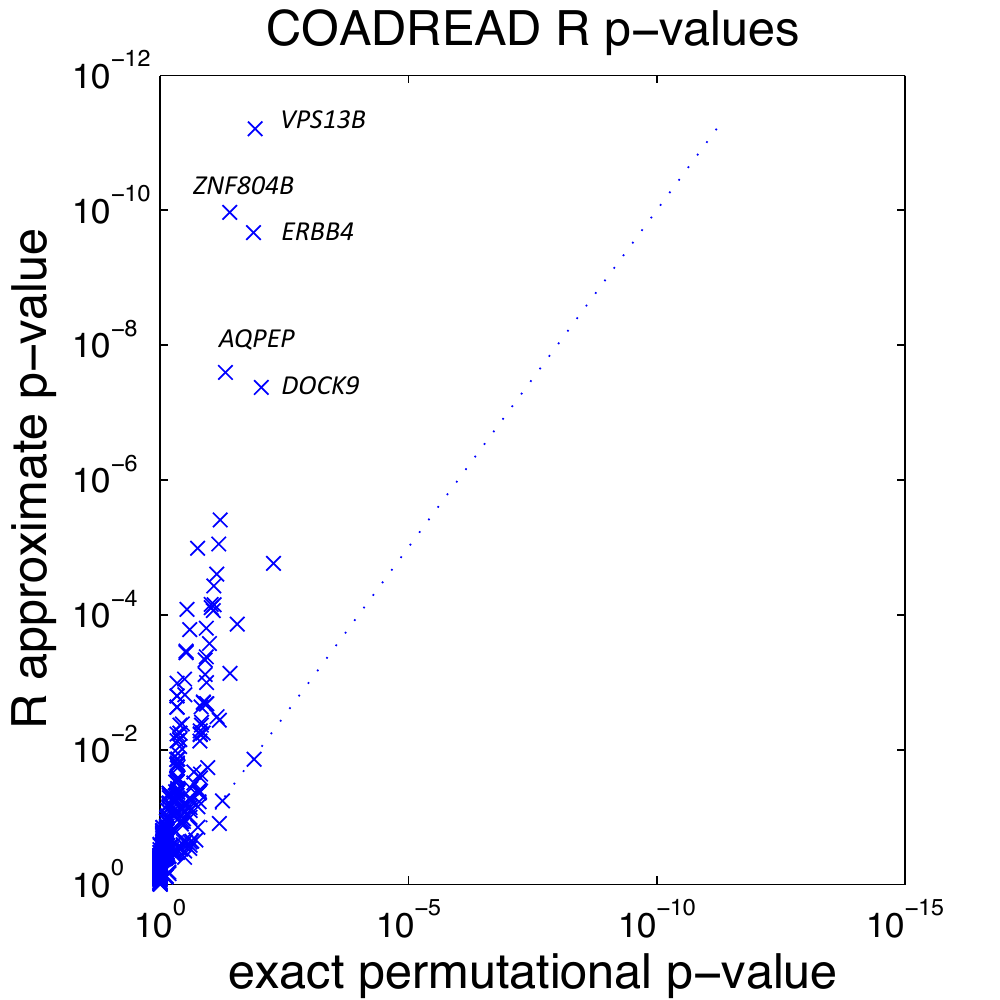}
                \label{fig:R_COADREAD}
                }
                \hspace{-0.7cm}
                \quad
                \subfloat[][]{
                \includegraphics[width=0.16\textwidth]{GBM_FPTAS_vs_R.pdf}
                \label{fig:R_GBM}
                }
                \quad
                \hspace{-0.7cm}
                \subfloat[][]{
                \includegraphics[width=0.16\textwidth]{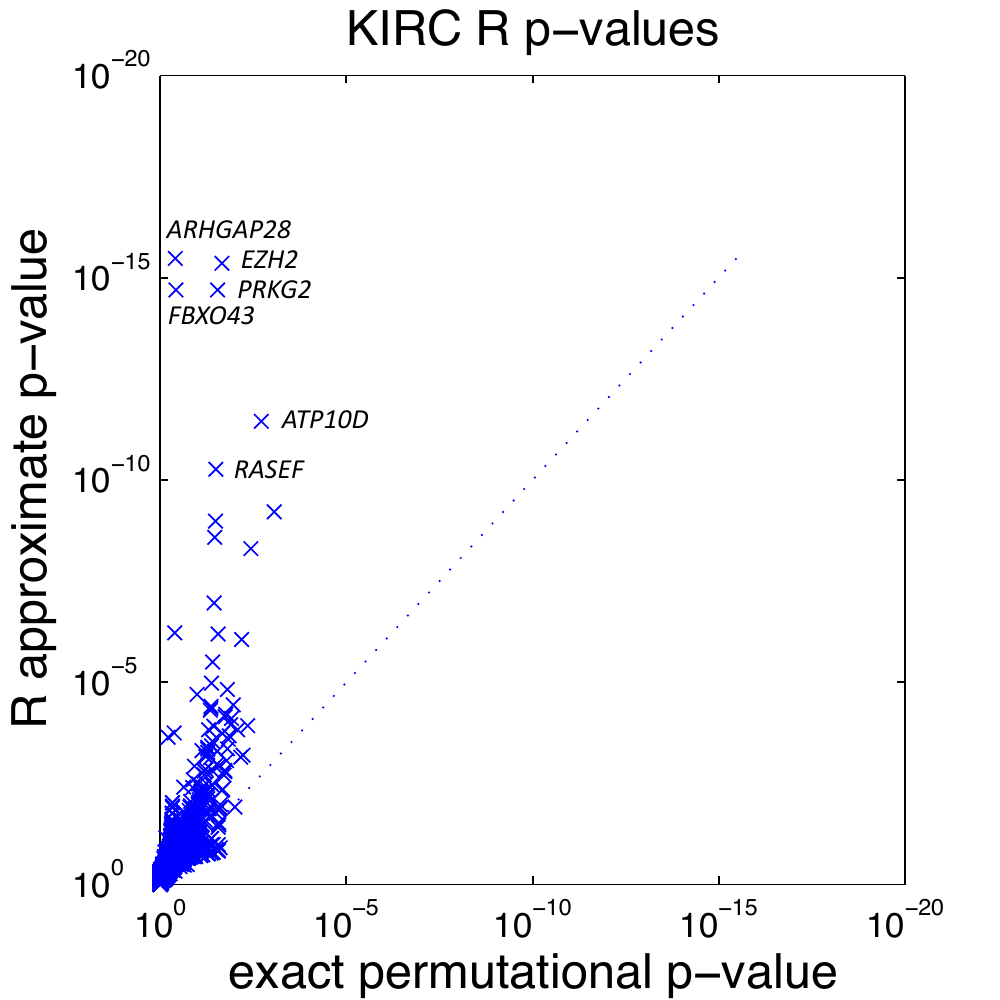}
                \label{fig:R_KIRC}
                }
                \quad
                \hspace{-0.7cm}
                \subfloat[][]{
                \includegraphics[width=0.16\textwidth]{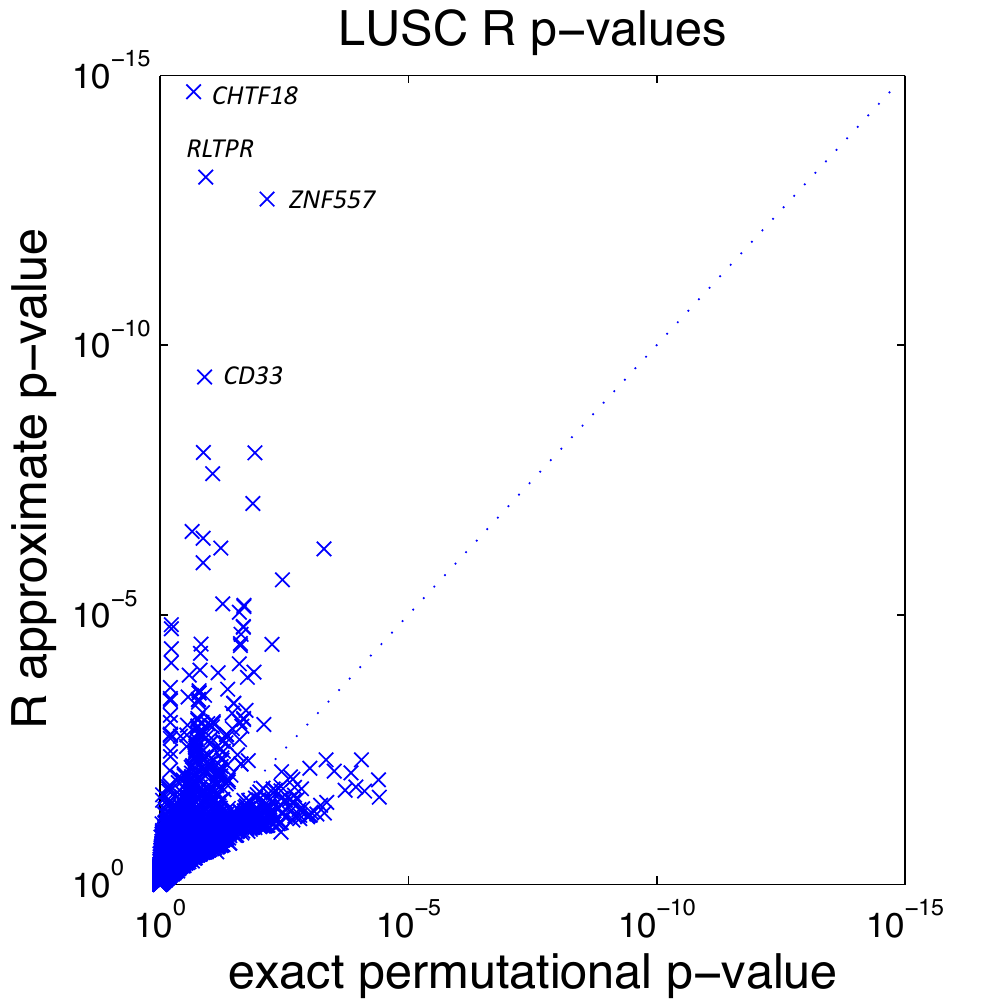}
                \label{fig:R_LUSC}
                }
                \quad
                \hspace{-0.7cm}
                 \subfloat[][]{
                \includegraphics[width=0.16\textwidth]{OV_FPTAS_vs_R.pdf}
                \label{fig:R_OV}
                }
		\quad                
		\hspace{-0.7cm}
                 \subfloat[][]{
                \includegraphics[width=0.16\textwidth]{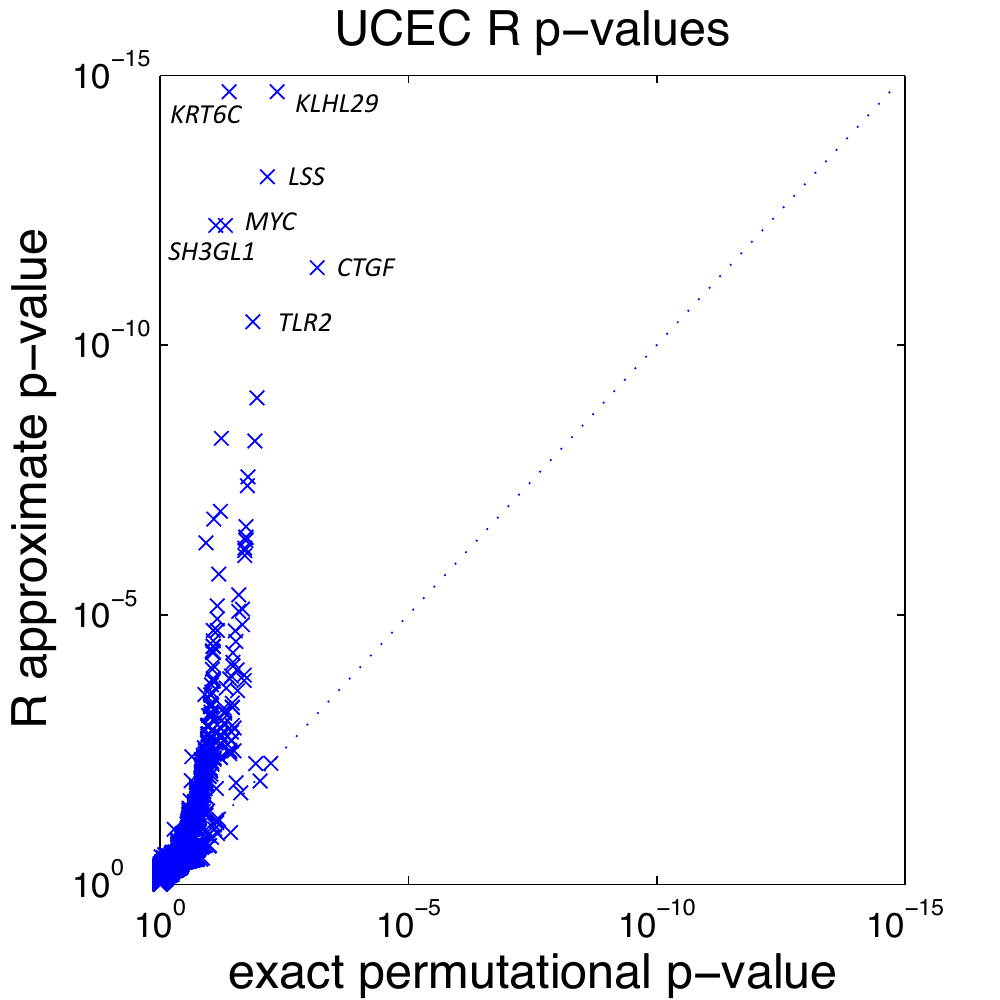}
                \label{fig:R_UCEC}
                }
                
		 \subfloat[][]{
                \includegraphics[width=0.16\textwidth]{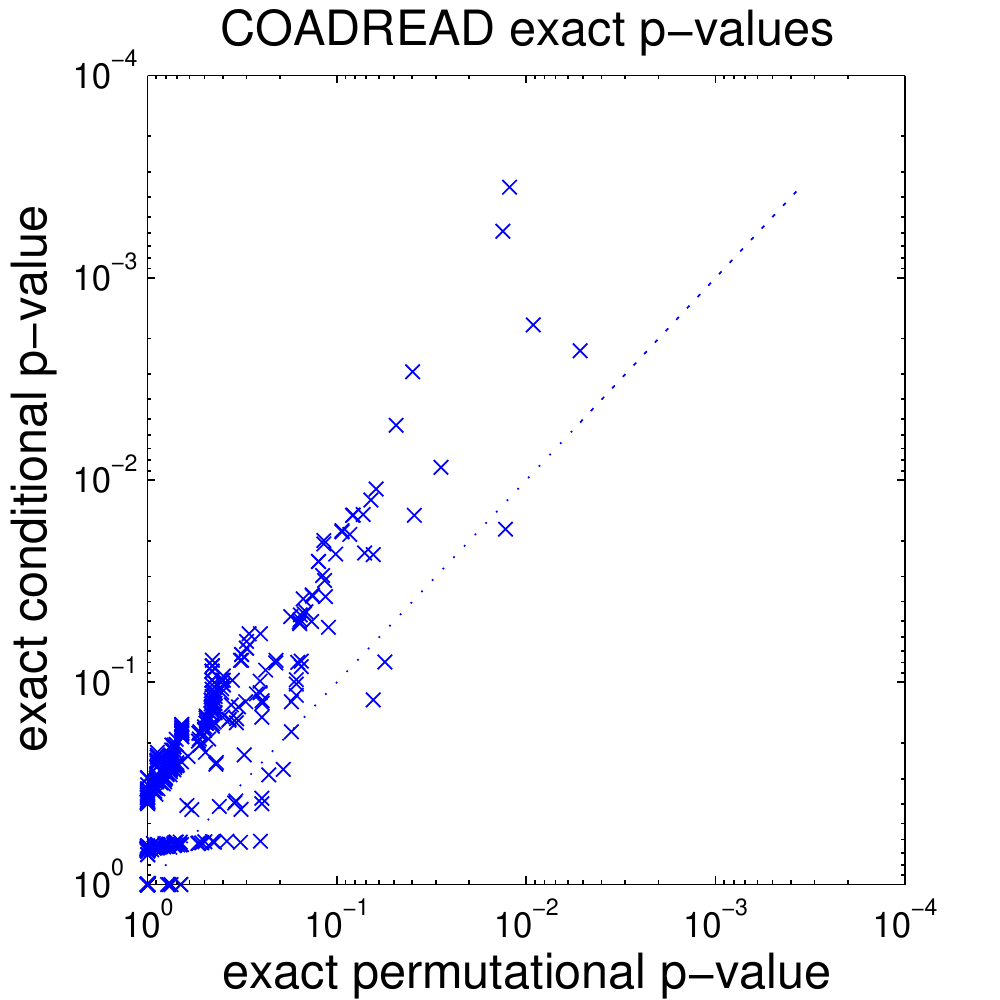}
                \label{fig:exact_COADREAD}
		}
		\quad
		\hspace{-0.7cm}
		 \subfloat[][]{
                \includegraphics[width=0.16\textwidth]{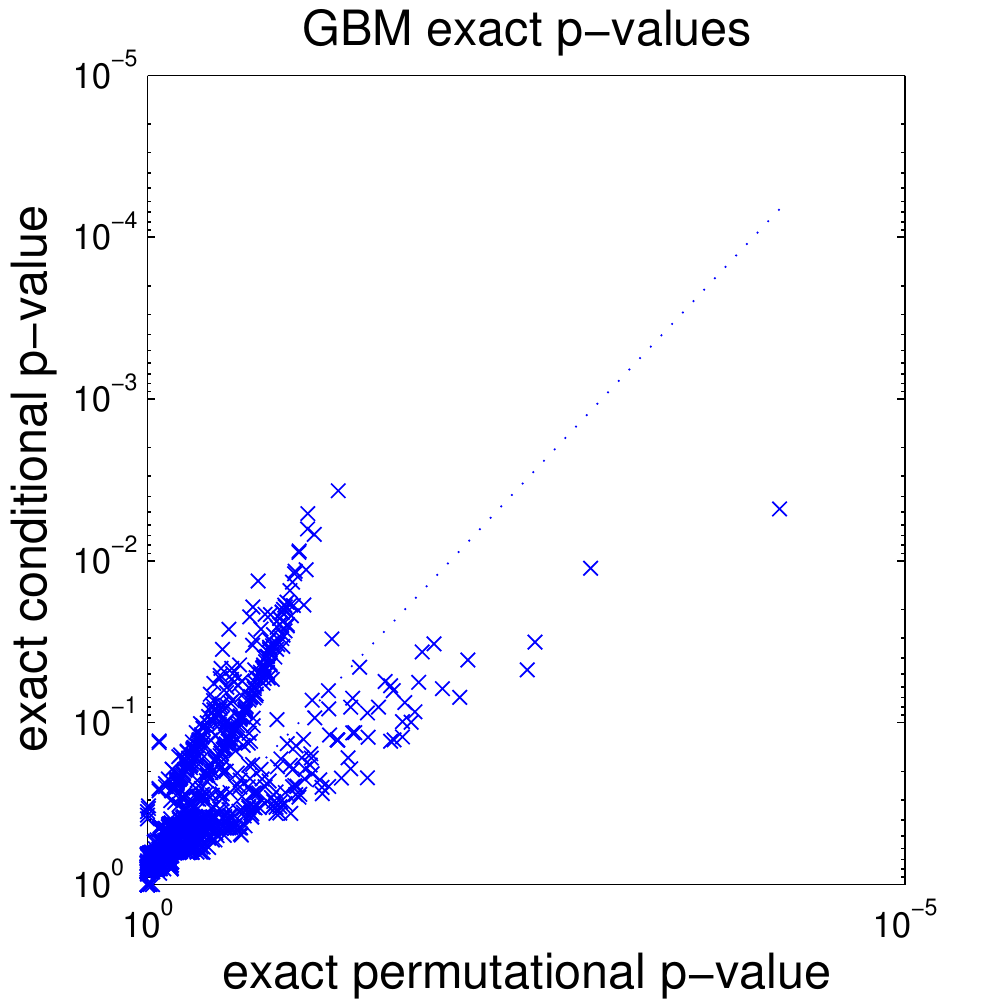}
                \label{fig:exact_GBM}
                }                
               \quad
		\hspace{-0.7cm}
		 \subfloat[][]{
                \includegraphics[width=0.16\textwidth]{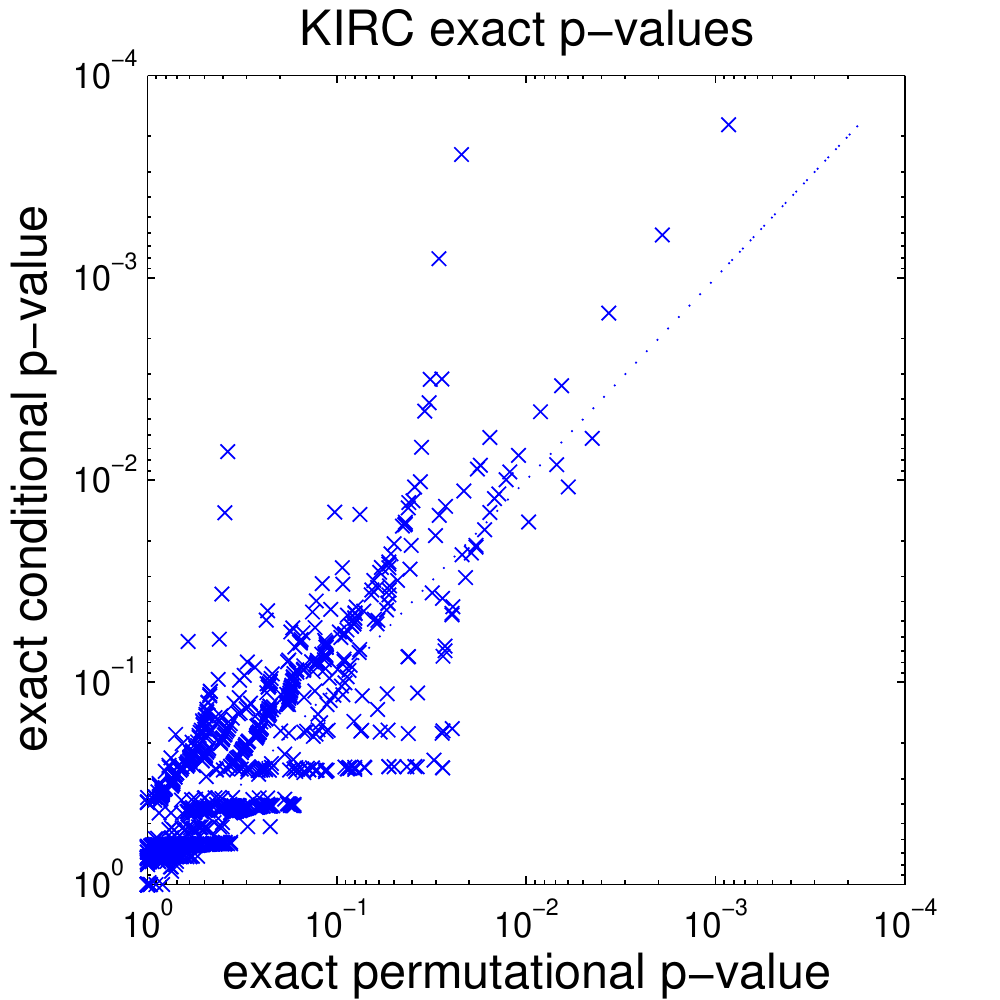}
                \label{fig:exact_KIRC}
                }                
               \quad
               \hspace{-0.7cm}
		 \subfloat[][]{
                \includegraphics[width=0.16\textwidth]{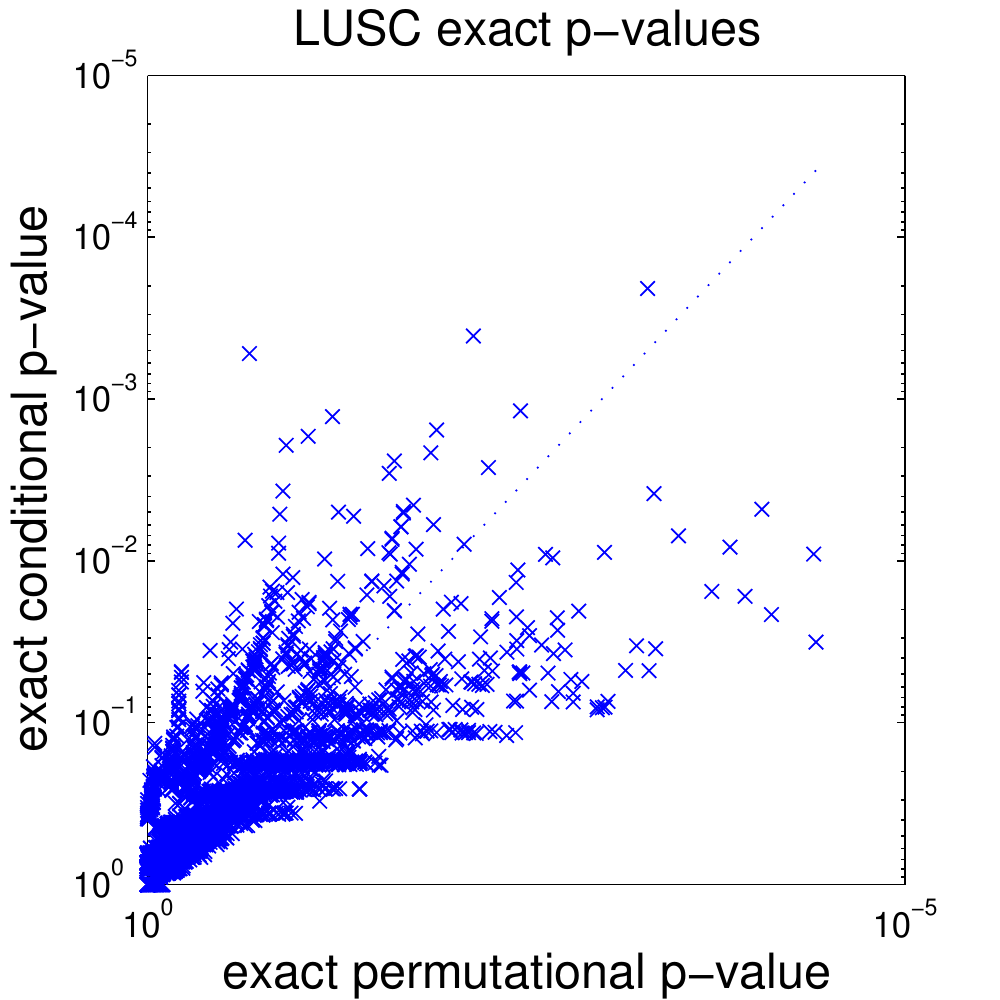}
                \label{fig:exact_LUSC}
                }
                  \quad
                  \hspace{-0.7cm}
		 \subfloat[][]{
                \includegraphics[width=0.16\textwidth]{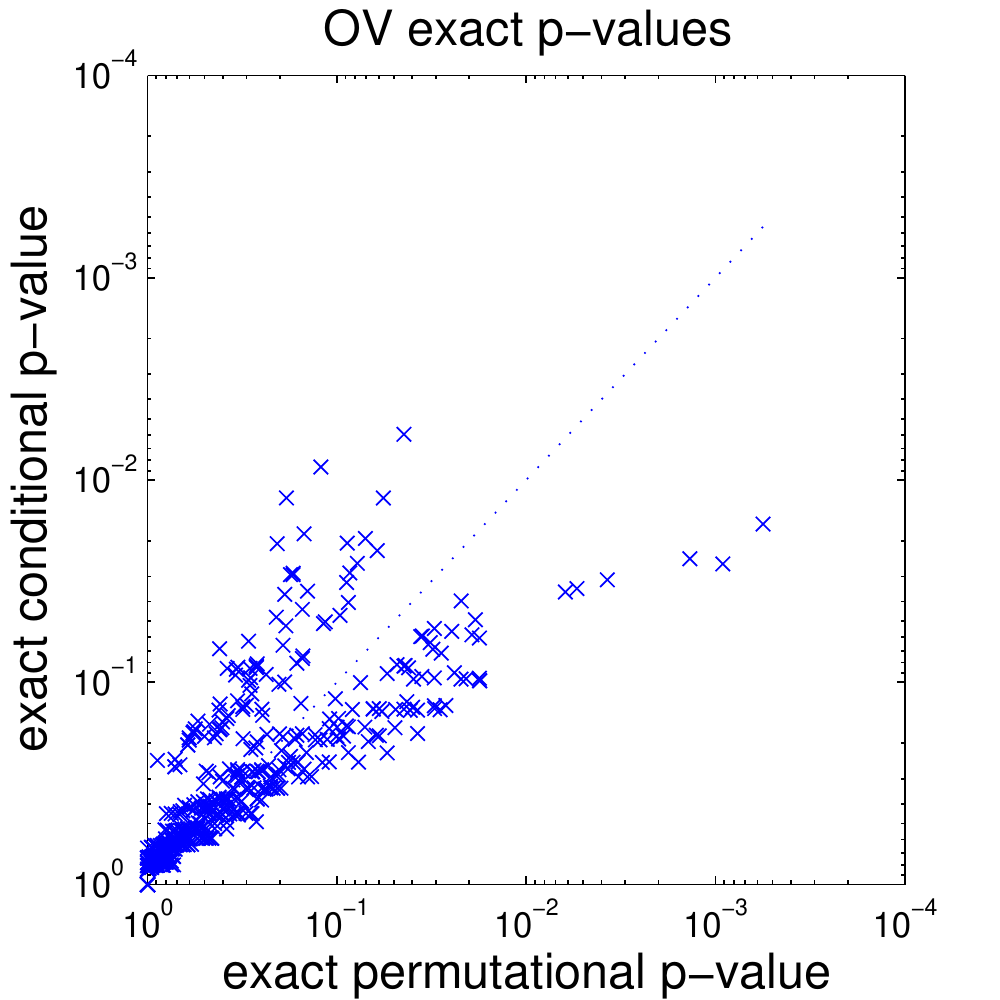}
                \label{fig:exact_OV}
                }
                \quad
                \hspace{-0.7cm}
		 \subfloat[][]{
                \includegraphics[width=0.16\textwidth]{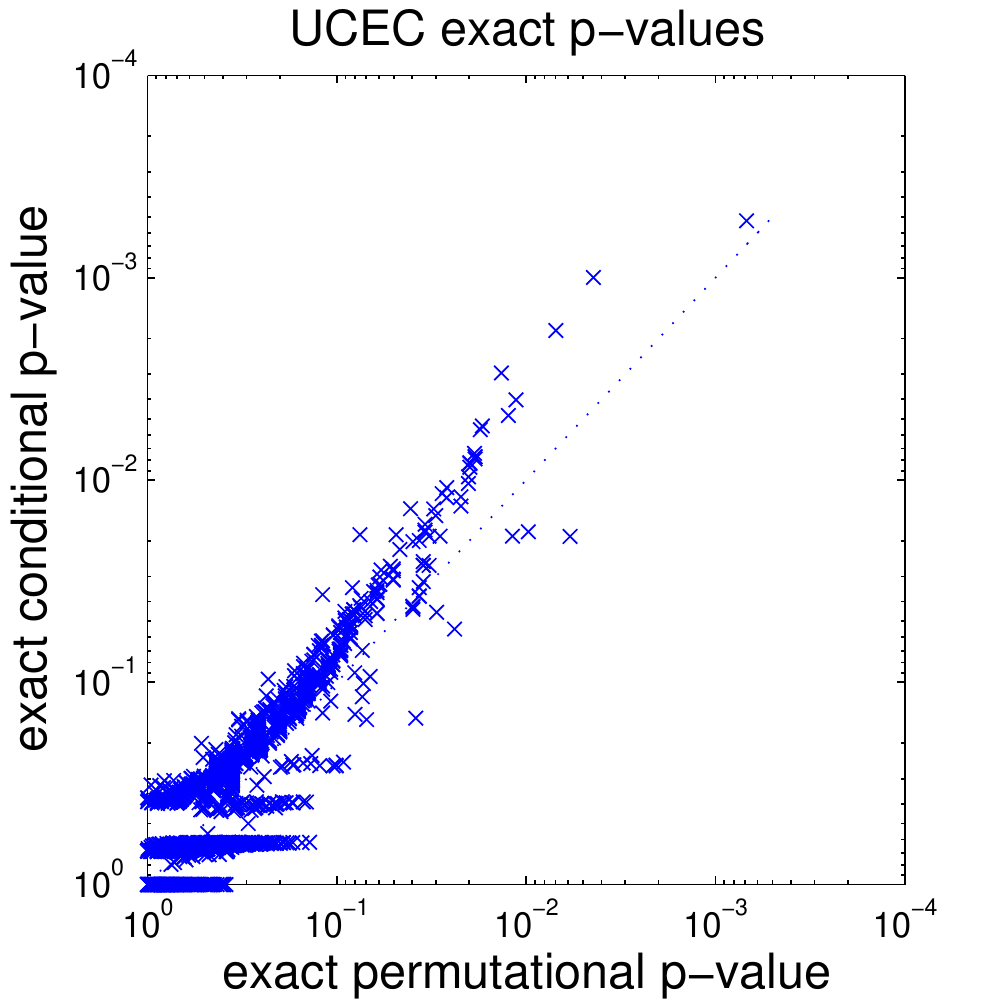}
                \label{fig:exact_UCEC}
                }
        \caption{Comparison of the $p$-values from the exact permutational test, the exact conditional test, and the asymptotic approximation (as implemented in the \texttt{survdiff} package in \R) for cancer datasets COADREAD, GBM, KIRC, LUSC, OV, UCEC.  (a,b,c,d): Each data point represents a gene, and the $p$-values computed using  the exact permutational test and the  $p$-values from \R~ \texttt{survdiff} for the gene are shown. (e,f,g,h): Each data point represents a gene, and the $p$-values computed using  the exact permutational test and the  exact condtional test for the gene are shown.}\label{fig:cancer_data1}
        \end{figure}

\begin{table}[htbp]
\begin{center}
\caption{The 10 genes  with smallest $p$-values identified using the exact permutational test in cancer data. For each cancer type, we show the top 10 genes, their rank, and $p$-value using the exact permutational test, the exact conditional test, and \R~ \texttt{survdiff}. The number of samples with a mutation in the gene and references supporting the association of mutations in the gene with survival are also reported.}
\label{table:permutational}
{\scriptsize
\begin{tabular}{c | c c c c c c c c c c c}
& & \multicolumn{2}{c}{exact permutational} & & \multicolumn{2}{c}{exact conditional} & &  \multicolumn{2}{c}{R \texttt{survdiff}} & & \\
\cline{3-4} \cline{6-7} \cline{9-10} 
cancer type & gene & rank & $p$ & & rank & $p$ & & rank & $p$ & num. mut. & refs\\
\hline
COADREAD & PTPRM & 1 & 5.20e-03&  & 4 & 2.30e-03 & & 10 & 1.74e-05& 8\\
& DOCK9 & 2 & 9.18e-03&  & 3 & 1.72e-03 & & 5 & 4.25e-08& 6\\
& VPS13B & 3 & 1.22e-02&  & 1 & 3.57e-04 & & 1 & 6.21e-12& 8\\
& KRAS & 4 & 1.29e-02&  & 14 & 1.75e-02 & & 82 & 1.40e-02& 85 & \cite{Heinemann:2009kx}\\
& ERBB4 & 5 & 1.33e-02&  & 2 & 5.89e-04 & & 3 & 2.16e-10& 8\\
& CACNA1E & 6 & 2.82e-02&  & 7 & 8.68e-03 & & 21 & 1.38e-04& 7\\
& PCDH11X & 7 & 3.90e-02&  & 12 & 1.49e-02 & & 37 & 7.42e-04& 7\\
& ZNF804B & 8 & 3.97e-02&  & 5 & 2.92e-03 & & 2 & 1.07e-10& 5\\
& AQPEP & 9 & 4.86e-02&  & 6 & 5.38e-03 & & 4 & 2.55e-08& 5\\
& FBXW7 & 10 & 5.58e-02&  & 64 & 7.95e-02 & & 142 & 5.75e-02& 21 &\cite{Jenkinson:2012vn,Cancer-Genome-Atlas-Network:2012ys}\\
GBM & IDH1 & 1 & 6.73e-05&  & 2 & 4.79e-03 & & 27 & 3.76e-03& 15 & \cite{Nobusawa:2009fk,Myung:2012zr}\\
& MED12L & 2 & 1.19e-03&  & 8 & 1.11e-02 & & 62 & 1.31e-02& 3\\
& EMILIN1 & 3 & 2.77e-03&  & 45 & 3.18e-02 & & 100 & 3.17e-02& 4\\
& ATP6V0A4 & 4 & 3.11e-03&  & 77 & 4.72e-02 & & 101 & 3.29e-02& 4\\
& VARS2 & 5 & 7.67e-03&  & 64 & 4.11e-02 & & 111 & 3.80e-02& 3 & \cite{Chae:2011ly}\\
& GALR1 & 6 & 8.67e-03&  & 119 & 6.97e-02 & & 156 & 6.61e-02& 3\\
& KCNH4 & 7 & 1.13e-02&  & 108 & 6.16e-02 & & 125 & 4.57e-02& 3\\
& CXorf22 & 8 & 1.29e-02&  & 47 & 3.26e-02 & & 102 & 3.32e-02& 8\\
& ITGAM & 9 & 1.54e-02&  & 57 & 3.65e-02 & & 109 & 3.73e-02& 6\\
& PLXNB3 & 10 & 1.62e-02&  & 99 & 5.63e-02 & & 141 & 5.67e-02& 6\\
KIRC & CDCA2 & 1 & 8.54e-04&  & 1 & 1.75e-04 & & 7 & 6.15e-10& 6\\
& ATP10D & 2 & 1.91e-03&  & 3 & 6.15e-04 & & 5 & 3.55e-12& 3\\
& TDRD7 & 3 & 3.66e-03&  & 5 & 1.50e-03 & & 10 & 5.00e-09& 3\\
& KIF27 & 4 & 4.48e-03&  & 13 & 6.23e-03 & & 27 & 1.19e-04& 4\\
& RELN & 5 & 5.99e-03&  & 23 & 1.08e-02 & & 44 & 6.37e-04& 4\\
& TOPORS & 6 & 6.50e-03&  & 8 & 3.42e-03 & & 14 & 8.83e-07& 3\\
& TCF20 & 7 & 6.89e-03&  & 17 & 8.40e-03 & & 46 & 7.18e-04& 6\\
& CHD7 & 8 & 8.41e-03&  & 11 & 4.60e-03 & & 30 & 1.51e-04& 7\\
& BAP1 & 9 & 9.72e-03&  & 37 & 1.62e-02 & & 98 & 1.23e-02& 27 & \cite{Hakimi:2012ve}\\
& BCL9 & 10 & 1.10e-02&  & 16 & 7.57e-03 & & 19 & 3.67e-05& 3\\
LUSC & ATXN10 & 1 & 3.87e-05&  & 121 & 3.18e-02 & & 229 & 2.41e-02& 3\\
& ZNF304 & 2 & 4.00e-05&  & 38 & 9.10e-03 & & 166 & 1.14e-02& 5\\
& LINGO1 & 3 & 7.61e-05&  & 83 & 2.15e-02 & & 210 & 1.87e-02& 4\\
& C10orf79 & 4 & 8.80e-05&  & 16 & 4.80e-03 & & 116 & 4.91e-03& 9\\
& SPANXN1 & 5 & 1.14e-04&  & 62 & 1.66e-02 & & 190 & 1.57e-02& 5\\
& OR5AS1 & 6 & 1.43e-04&  & 31 & 8.24e-03 & & 146 & 8.59e-03& 7\\
& SLC7A13 & 7 & 1.89e-04&  & 56 & 1.55e-02 & & 203 & 1.80e-02& 5\\
& CCDC85A & 8 & 3.12e-04&  & 25 & 7.01e-03 & & 142 & 7.95e-03& 9\\
& CIT & 9 & 4.43e-04&  & 132 & 3.49e-02 & & 271 & 3.00e-02& 4\\
& DUSP27 & 10 & 4.53e-04&  & 14 & 3.85e-03 & & 117 & 4.93e-03& 11\\
OV & DFNB31 & 1 & 5.61e-04&  & 5 & 1.65e-02 & & 18 & 1.23e-02& 4\\
& ERN2 & 2 & 9.16e-04&  & 13 & 2.60e-02 & & 24 & 2.25e-02& 5\\
& NCOA3 & 3 & 1.37e-03&  & 11 & 2.45e-02 & & 21 & 1.75e-02& 4\\
& CIC & 4 & 3.72e-03&  & 18 & 3.12e-02 & & 26 & 2.47e-02& 4\\
& PKP4 & 5 & 5.40e-03&  & 20 & 3.45e-02 & & 31 & 3.33e-02& 5\\
& BRCA2 & 6 & 6.20e-03&  & 22 & 3.58e-02 & & 32 & 3.45e-02& 9 & \cite{Bolton:2012bh,Cancer-Genome-Atlas-Research-Network:2011uq}\\
& MYST4 & 7 & 1.76e-02&  & 74 & 9.86e-02 & & 71 & 8.00e-02& 5\\
& OVGP1 & 8 & 1.77e-02&  & 38 & 6.05e-02 & & 36 & 4.50e-02& 4\\
& ADAR & 9 & 1.77e-02&  & 68 & 9.55e-02 & & 64 & 7.08e-02& 4\\
& PCDH9 & 10 & 1.77e-02&  & 73 & 9.81e-02 & & 66 & 7.36e-02& 4\\
UCEC & CTGF & 1 & 6.87e-04&  & 1 & 5.23e-04 & & 6 & 3.68e-12& 3\\
& KRT6C & 2 & 4.41e-03&  & 2 & 9.98e-04 & & 2 & 0.00e+00& 3\\
& FAT3 & 3 & 5.88e-03&  & 36 & 1.91e-02 & & 166 & 5.70e-03& 23\\
& LSS & 4 & 6.96e-03&  & 3 & 1.82e-03 & & 3 & 7.61e-14& 3\\
& ARID1A & 5 & 9.74e-03&  & 29 & 1.81e-02 & & 219 & 1.21e-02& 73 & \cite{Quesada:2012kx}\\
& EIF2C4 & 6 & 1.13e-02&  & 5 & 4.02e-03 & & 8 & 9.56e-10& 3\\
& DMD & 7 & 1.18e-02&  & 33 & 1.90e-02 & & 167 & 5.77e-03& 28\\
& CPNE8 & 8 & 1.24e-02&  & 6 & 4.81e-03 & & 13 & 5.98e-09& 3\\
& TLR2 & 9 & 1.35e-02&  & 4 & 2.96e-03 & & 7 & 3.68e-11& 5\\
& CFP & 10 & 1.70e-02&  & 7 & 5.42e-03 & & 14 & 2.78e-08& 5\\
\end{tabular}
}
\end{center}
\end{table}

\begin{table}[htbp]
\begin{center}
\caption{The 10 genes  with smallest $p$-values identified using the exact conditional test in cancer data. For each cancer type, we show the top 10 genes, their rank, and $p$-value using the exact permutational test, the exact conditional test, and \R~ \texttt{survdiff}. The number of samples with a mutation in the gene and references supporting the association of mutations in the gene with survival are also reported.}
\label{table:conditional}
{\scriptsize
\begin{tabular}{c | c c c c c c c c c c c}
& & \multicolumn{2}{c}{exact permutational} & & \multicolumn{2}{c}{exact conditional} & &  \multicolumn{2}{c}{R \texttt{survdiff}} & & \\
\cline{3-4} \cline{6-7} \cline{9-10} 
cancer type & gene & rank & $p$ & & rank & $p$ & & rank & $p$ & num. mut. & refs\\
\hline
COADREAD & VPS13B & 3 & 1.22e-02& & 1 & 3.57e-04 &  & 1 & 6.21e-12& 8\\
& ERBB4 & 5 & 1.33e-02& & 2 & 5.89e-04 &  & 3 & 2.16e-10& 8\\
& DOCK9 & 2 & 9.18e-03& & 3 & 1.72e-03 &  & 5 & 4.25e-08& 6\\
& PTPRM & 1 & 5.20e-03& & 4 & 2.30e-03 &  & 10 & 1.74e-05& 8\\
& ZNF804B & 8 & 3.97e-02& & 5 & 2.92e-03 &  & 2 & 1.07e-10& 5\\
& AQPEP & 9 & 4.86e-02& & 6 & 5.38e-03 &  & 4 & 2.55e-08& 5\\
& CACNA1E & 6 & 2.82e-02& & 7 & 8.68e-03 &  & 21 & 1.38e-04& 7\\
& DOCK3 & 11 & 6.21e-02& & 8 & 1.11e-02 &  & 6 & 3.92e-06& 3\\
& IGF1R & 14 & 6.61e-02& & 9 & 1.26e-02 &  & 7 & 8.95e-06& 3\\
& ZNF568 & 16 & 7.26e-02& & 10 & 1.48e-02 &  & 11 & 2.50e-05& 3\\
GBM & TBC1D2 & 33 & 5.51e-02& & 1 & 3.69e-03 &  & 4 & 7.96e-05& 5\\
& IDH1 & 1 & 6.73e-05& & 2 & 4.79e-03 &  & 27 & 3.76e-03& 15 & \cite{Nobusawa:2009fk,Myung:2012zr}\\
& HDAC9 & 52 & 8.72e-02& & 3 & 5.12e-03 &  & 1 & 6.13e-06& 3\\
& ATP13A2 & 53 & 8.77e-02& & 4 & 6.36e-03 &  & 8 & 3.71e-04& 6\\
& CASKIN2 & 46 & 7.92e-02& & 5 & 6.88e-03 &  & 5 & 1.59e-04& 4\\
& GRM7 & 61 & 9.95e-02& & 6 & 8.68e-03 &  & 2 & 7.33e-05& 3\\
& MYST3 & 62 & 1.00e-01& & 7 & 8.88e-03 &  & 3 & 7.86e-05& 3\\
& MED12L & 2 & 1.19e-03& & 8 & 1.11e-02 &  & 62 & 1.31e-02& 3\\
& BID & 54 & 8.96e-02& & 9 & 1.14e-02 &  & 12 & 7.07e-04& 4\\
& NFYC & 68 & 1.06e-01& & 10 & 1.16e-02 &  & 6 & 1.96e-04& 3\\
KIRC & CDCA2 & 1 & 8.54e-04& & 1 & 1.75e-04 &  & 7 & 6.15e-10& 6\\
& FBXO43 & 26 & 2.19e-02& & 2 & 2.46e-04 &  & 4 & 4.44e-16& 5\\
& ATP10D & 2 & 1.91e-03& & 3 & 6.15e-04 &  & 5 & 3.55e-12& 3\\
& EZH2 & 42 & 2.89e-02& & 4 & 8.07e-04 &  & 2 & 0.00e+00& 3 & \\ 
& TDRD7 & 3 & 3.66e-03& & 5 & 1.50e-03 &  & 10 & 5.00e-09& 3\\
& PKD1L3 & 40 & 2.79e-02& & 6 & 3.17e-03 &  & 13 & 6.48e-07& 5\\
& RASEF & 46 & 3.21e-02& & 7 & 3.19e-03 &  & 6 & 5.48e-11& 3\\
& TOPORS & 6 & 6.50e-03& & 8 & 3.42e-03 &  & 14 & 8.83e-07& 3\\
& PCDH11X & 47 & 3.26e-02& & 9 & 4.16e-03 &  & 8 & 1.05e-09& 3\\
& ZFAT & 48 & 3.43e-02& & 10 & 4.57e-03 &  & 9 & 2.65e-09& 3\\
LUSC & ZNF527 & 13 & 5.00e-04& & 1 & 2.08e-04 &  & 12 & 6.00e-07& 9\\
& ZNF557 & 75 & 7.07e-03& & 2 & 4.09e-04 &  & 3 & 1.96e-13& 3\\
& CHTF18 & 1014 & 2.13e-01& & 3 & 5.24e-04 &  & 1 & 0.00e+00& 4\\
& CNGA4 & 42 & 3.44e-03& & 4 & 1.19e-03 &  & 14 & 2.27e-06& 5\\
& BCL11A & 362 & 6.00e-02& & 5 & 1.28e-03 &  & 11 & 5.80e-07& 7\\
& SIN3B & 111 & 1.23e-02& & 6 & 1.55e-03 &  & 6 & 1.00e-08& 3\\
& ANKRD11 & 487 & 8.70e-02& & 7 & 1.70e-03 &  & 7 & 2.41e-08& 5\\
& RLTPR & 656 & 1.21e-01& & 8 & 1.94e-03 &  & 2 & 7.70e-14& 3\\
& NARG2 & 120 & 1.35e-02& & 9 & 2.15e-03 &  & 8 & 8.62e-08& 3\\
& ATP2B2 & 176 & 2.34e-02& & 10 & 2.42e-03 &  & 23 & 2.32e-05& 7\\
OV & PABPC3 & 39 & 4.42e-02& & 1 & 5.95e-03 &  & 2 & 1.04e-04& 4\\
& KIAA0913 & 82 & 1.21e-01& & 2 & 8.65e-03 &  & 1 & 6.22e-05& 4\\
& NIPBL & 110 & 1.85e-01& & 3 & 1.23e-02 &  & 3 & 2.56e-04& 5\\
& PCDHA9 & 46 & 5.67e-02& & 4 & 1.23e-02 &  & 4 & 8.84e-04& 4\\
& DFNB31 & 1 & 5.61e-04& & 5 & 1.65e-02 &  & 18 & 1.23e-02& 4\\
& RBM44 & 91 & 1.49e-01& & 6 & 1.85e-02 &  & 5 & 9.99e-04& 4\\
& TNRC18B & 55 & 7.05e-02& & 7 & 1.95e-02 &  & 7 & 2.96e-03& 4\\
& PTCH1 & 65 & 8.80e-02& & 8 & 2.05e-02 &  & 10 & 5.46e-03& 6\\
& EFCAB6 & 124 & 2.07e-01& & 9 & 2.07e-02 &  & 6 & 1.56e-03& 5\\
& PCDH15 & 50 & 6.10e-02& & 10 & 2.24e-02 &  & 12 & 6.58e-03& 5\\
UCEC & CTGF & 1 & 6.87e-04& & 1 & 5.23e-04 &  & 6 & 3.68e-12& 3\\
& KRT6C & 2 & 4.41e-03& & 2 & 9.98e-04 &  & 2 & 0.00e+00& 3\\
& LSS & 4 & 6.96e-03& & 3 & 1.82e-03 &  & 3 & 7.61e-14& 3\\
& TLR2 & 9 & 1.35e-02& & 4 & 2.96e-03 &  & 7 & 3.68e-11& 5\\
& EIF2C4 & 6 & 1.13e-02& & 5 & 4.02e-03 &  & 8 & 9.56e-10& 3\\
& CPNE8 & 8 & 1.24e-02& & 6 & 4.81e-03 &  & 13 & 5.98e-09& 3\\
& CFP & 10 & 1.70e-02& & 7 & 5.42e-03 &  & 14 & 2.78e-08& 5\\
& TUT1 & 11 & 1.75e-02& & 8 & 5.65e-03 &  & 15 & 4.09e-08& 5\\
& SH3TC2 & 15 & 1.87e-02& & 9 & 7.39e-03 &  & 18 & 2.31e-07& 4\\
& USP20 & 14 & 1.86e-02& & 10 & 7.68e-03 &  & 19 & 3.61e-07& 3\\
\end{tabular}
}
\end{center}
\end{table}

\begin{table}[htbp]
\begin{center}
\caption{The 10 genes  with smallest $p$-value identified using \R~ \texttt{survdiff} in cancer data. For each cancer type, we show the top 10 genes, their rank, and $p$-value using the exact permutational test, the exact conditional test, and \R~ \texttt{survdiff}. The number of samples with a mutation in the gene and references supporting the association of mutations in the gene with survival are also reported.}
\label{table:Rsurvdiff}
{\scriptsize
\begin{tabular}{c | c c c c c c c c c c c}
& & \multicolumn{2}{c}{exact permutational} & & \multicolumn{2}{c}{exact conditional} & &  \multicolumn{2}{c}{R \texttt{survdiff}} & & \\
\cline{3-4} \cline{6-7} \cline{9-10} 
cancer type & gene & rank & $p$ & & rank & $p$ & & rank & $p$ & num. mut. & refs\\
\hline
COADREAD & VPS13B & 3 & 1.22e-02& & 1 & 3.57e-04 &  & 1 & 6.21e-12&  8\\
& ZNF804B & 8 & 3.97e-02& & 5 & 2.92e-03 &  & 2 & 1.07e-10&  5\\
& ERBB4 & 5 & 1.33e-02& & 2 & 5.89e-04 &  & 3 & 2.16e-10&  8\\
& AQPEP & 9 & 4.86e-02& & 6 & 5.38e-03 &  & 4 & 2.55e-08&  5\\
& DOCK9 & 2 & 9.18e-03& & 3 & 1.72e-03 &  & 5 & 4.25e-08&  6\\
& DOCK3 & 11 & 6.21e-02& & 8 & 1.11e-02 &  & 6 & 3.92e-06&  3\\
& IGF1R & 14 & 6.61e-02& & 9 & 1.26e-02 &  & 7 & 8.95e-06&  3\\
& ARAP1 & 50 & 1.75e-01& & 35 & 4.74e-02 &  & 8 & 1.03e-05&  3\\
& PRKCB & 51 & 1.75e-01& & 36 & 4.74e-02 &  & 9 & 1.03e-05&  3\\
& PTPRM & 1 & 5.20e-03& & 4 & 2.30e-03 &  & 10 & 1.74e-05&  8\\
GBM & HDAC9 & 52 & 8.72e-02& & 3 & 5.12e-03 &  & 1 & 6.13e-06&  3\\
& GRM7 & 61 & 9.95e-02& & 6 & 8.68e-03 &  & 2 & 7.33e-05&  3\\
& MYST3 & 62 & 1.00e-01& & 7 & 8.88e-03 &  & 3 & 7.86e-05&  3\\
& TBC1D2 & 33 & 5.51e-02& & 1 & 3.69e-03 &  & 4 & 7.96e-05&  5\\
& CASKIN2 & 46 & 7.92e-02& & 5 & 6.88e-03 &  & 5 & 1.59e-04&  4\\
& NFYC & 68 & 1.06e-01& & 10 & 1.16e-02 &  & 6 & 1.96e-04&  3\\
& ZNF618 & 70 & 1.07e-01& & 11 & 1.20e-02 &  & 7 & 2.26e-04&  3\\
& ATP13A2 & 53 & 8.77e-02& & 4 & 6.36e-03 &  & 8 & 3.71e-04&  6\\
& LIMCH1 & 73 & 1.10e-01& & 13 & 1.36e-02 &  & 9 & 3.76e-04&  3\\
& LMX1B & 79 & 1.15e-01& & 14 & 1.53e-02 &  & 10 & 5.16e-04&  3\\
KIRC & ARHGAP28 & 495 & 3.76e-01& & 15 & 7.25e-03 &  & 1 & 0.00e+00&  3\\
& EZH2 & 42 & 2.89e-02& & 4 & 8.07e-04 &  & 2 & 0.00e+00&  3 & \\
& PRKG2 & 508 & 3.91e-01& & 34 & 1.45e-02 &  & 3 & 3.33e-16&  3\\
& FBXO43 & 26 & 2.19e-02& & 2 & 2.46e-04 &  & 4 & 4.44e-16&  5\\
& ATP10D & 2 & 1.91e-03& & 3 & 6.15e-04 &  & 5 & 3.55e-12&  3\\
& RASEF & 46 & 3.21e-02& & 7 & 3.19e-03 &  & 6 & 5.48e-11&  3\\
& CDCA2 & 1 & 8.54e-04& & 1 & 1.75e-04 &  & 7 & 6.15e-10&  6\\
& PCDH11X & 47 & 3.26e-02& & 9 & 4.16e-03 &  & 8 & 1.05e-09&  3\\
& ZFAT & 48 & 3.43e-02& & 10 & 4.57e-03 &  & 9 & 2.65e-09&  3\\
& TDRD7 & 3 & 3.66e-03& & 5 & 1.50e-03 &  & 10 & 5.00e-09&  3\\
LUSC & CHTF18 & 1014 & 2.13e-01& & 3 & 5.24e-04 &  & 1 & 0.00e+00&  4\\
& RLTPR & 656 & 1.21e-01& & 8 & 1.94e-03 &  & 2 & 7.70e-14&  3\\
& ZNF557 & 75 & 7.07e-03& & 2 & 4.09e-04 &  & 3 & 1.96e-13&  3\\
& CD33 & 681 & 1.27e-01& & 13 & 3.71e-03 &  & 4 & 3.90e-10&  3\\
& SPATA22 & 703 & 1.34e-01& & 20 & 5.15e-03 &  & 5 & 9.90e-09&  3\\
& SIN3B & 111 & 1.23e-02& & 6 & 1.55e-03 &  & 6 & 1.00e-08&  3\\
& ANKRD11 & 487 & 8.70e-02& & 7 & 1.70e-03 &  & 7 & 2.41e-08&  5\\
& NARG2 & 120 & 1.35e-02& & 9 & 2.15e-03 &  & 8 & 8.62e-08&  3\\
& ZC3H12C & 1088 & 2.27e-01& & 28 & 7.46e-03 &  & 9 & 2.85e-07&  4\\
& VN1R4 & 712 & 1.36e-01& & 29 & 7.78e-03 &  & 10 & 3.83e-07&  3\\
OV & KIAA0913 & 82 & 1.21e-01& & 2 & 8.65e-03 &  & 1 & 6.22e-05&  4\\
& PABPC3 & 39 & 4.42e-02& & 1 & 5.95e-03 &  & 2 & 1.04e-04&  4\\
& NIPBL & 110 & 1.85e-01& & 3 & 1.23e-02 &  & 3 & 2.56e-04&  5\\
& PCDHA9 & 46 & 5.67e-02& & 4 & 1.23e-02 &  & 4 & 8.84e-04&  4\\
& RBM44 & 91 & 1.49e-01& & 6 & 1.85e-02 &  & 5 & 9.99e-04&  4\\
& EFCAB6 & 124 & 2.07e-01& & 9 & 2.07e-02 &  & 6 & 1.56e-03&  5\\
& TNRC18B & 55 & 7.05e-02& & 7 & 1.95e-02 &  & 7 & 2.96e-03&  4\\
& FARSA & 103 & 1.70e-01& & 15 & 2.89e-02 &  & 8 & 3.84e-03&  4\\
& PLEKHG1 & 104 & 1.70e-01& & 16 & 2.94e-02 &  & 9 & 4.02e-03&  4\\
& PTCH1 & 65 & 8.80e-02& & 8 & 2.05e-02 &  & 10 & 5.46e-03&  6\\
UCEC & KLHL29 & 51 & 4.08e-02& & 24 & 1.39e-02 &  & 1 & 0.00e+00&  3\\
& KRT6C & 2 & 4.41e-03& & 2 & 9.98e-04 &  & 2 & 0.00e+00&  3\\
& LSS & 4 & 6.96e-03& & 3 & 1.82e-03 &  & 3 & 7.61e-14&  3\\
& MYC & 86 & 7.56e-02& & 31 & 1.87e-02 &  & 4 & 6.05e-13&  4\\
& SH3GL1 & 53 & 4.86e-02& & 32 & 1.87e-02 &  & 5 & 6.05e-13&  3\\
& CTGF & 1 & 6.87e-04& & 1 & 5.23e-04 &  & 6 & 3.68e-12&  3\\
& TLR2 & 9 & 1.35e-02& & 4 & 2.96e-03 &  & 7 & 3.68e-11&  5\\
& EIF2C4 & 6 & 1.13e-02& & 5 & 4.02e-03 &  & 8 & 9.56e-10&  3\\
& LOC342346 & 59 & 5.83e-02& & 47 & 2.79e-02 &  & 9 & 5.35e-09&  3\\
& RAB40AL & 60 & 5.83e-02& & 48 & 2.79e-02 &  & 10 & 5.35e-09&  3\\
\end{tabular}
}
\end{center}
\end{table}

\subsubsection*{Published Cancer Studies}

We analyzed differences between survival distributions reported in two published genomic studies~\cite{Huang:2009fk,Yan:2009uq}. We considered only cases where the smallest population included at most $30\%$ of all samples. We compared the exact permutational $p$-value with the $p$-value reported in the publications obtained using asymptotic approximations. Since the data for these studies is not publicly available, we inferred the data necessary to perform the log-rank test using the figures in the publications. In particular, since the exact time of events (censored or not) is not used by the log-rank test, we only inferred the order of events into the two populations and the censoring information. We then used \R~\texttt{survdiff} to obtain the $p$-value from the asymptotic approximation, and compared it with the $p$-value reported in the paper to validate the information we inferred from the figures.

In particular, for~\cite{Huang:2009fk} we considered:
\begin{itemize}
\item Figure~2L: population sizes: 2 and 14. Reported $p=0.012$; \R~\texttt{survdiff} $p=0.012$; exact permutational $p= 0.017$;
\item Figure~2I: population sizes: 8 and 48. Reported $p<10^{-4}$; \R~\texttt{survdiff} $p=6\times10^{-6}$; exact permutational $p= 2\times 10^{-3}$;
\item Figure~2H: population sizes: 9 and 21. Reported $p=3.2\times10^{-3}$; \R~\texttt{survdiff} $p=3.2\times10^{-3}$; exact permutational $p= 1.1\times 10^{-2}$.
\end{itemize}

For~\cite{Yan:2009uq} we considered:
\begin{itemize}
\item Figure~3A: population sizes: 14 and 115. Reported $p=2\times 10^{-3}$; \R~\texttt{survdiff} $p=2.3\times 10^{-3}$; exact permutational $p= 4.8\times 10^{-4}$;
\item Figure~3B: population sizes: 14 and 38. Reported $p< 10^{-3}$; \R~\texttt{survdiff} $p=6\times 10^{-6}$; exact permutational $p= 5\times 10^{-4}$.
\end{itemize}

\section*{Comparison of Exact Permutational Test and Cox Proportional-Hazard Model on Synthetic Data}
Three asymptotically equivalent statistical tests are commonly used to assess  significance using the Cox Proportional-Hazard model: the score test, the Wald test, and the likelihood ratio test. All the three tests are based on an asymptotic approximation for the distribution of the test statistic.

We applied the three tests (score test, Wald test, and likelihood ratio test) to assess statistical significance under the Cox Proportional-Hazard Model, on randomly generated survival and mutation data, where no mutation is associated with survival. The three tests are based on asymptotic approximations for the distribution of the test statistic. We focused on the case of unbalanced populations. In particular, we considered $n=100$ total patients, and $n_1=5$ patients in the small population. We used the \R~\texttt{coxph} package to compute the $p$-values. Fig.~\ref{fig:Cox_asymptotic} shows that for the score test and the Wald test the asymptotic approximation is inaccurate, while the asymptotic approximation is pretty accurate for the likelihood ratio test.

We then compared the accuracy of the $p$-values obtained with the exact permutational test and the $p$-values from the Cox likelihood ratio test. In particular, we use \emph{synthetic data}, generated using the same distribution for the survival time and for the censoring time for all patients, and using different procedures to generate the mutation data. More specifically, we compared the empirical $p$-value (obtained by generating the data a number of times using the same parameters for the distribution) with the $p$-values from the exact permutational test and the $p$-values from the Cox likelihood ratio test using synthetic data. We generate the mutation data using two related but different procedures.  In the first procedure, we mutate a gene $g$ in exactly a fraction $f$ of all patients.  In the second procedure, we mutated a gene $g$ in each patient independently with probability $f$. The second procedure models the fact that mutations in a gene $g$ are found in each patient independently with a certain 
probability (that depends on the background mutation rate, the length of the gene, etc.).  Thus, when repeating a study on a cohort of patients of the same size, only the expected number of patients in which $g$ is mutated is the same, and the observed number may vary. In both cases, the survival information is generated from the same distribution for all patients. In particular, the survival time comes from the exponential distribution with expectation equal to 30, and censoring variable from an exponential distribution resulting in $30\%$ of censoring.  In Fig.~\ref{fig:cmp_Cox_exact} we compare the $p$-values computed from the exact permutational test and the Cox likelihood ratio test with the empirical $p$-values (obtained repeating the experiment 10000 times) for the first distribution, while in Fig.~\ref{fig:cmp_Cox_expect} we compare the $p$-values computed from the exact permutational test and the Cox likelihood ratio test with the empirical $p$-values (obtained repeating the experiment 10000 times) 
for the second distribution. In both cases, the $p$-values (restricted to $p$-values $\le 0.01$) from the exact permutational distribution have higher R coefficients than the $p$-values from the Cox likelihood ratio test when compared to the empirical $p$-values (considering the $-log_{10}$ $p$-values in order to compute the R coefficient).

\begin{figure}[htbp]
\centering
                \subfloat[][]{
             \includegraphics[width=0.5\textwidth]{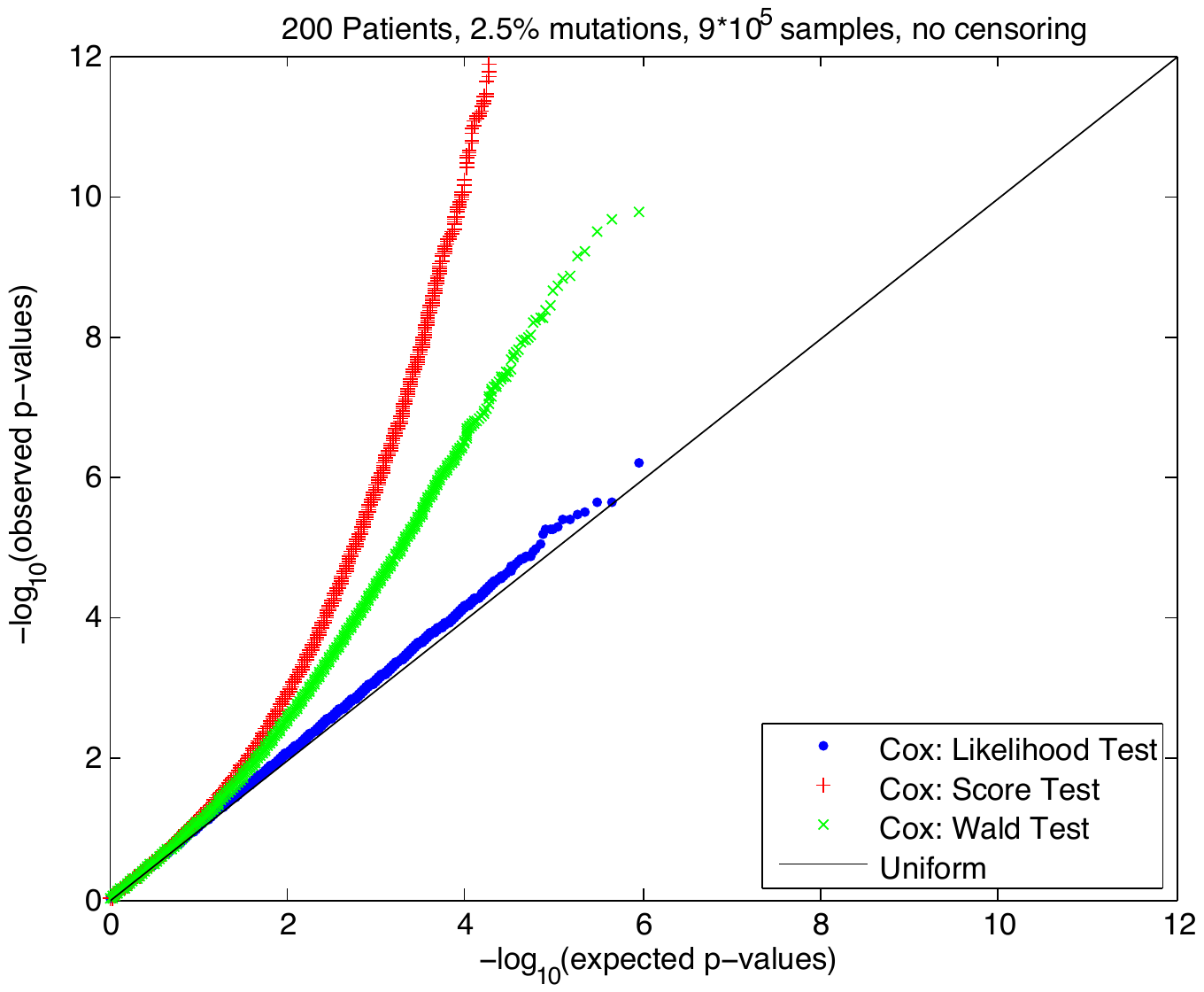}
              \label{fig:Cox_asymptotic}
                }
                \\
		 \subfloat[][]{
		 \includegraphics[width=0.4375\textwidth]{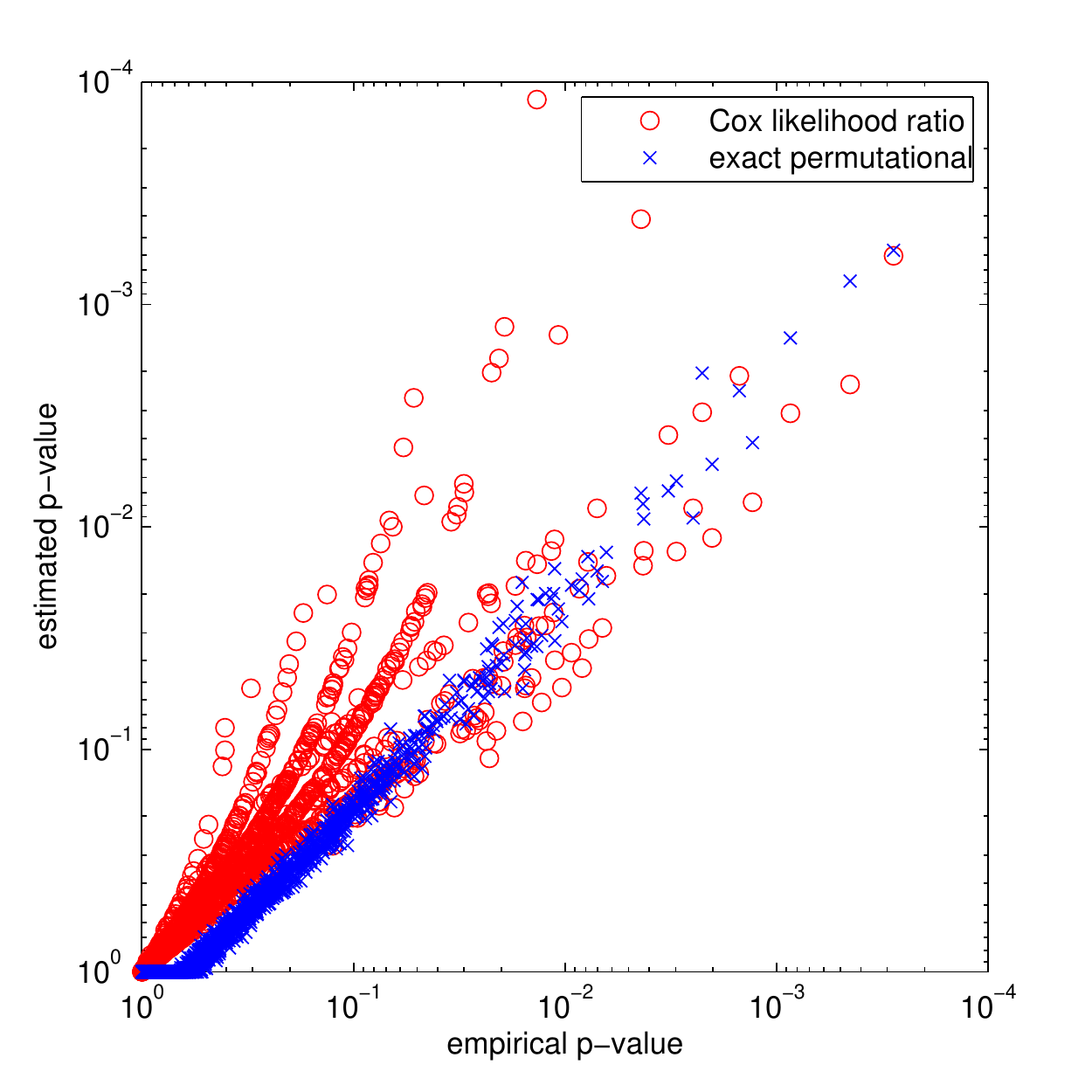}
		\label{fig:cmp_Cox_exact} 
                }
                \quad
                \subfloat[][]{
		 \includegraphics[width=0.4375\textwidth]{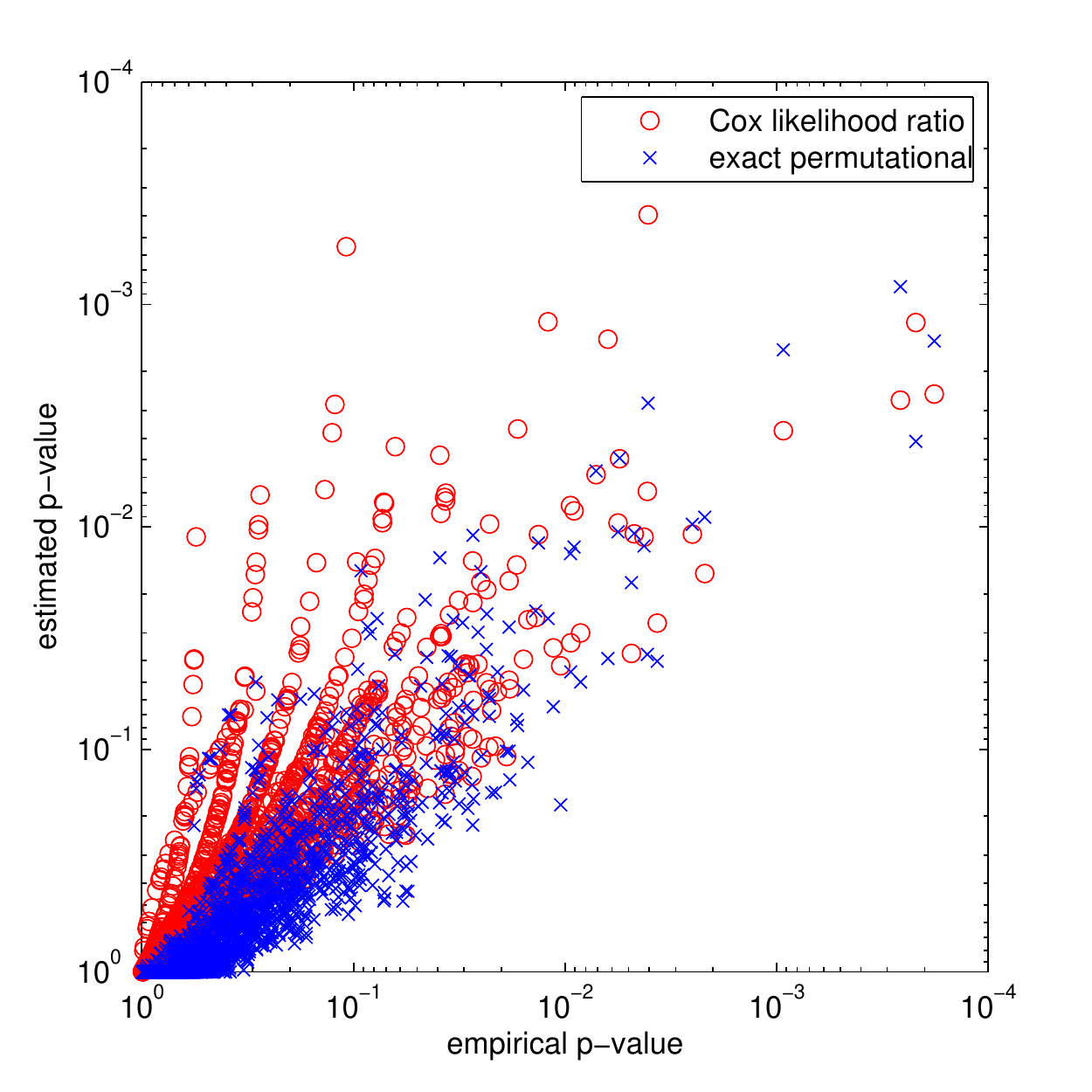}
		\label{fig:cmp_Cox_expect} 
                }
\caption{
Comparison of the $p$-values from asymptotic approximations for the Cox Proportional-Hazard model and the uniform distribution, and comparison of the $p$-values from exact permutational tests and the Cox likelihood ratio test with the empirical $p$-values for two different null distributions. (a) Distribution of $p$-values obtained using the asymptotic approximation for the Cox Proportional-Hazard model and the distribution of $p$-values for the uniform distribution. Generated considering $9\times10^5$ instances with $n=100$ total samples, $n_1=5$ samples in the small population and same survival distribution for all patients (no censoring). 
(b) Comparison of Cox likelihood ratio $p$-values, exact permutational $p$-values, and empirical $p$-values for $n=100, n_1=5\%n$, and $30\%$ censoring. Each point represents an instance of survival data. ({c}) Comparison of Cox likelihood ratio $p$-values, exact permutational $p$-values, and empirical $p$-values for $n=100$, expectation($n_1$)=$5\%n$, and $30\%$ censoring. Each point represents an instance of mutations and survival data. The R coefficients comparing the $-log_{10}$ exact $p$-values to the $-log_{10}$ empirical $p$-values are the following: in Fig.~(b), permutational $= 0.96$, Cox likelihood ratio $= 0.70$; in Fig~(c), permutational $= 0.72$, Cox likelihood ratio $= 0.46$. For both distributions the difference between R coefficients is significant $(p<10^{-3})$.}
\label{fig:cmp_empirical}
\end{figure}

The Cox proportional-hazards model is often used to correct for other variables that may be correlated to survival, like age, gender, or tumor stage. While the multivariate case is not the focus of this work, 
in this scenario a common rule of thumb~\cite{Concato:1995fk,Peduzzi:1995uq,Peduzzi:1996kx} states that Cox models should be used with a minimum of 10 outcome events per predictor variable to obtain reliable results. This limits its applicability to moderately frequent events even in large genomic studies.  For example, 
if we include three predictor variables in the model in addition to the mutation status of a gene, then
only two of our seven cancer datasets (LUSC and UCEC) have more than 3 genes (7 and 8, respectively) that have the minimum recommended number of mutations. 
Extensions of the exact test presented here might prove useful in such settings of a small number of events per predictor variable.  In particular, a stratified log-rank test using an exact distribution is a promising alternative.

\newpage
\bibliographystyle{plain}
\bibliography{survival}  

\end{document}